\documentclass[11pt,a4paper]{article}
\pdfoutput=1
\usepackage{jheppub}
\usepackage{macros}
\usepackage{rotating}

\preprint{ }

\title{Branes and Representations of DAHA $C^\vee C_1$:\\ affine braid group action on category}

\author{Junkang Huang,}
\author{Satoshi Nawata,}
\author{Yutai Zhang,}
\author{and Shutong Zhuang}

\affiliation{Department of Physics and Center for Field Theory and Particle Physics, Fudan University, \\
20005, Songhu Road, 200438 Shanghai, China}

\emailAdd{snawata@gmail.com}

\abstract{
We study the representation theory of the spherical double affine Hecke algebra (DAHA) of $C^\vee C_1$, using brane quantization. By showing a one-to-one correspondence between Lagrangian $A$-branes with compact support and finite-dimensional representations of the spherical DAHA, we provide evidence of derived equivalence between the $A$-brane category of $\SL(2,\bC)$-character variety of a four-punctured sphere and the representation category of DAHA of $C^\vee C_1$. The $D_4$ root system plays an essential role in understanding both the geometry and representation theory.  In particular, this $A$-model approach reveals the action of an affine braid group on the category. As a by-product, our geometric investigation offers detailed information about the low-energy effective dynamics of the SU(2) $N_f=4$ Seiberg-Witten theory.
}

\begin{document}
\maketitle

% \epigraph{I felt that the theory of elliptic surfaces was not something I invented myself, but rather something that was buried within the great tree of mathematics. What I did was merely to uncover it with paper and pencil.}{Kunihiko Kodaira}

\newpage

\section{Introduction and Summary}

Brane quantization, as introduced by Gukov and Witten~\cite{Gukov:2008ve}, is a framework that applies ideas from $A$-model topological string theory to the quantization of symplectic manifolds. The core idea is to approach the quantization process via the topological $A$-model on a ``complexification'' $\X$ of the symplectic manifold $M$. The complexified target space $\X$ is usually an affine \HK manifold with a holomorphic symplectic form $\Omega$, whose real part $\Re \Omega$ restricts to the symplectic form on $M$, and imaginary part $\Im \Omega$ restricts to zero on $M$. Considering the $A$-model on $\X$ with a symplectic form $\omega_\X:=\Im \Omega$, the category of $A$-branes $\ABrane(\X,\omega_\X)$ provides a unified framework to quantize not only $M$ but all the $A$-branes in the target space $\X$. 

A key ingredient of this approach is the distinguished $A$-brane called the \emph{canonical coisotropic brane} \cite{Kapustin:2001ij,Kapustin:2006pk} whose support is the entire space $\X$. Its endomorphism algebra gives rise to the deformation quantization of the coordinate ring of $\X$ with respect to the holomorphic symplectic form $\Omega$:
\be 
\End(\Bcc)=\OO^{q}(\X)~.
\ee
The essence of the brane quantization lies in the fact that $M$ is a Lagrangian submanifold in the symplectic manifold $(\X, \omega_\X)$. This allows the original symplectic manifold itself to define an $A$-brane, denoted $\brane_M$, in $(\X, \omega_\X)$. In particular, as shown in \cite{Gukov:2008ve,Gukov:2010sw,Gaiotto:2021kma}, the morphism space $\Hom(\brane_M,\Bcc)$ can be identified with the geometric quantization of $M$, assuming that $M$ is a \K manifold. Thus, brane quantization serves as a bridge connecting deformation quantization with geometric quantization.

Brane quantization, however, goes far beyond this basic construction. A morphism $\Hom(\brane,\brane^\prime)$ between $A$-branes can be understood as $(\brane^\prime,\brane)$-open string in the $A$-model. Joining $(\Bcc,\Bcc)$-string to $(\Bcc,\brane)$-string for any other $A$-brane $\brane$ naturally defines an action of the algebra $\End(\Bcc)=\OO^{q}(\X)$. In this way, brane quantization naturally proposes a functor\footnote{We consider left $\OO^q(\X)$-modules so that a functor is expressed as $\Hom(-,\Bcc)$ rather than $\Hom(\Bcc,-)$, following the standard mathematical convention. For the same reason,  the space of $(\brane^\prime,\brane)$-open strings is represented by a morphism $\Hom(\brane,\brane^\prime)$.}
\be \label{eq:functor0}
\RHom(-,\Bcc): D^b\ABrane(\MS,\omega_\MS) \to D^b \Rep(\OO^q(\MS))
\ee 
which is conjectured to establish a derived equivalence between the category of $A$-branes and the derived category of $\OO^q(\MS)$-modules. This equivalence is understood in the derived sense, meaning that it incorporates grading shifts and treats objects and morphisms up to quasi-isomorphism.
Consequently, the functor provides a geometric framework for understanding the category of $\OO^q(\X)$-modules. Notably, the role of $M$ in this framework is not special;  it represents just one of many possible $A$-branes, each of which naturally corresponds to an $\OO^q(\X)$-module. The framework of brane quantization is mathematically formulated as generalized Riemann-Hilbert correspondence in \cite{Kontsevich:2024esg}. The goal of this paper is thus to present an explicit instance of this correspondence, giving a concrete manifestation of this derived equivalence in the general framework.

In this paper, we consider the target space $\X$ of 2d non-linear sigma-model to be the moduli space of flat $\SL(2,\bC)$-connections on a four-punctured sphere $C_{0,4}$ (a.k.a. the $\SL(2,\bC)$-character variety):
\be 
\X = \MF(C_{0,4}, \SL(2,\bC))~.
\ee 
This space is an affine \HK variety with a distinguished holomorphic symplectic form, which naturally fits into the framework of brane quantization.  
As proven in \cite{oblomkov2004cubic}, the algebra $\OO^q(\X)$ is the spherical subalgebra of the double affine Hecke algebra (DAHA) of type $C^\vee C_1$, which we denote
\be 
\OO^q(\X)\cong \SH~.
\ee 
The main objective of this paper is to apply the brane quantization framework to this space and provide compelling evidence for the equivalence \eqref{eq:functor0} between the $A$-brane category on $(\X, \omega_\X)$ and the representation category of the spherical DAHA of type $C^\vee C_1$. Note that in our example, the target space $\X$ is a real four-dimensional symplectic manifold. Due to dimensional constraints, no other coisotropic branes can exist besides the canonical coisotropic brane. Consequently, all other $A$-branes in this setting are necessarily Lagrangian and thus fall into the framework of the Fukaya category.

DAHA was introduced by Cherednik \cite{cherednik1992double,cherednik1995double,Cherednik-book} as an underlying algebra for $q$-difference operators that govern multi-variable orthogonal special functions such as Macdonald polynomials. Building on this, a series of works \cite{Noumi1995,vD96,sahi1999nonsymmetric,sahi2000koornwinder,stokman2000koornwinder} developed DAHA of type $C^\vee C_1$, which is the main focus of this paper, as the algebra governing the Askey-Wilson polynomials \cite{askey1985some}. Associated to the $\SL(2,\bC)$-holonomies around the punctures, DAHA of type $C^\vee C_1$ depends on four deformation parameters $\boldsymbol{t}=(t_1,t_2,t_3,t_4)$. In \S\ref{sec:DAHA-CC1}, we review this algebra and its representations, especially the polynomial representation involving the Askey-Wilson polynomials.

Furthermore, pioneering works \cite{oblomkov2004Calogero,oblomkov2004cubic,vasserot2005induced,bezrukavnikov2005equivariant,varagnolo2010double} explored DAHA in geometric contexts such as the deformation quantization of coordinate rings of character varieties and the equivariant $K$-rings of affine Grassmannians. These perspectives were further examined from a physics standpoint in \cite{Gukov:2022gei}, where brane quantization was employed to study the representation theory of DAHA of type $A_1$, paving the way for this investigation. A key insight is that the algebra of line operators \cite{Kapustin:2007wm,Okuda:2019emk,Cirafici:2020qlf} in 4d $\cN=2^*$ theory on the $\Omega$-background provides the deformation quantization of the coordinate ring of the Coulomb branch \cite{Gaiotto:2010be} and is related to the spherical DAHA of the reduced affine root system associated with the gauge algebra.\footnote{Since the spectrum of line operators depends on the global structure of the gauge group $G$ \cite{Aharony:2013hda,Tachikawa:2013hya}, the algebra of line operators is generally not strictly isomorphic to the spherical DAHA. Instead, for a simple Lie group $G$, it corresponds to a quotient of the spherical DAHA by a maximal isotropic subgroup of $Z(\tilde G) \oplus Z(\tilde G)$, where $Z(\tilde G)$ denotes the center of the universal covering of $G$.
}
Character varieties and affine Grassmannians naturally arise as moduli spaces associated with line operators in 4d $\cN=2^*$ theories. Our study is motivated by the physics perspective that the study of algebras of line operators in more general 4d $\cN=2$ theories provides a vast generalization of DAHA. While the 4d $\cN=2^*$ theory is associated to a once-punctured torus via the class $\cS$ construction \cite{Gaiotto:2009we,Gaiotto:2009hg}, this paper takes the first non-trivial step in extending this framework to the case of a four-punctured sphere. As demonstrated in \cite{Gukov:2022gei,Shan:2023xtw,Shan:2024yas}, the combination of physical insights with geometric approaches in the study of the DAHA representation theory reveals intricate connections among algebra, geometry, and physics.

Our target space $\X$ is described as an affine cubic surface—a simple yet remarkably rich geometric space that has been studied since the 19th century. True to the nature of any interesting geometric object, $\X$ exhibits multiple facets. For instance, it can be interpreted as the moduli space of  
 $\SU(2)$ parabolic Higgs bundles on $C_{0,4}$ (a.k.a.
the Hitchin moduli space). From a physical perspective, $\X$ corresponds to the Coulomb branch of a 4d $\cN=2$ supersymmetric $\SU(2)$ gauge theory (SQCD) with four fundamental hypermultiplets ($N_f = 4$) \cite{Seiberg:1994aj} on $S^1 \times \bR^3$. As we demonstrate in \S\ref{sec:prelude}, the 2d sigma-model in $\X$ can be related to the 4d $\cN=2$ SQCD by compactification on $T^2$.
By examining these perspectives, we conduct a detailed investigation of the geometry of $\X$ in \S\ref{sec:geometry}. This in-depth study not only establishes a solid foundation for analyzing $A$-branes in the target space $(\X, \omega_{\X})$ but also provides detailed insight into the low-energy dynamics of the 4d $\cN=2$ theory.

Remarkably, the $D_4$ root system plays a central role in controlling the geometry of the target space $\X$. First, note that the Lie algebra  $\frakso(8)\cong D_4$ is the flavor symmetry algebra of the 4d SU(2) SQCD with $N_f=4$.
One key feature is that the second integral homology group is isomorphic to the affine $D_4$ root lattice, and the action of the affine Weyl group is realized through the Picard-Lefschetz monodromy transformation:
\be\label{dtWD4action}
\dt{W}(D_4) \ \rotatebox[origin=c]{-90}{$\circlearrowright$}\ H_2(\X,\bZ) \cong \dt\sfQ(D_4)~.
\ee
In addition, when the deformation parameters $\boldsymbol{t}$ are specialized, du Val singularities emerge in the geometry of the target space. The specific specializations of $\boldsymbol{t}$ and the resulting types of du Val singularities are also determined by the $D_4$ root system.
Moreover, the target space $\X$ admits an elliptic fibration as the Hitchin fibration. The ramification parameters of the Higgs field and the Kodaira types of singular fibers are similarly controlled by the structure of the $D_4$ root system. As the four deformation parameters $\boldsymbol{t}$ vary, the geometry of $\X$ undergoes drastic changes. However, as we see in \S\ref{sec:geometry}, these geometric changes can be neatly encapsulated using the root system $D_4$.

Having laid the geometric foundation, we turn to the study of $A$-branes in the target space $\X$ and the corresponding $\SH$-modules from the viewpoint of brane quantization in \S\ref{sec:brane-rep}. First, we show the following result for non-compact $A$-branes:
\begin{claim}   \label{claim:1}
Under the functor \eqref{eq:functor0}, the $(A,B,A)$-branes supported on 24 lines in the affine cubic surface $\X$ correspond to the polynomial representation of $\SH$ and its images under the $D_4$ Weyl group and the cyclic permutation group corresponding to the triality of $\frakso(8)$.
\end{claim}

We then proceed to explicitly identify a compact Lagrangian $A$-brane in $\X$ for each finite-dimensional irreducible representation of $\SH$. In particular, we match the corresponding objects by analyzing their parameter spaces, dimensions, and shortening conditions on both sides. In addition,  we examine the spaces of derived morphisms between objects in the two categories.
This detailed investigation within the framework of brane quantization provides solid evidence for the following statement:

\begin{claim}  \label{claim:2}
For $\X=\MF(C_{0,4}, \SL(2, \bC))$, the functor \eqref{eq:functor0} restricts to a derived equivalence of the full subcategory of compact Lagrangian $A$-branes of $\X$ and the category of finite-dimensional \SH-modules.
\end{claim}

One of the most profound lessons we have learned from homological mirror symmetry \cite{KontsevichICM} and the geometric Langlands program \cite{Kapustin:2006pk,Gukov:2006jk,Frenkel:2007tx} is the importance of treating the entire collection of boundary conditions as a category. This perspective reveals hidden structures in the category of boundary conditions. Although a precise definition of the $A$-brane category is still lacking, its physical meaning, behavior, and properties are fairly well understood and have been extensively studied in the context of 2d non-linear sigma-models.

The results outlined above provide a geometric perspective on the representation theory of $\SH$. However, brane quantization reveals even hidden structures at the categorical level. To see such hidden structures, we consider the quadruple of parameters $(\talpha,\tbeta,\tgamma,\teta)\in T_{D_4}\times \frakt_{D_4}\times \frakt_{D_4}\times T_{D_4}^\vee$, which parametrizes a family of 2d $\cN=(4,4)$ sigma-models. Among these parameters, the first three parameters $(\talpha,\tbeta,\tgamma)$ correspond to tame ramification parameters of a Higgs bundle at punctures, while $\teta$ originates from the $B$-field in the 2d sigma model. 

Depending on the choice of the symplectic form $\omega_\X = \Im \Omega$, the parameter space of the quadruple splits into two subspaces: one acts as the moduli space of complex structures whereas the other serves as the complexified \K moduli space. Since the $A$-model depends solely on the complexified \K moduli, variations along a loop in the complex structure moduli space (avoiding singularities) induce non-trivial transformations of the $A$-brane category \cite{Gukov:2006jk}.

In our setup, for a generic choice of the complexified \K parameter $v$, the group of transformations arising from loops in the complex structure moduli space is an affine braid group $\dt\Br_v$, whose explicit form will be presented in \S\ref{sec:category}. Therefore, we obtain the following claim for the categorical structure:

\begin{claim} \label{claim:3}
Given a complexified \K parameter $v$ for the 2d $A$-model on $\X$, the affine braid group  $\dt\Br_v$ acts on the category $\ABrane(\X,\omega_\X)$, and therefore on the representation category $\Rep(\SH)$ of the spherical DAHA, as the group of auto-equivalences. 
\end{claim}

In summary, the approach of brane quantization to the representation theory of $\SH$ uncovers the rich geometric and categorical structures that underlie these algebraic objects. This offers a more unified understanding of their properties and interrelations. With this motivation in mind, we now begin our exploration to ``see'' and ``understand'' these structures and relationships.

\bigskip

The structure of the paper is organized as follows. In \S\ref{sec:DAHA-CC1}, we provide a comprehensive overview of the double affine Hecke algebra (DAHA) of type $C^\vee C_1$, beginning with its defining relations and symmetry properties in \S\ref{sec:DAHA-CC1-subsec}. In \S\ref{sec:SH}, we introduce our main algebra, the spherical DAHA $\SH$ of type $C^\vee C_1$, emphasizing its symmetry and the connection to the $D_4$ root system. This is followed by a discussion of the polynomial representation of $\SH$ and the Askey-Wilson polynomials in \S\ref{sec:polyrep}. The section \S\ref{sec:DAHA-CC1} concludes with a classification of finite-dimensional representations derived from the polynomial representation.

Our motivation arises from the application of brane quantization to the target space $\X$ for the representation theory of $\SH$. To this end, we provide a detailed investigation of the geometry of the target space $\X$ in \S\ref{sec:geometry}. For this purpose, the identification of $\X \cong \MH(C_{0,4}, \SU(2))$ as the $\SU(2)$ Hitchin moduli space on the four-punctured sphere proves particularly useful. In \S\ref{sec:prelude}, we begin by revisiting the Seiberg-Witten theory for $\SU(2)$ with $N_f = 4$, providing a physical background and its connection to the Hitchin system. In \S\ref{sec:brane-quantization}, we present a detailed explanation of the framework of brane quantization, linking the geometry of $\X$ to the representation theory of $\SH$.
To prepare for the match between $A$-branes and representations of $\SH$, we first examine the Hitchin fibration in the special case where $\tbeta = 0 = \tgamma$ in \S\ref{sec:wallcrossing}, showing the connection to the affine $D_4$ root system. In particular, we explicitly construct the action \eqref{dtWD4action} of the affine Weyl group $\dt W(D_4)$ on the homology cycles through wall-crossing. Next, in \S\ref{sec:classification}, we classify the configurations of Kodaira singular fibers in the Hitchin fibration of $\MH(C_{0,4}, \SU(2))$ when the ramification parameter $\tgamma$ is turned on. 
In \S\ref{sec:volume}, we identify the generators of the second homology group and compute their volumes with respect to $\Omega$ for these configurations. The connection to the $D_4$ root system is particularly powerful throughout this geometric study. The section concludes with a discussion of the symmetry actions on homology cycles in \S\ref{sec:symmetry}.

In \S\ref{sec:brane-rep}, we explore the interplay between branes and representations within the framework of brane quantization. Specifically, \S\ref{sec:category} examines symmetry actions on the categories $\ABrane(\X,\omega_\X)$ and $\Rep(\SH)$, with a focus on the affine braid group action on the category (Claim \ref{claim:3}).  In \S\ref{sec:brane-poly}, we investigate the correspondence between non-compact $(A, B, A)$-branes and polynomial representations (Claim \ref{claim:1}). We then proceed with a detailed analysis of the correspondence between compact $A$-branes and finite-dimensional $\SH$-modules (Claim \ref{claim:2}).  We begin by considering the case where $\hbar$ is real, establishing the matching between objects and morphisms in the two categories. We first examine this matching in the context of generic fibers of the Hitchin fibration in \S\ref{sec:genericfiber}. 
To show solid evidence for the correspondence, we analyze their parameter spaces, dimensions, and shortening conditions on both sides, and further compare morphism structures. We start with the most degenerate singular fiber, the global nilpotent cone, and progress to less degenerate singular fibers. The section concludes with a discussion of the equivalence for generic values of $\hbar$ in \S\ref{sec:generic_hbar}.

\section{Double Affine Hecke Algebra of type  \texorpdfstring{$C^\vee C_1$}{C*C1}}\label{sec:DAHA-CC1}

In the late 1980s, Macdonald introduced multivariable $q$-symmetric polynomials \cite{macdonald1987orthogonal}, which form a $\bC_{q,t}$-basis for $\bC_{q,t}[X]^{\mathfrak{S}_n}$ and constitute an orthogonal system of polynomials with respect to the inner product defined in $\bC_{q,t}[X]^{\mathfrak{S}_n}$. The properties of these Macdonald symmetric polynomials $P_\lambda(X; q, t)$ can be explored using both difference operators and the associated inner product.

Cherednik reconstructed these $q$-difference operators based on the representation theory of the affine Hecke algebra \cite{cherednik1992double,cherednik1995double,cherednik1995macdonald,cherednik1995nonsymmetric}. In fact, Macdonald's theory of orthogonal polynomials can be extended to any admissible pair $(\Delta, \Delta')$ of affine root systems \cite{macdonald2003affine,stokman2020macdonald}. Given such a pair, one can construct a pair of affine Hecke algebras, referred to as the double affine Hecke algebra (DAHA) \cite{Cherednik-book}, which governs the associated Macdonald polynomials.

Through the works of \cite{Noumi1995,vD96,sahi1999nonsymmetric,sahi2000koornwinder,stokman2000koornwinder}, Macdonald-Cherednik theory has been further extended to non-reduced affine root systems, particularly to the type $(C_n^{\vee}, C_n)$, giving rise to the Koornwinder polynomials of this type \cite{koornwinder1992askey}.

The double affine Hecke algebra (DAHA) of type $(C_1^{\vee}, C_1)$ (in short $C^\vee C_1$), denoted by $\HH$, is the main focus of this paper and governs the Askey-Wilson polynomials \cite{askey1985some}, which are the rank-one Koornwinder polynomials. In addition to the deformation parameter $q$, the DAHA involves a set of $t$-parameters, whose number corresponds to the number of orbits in the associated affine Weyl group. For the $C^{\vee}C_1$ case, there are four such parameters, $(t_1, t_2, t_3, t_4)$, which appear both in the algebra and in the Askey-Wilson polynomials.

\subsection{DAHA of type \texorpdfstring{$C^\vee C_1$}{C*C1}}\label{sec:DAHA-CC1-subsec}

Let us now review some necessary details of DAHA of type $C^\vee C_1$.  
 For an admissible pair $C^\vee C_1$, DAHA is closely related to the fundamental group of a four-punctured sphere, which is generated by  $T_0, T_0^{\vee}, T_1,T_1^{\vee}$ as illustrated in Figure \ref{fig:daha generators}
\begin{equation}
\pi_1\left(S^2 \backslash \{p_1,p_2,p_3,p_4\}\right) =\langle T_0, T_0^{\vee}, T_1,T_1^{\vee} \mid T_0^\vee T_0 T_1 T_1^\vee=1\rangle.
\end{equation}
A deformation of the group algebra of this fundamental group is achieved by modifying the relation to
\begin{equation}
    T_0^\vee T_0 T_1 T_1^\vee=q^{-\frac{1}{2}}.
\end{equation}
Then, DAHA of type $C^\vee C_1$, denoted as $\HH$, can be constructed by the quotient of the algebra by four Hecke relations\footnote{We make a slight adjustment to the standard notation so that the parameters $t_i$ agree exactly with the geometric parameters, which will be important for later use. For the standard notation, we refer the reader to Appendix \ref{app:notation}.}, namely imposing the Hecke relation for each generator
\cite{sahi1999nonsymmetric}
\bea\label{CC1}
    (T_0+iq^{\frac{1}{2}}t_1)(T_0+iq^{-\frac{1}{2}}t_1^{-1})&=0~, \\
(T_0^\vee+it_2)(T_0^\vee+it_2^{-1})&=0~,\\
(T_1+it_3)(T_1+it_3^{-1})&=0~,\\
(T_1^\vee+it_4)(T_1^\vee+it_4^{-1})&=0~.
\eea
Thus, the algebra $\HH$ is parameterized by four parameters from the Hecke relations, $\boldsymbol{t} = (t_1, t_2, t_3, t_4) \in (\bC^\times)^4$, along with the deformation parameter $q\in \bC^\times$.

\begin{figure}
    \centering
    \includegraphics{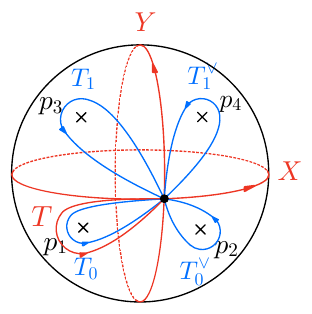} \qquad\qquad\qquad \includegraphics{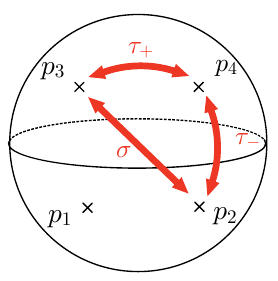}
    \caption{(Left) Two sets of generators of DAHA of type $C^\vee C_1$ presented by the fundamental group of a four-punctured sphere. (Right) Braid group action of $B_3$ on the four-punctured sphere.}
    \label{fig:daha generators}
\end{figure}

In \cite{cherednik1992double, Cherednik-book}, Cherednik showed that DAHAs associated with reduced root systems admit a projective action of $\SL(2, \bZ)$. For the DAHA of type $C^\vee C_1$, there is an action of the braid group $B_3$ \cite{noumi2004askey, stokman2003difference}, which can be visualized geometrically (see Figure \ref{fig:daha generators}). The braid group $B_3$ is generated by elements $\tau_\pm$, with the following relations:
\begin{equation}
 \sigma=\tau_+\tau_-^{-1}\tau_+=\tau_-\tau_+^{-1}\tau_-~. 
\end{equation}
The explicit action of the braid group on the generators of the DAHA is given by
\bea\label{braid_action_daha}
    \tau_+:(T_0,T_0^\vee,T_1,T_1^\vee)&\mapsto (T_0, T_0^\vee,T_1T_1^\vee T_1^{-1}, T_1)~,\cr
    (t_1,t_2,t_3,t_4)&\mapsto (t_1,t_2,t_4,t_3)~,\cr
    \tau_-:(T_0,T_0^\vee,T_1,T_1^\vee)&\mapsto ( T_0,T_0^{\vee-1}T_1^\vee T_0^\vee,T_1,T_0^\vee)~,\cr
    (t_1,t_2,t_3,t_4)&\mapsto (t_1,t_4,t_3,t_2)~,\cr
    \sigma:( T_0,T_0^\vee,T_1,T_1^\vee)&\mapsto \left(T_0, T_0^\vee,T_1T_1^\vee T_0^\vee (T_1T_1^\vee)^{-1}, T_1\right)~,\cr
    (t_1,t_2,t_3,t_4)&\mapsto (t_1,t_3,t_2,t_4)~.
\eea

For convenience, we introduce an alternative set of generators $X, Y, T$, which are related to the original ones as follows:
\begin{equation}
X=(T_0^\vee T_0)^{-1},\quad Y=T_0 T_1,\quad T=T_0~.
\end{equation}
In terms of these new generators, we have the following expressions:
\begin{equation}
T_0 = T,\quad T_0^\vee = X^{-1}T^{-1}, \quad T_1 = T^{-1}Y,\quad T_1^\vee = q^{-\frac{1}{2}}Y^{-1}TX~.
\end{equation}
For instance, if we specialize the parameters to
\be \label{A1Parameters}
(t_1,t_2,t_3,t_4)=(-iq^{-1/2}t,-i,-i,-i)~,
\ee 
the algebra reduces to the DAHA of type $A_1$, with  the defining relations:
\begin{equation}
T X T=X^{-1}, \quad T Y^{-1} T=Y, \quad(T-t^{-1})(T+t)=0, \quad     XYX^{-1}Y^{-1}T^2=q^{-1}~.
\end{equation}

\subsection{Spherical DAHA of \texorpdfstring{$C^\vee C_1$}{C*C1}}\label{sec:SH}

A distinguished subalgebra of the DAHA, known as the spherical subalgebra, is invariant under a certain Weyl group symmetry. We refer to this as the spherical DAHA, denoted by $\SH$, which will be the primary focus of this paper. 

It follows directly from the Hecke relation \eqref{CC1} that the element
\be
\mathbf{e} = \frac{iT_0 - q^{-\frac{1}{2}} t_1^{-1}}{q^{\frac{1}{2}} t_1 - q^{-\frac{1}{2}} t_1^{-1}}
\ee
is idempotent, satisfying $\mathbf{e}^2 = \mathbf{e}$ in $\HH$.\footnote{Given the four Hecke relations, there are indeed four corresponding idempotent elements  given by
\begin{equation}\label{idenpotents}
(\mathbf{e}_1,\mathbf{e}_2,\mathbf{e}_3,\mathbf{e}_4) = \left( \frac{iT_0 - q^{-\frac{1}{2}}t_1^{-1}}{q^{\frac{1}{2}}t_1 - q^{-\frac{1}{2}}t_1^{-1}}, \frac{iT_0^{\vee} - t_2^{-1}}{t_2 - t_2^{-1}}, \frac{iT_1 - t_3^{-1}}{t_3 - t_3^{-1}}, \frac{iT_1^{\vee} - t_4^{-1}}{t_4 - t_4^{-1}} \right).
\end{equation}
In this work, we select one of these idempotents to define the spherical subalgebra.
} 
The spherical DAHA is then defined by the idempotent projection
\be\label{idempotent-proj}
\SH = \mathbf{e} \HH \mathbf{e}~.
\ee 

The spherical DAHA $\SH$ is generated by the following generators \cite{Terwilliger2013}
\begin{align}
\begin{aligned}
&x=\left(T_0^\vee T_0+(T_0^\vee T_0)^{-1}\right)\mathbf{e}~=(X+X^{-1})\mathbf{e}~, \cr
&y=\left(T_0 ~ T_1+(T_0 ~ T_1)^{-1}\right)\mathbf{e}~=
(Y+Y^{-1})\mathbf{e}~, \cr
&z=\left(T_0 T_1^\vee+(T_0 T_1^\vee)^{-1}\right)\mathbf{e}~.
\end{aligned}
\end{align}
The algebraic structure of $\SH$ \cite{bullock2000multiplicative,Terwilliger2013} can be made explicit using $q$-commutator. By defining a $q$-commutator as
\be 
    [f,g]_q \equiv q^{-\frac{1}{2}}fg-q^{\frac{1}{2}}gf 
\ee
the generator of $\SH$ satisfies the following commutation relations
\bea\label{sDAHA_algebra}
&[x,y]_q=(q^{-1}-q)z+(q^{-\frac{1}{2}}-q^{\frac{1}{2}})\theta_3~, \cr
&[y,z]_q=(q^{-1}-q)x+(q^{-\frac{1}{2}}-q^{\frac{1}{2}})\theta_1~, \cr
&[z,x]_q=(q^{-1}-q)y+(q^{-\frac{1}{2}}-q^{\frac{1}{2}})\theta_2~,
\eea
where, using the notation $\overline{f} \equiv f+f^{-1}$, $\theta_i$ are expressed by 
\bea\label{theta_values}
    \theta_1&=\overline{t_1} \ \overline{t_2} + \overline{t_3} \ \overline{t_4}~,\cr
\theta_2&=\overline{t_1} \ \overline{t_3} + \overline{t_2} \ \overline{t_4}~,\cr
\theta_3&=\overline{t_1} \ \overline{t_4} + \overline{t_2} \ \overline{t_3}~.
\eea
In addition, the generators $x,y,z$ satisfy a Casimir relation 
\begin{equation}\label{q_cubic_eqn}
    -q^{-\frac{1}{2}}xyz+q^{-1}x^2+qy^2+q^{-1}z^2+q^{-\frac{1}{2}}\theta_1x+q^{\frac{1}{2}}\theta_2y+q^{-\frac{1}{2}}\theta_3z+\theta_4(q)=0~,
\end{equation}
where
\begin{equation}\label{Casemir_value}  \theta_4(q)=\overline{t_1}^2+\overline{t_2}^2+\overline{t_3}^2+\overline{t_4}^2+\overline{t_1} \ \overline{t_2} \ \overline{t_3} \ \overline{t_4}-q-q^{-1}-2~.
\end{equation}

As is evident from \eqref{sDAHA_algebra}, the spherical DAHA $\SH$ becomes commutative in the ``classical'' limit $q=1$, whereas the DAHA $\HH$ remains non-commutative even at $q=1$. Indeed, in the classical limit $q \to 1$, the Casimir relation \eqref{Casemir_value} reduces to the equation of an affine cubic surface:
\begin{equation}\label{cubic_eqn_2}
    -xyz+x^2+y^2+z^2+\theta_1x+\theta_2y+\theta_3z+\theta_4=0~,
\end{equation}
where
\begin{equation}   \theta_4=\overline{t_1}^2+\overline{t_2}^2+\overline{t_3}^2+\overline{t_4}^2+\overline{t_1} \ \overline{t_2} \ \overline{t_3} \ \overline{t_4}-4~.
\end{equation}
The earliest known appearance of the non-commutative cubic surface is found in Zhedanov’s Askey-Wilson algebra \cite{zhedanov1991hidden}. Its equivalence with the spherical DAHA $\SH$ was later established in \cite{koornwinder2008zhedanov}.
As we will see in \eqref{cubic_eqn_2}, this affine cubic surface describes the moduli space of flat $\SL(2, \bC)$-connections on a four-punctured sphere \cite{oblomkov2004cubic}. Therefore, in the classical limit, $\SH$ corresponds to the coordinate ring of the moduli space:
\be \label{classical-SH}
\SH \xrightarrow[q\to 1]{} 
 \OO(\MF(C_{0,4}, \SL(2, \bC)))~.
\ee 
In the next section, we will explore this point in greater detail from the perspective of the geometry of the affine cubic surface. For now, we turn our attention to the symmetries of the algebra.

The spherical DAHA $\SH$ exhibits three key types of symmetries \cite{Crampe:2020abq}, each of which plays an essential role in its structure.

\paragraph{Weyl group $W(D_4)$:}
The parameters $\theta_1, \theta_2, \theta_3, \theta_4$ are equal to the module characters of the Lie algebra $\mathfrak{so}(8) = D_4$. Specifically, the correspondence is:
\begin{align}\label{theta-SO8characters}
    \theta_1 = \chi_{\boldsymbol{8}_V}, \quad
    \theta_2 = \chi_{\boldsymbol{8}_S}, \quad
    \theta_3 = \chi_{\boldsymbol{8}_C}, \quad
    \theta_4 = \chi_{\boldsymbol{\mathrm{adj}}},
\end{align}
which are the characters of the vector, spinor, conjugate spinor, and adjoint representations of $\frakso(8)$, respectively. Note that $\overline{t_j}$ can be understood as the character of the fundamental representation of the corresponding $\fraksu(2)_j$ subalgebra.

The characters are invariant under this Weyl group $W(D_4)$ of $\frakso(8)$, and so is the $q$-deformed version $\theta_4(q)$ in \eqref{Casemir_value}. Consequently, the spherical DAHA is also invariant under $W(D_4)$. In fact, the idempotent projection \eqref{idempotent-proj} ensures that the algebra becomes invariant under this particular Weyl group $W(D_4)$.

\paragraph{Braid group $B_3$:} As seen in \eqref{braid_action_daha}, the full DAHA $\HH$ receives the action of the braid group $B_3$. In fact, the spherical subalgebra $\SH$ is invariant under this action. More explicitly, the generators of $B_3$ act on $\SH$ as
\begin{equation}
\label{braid_action_sdaha}
\begin{aligned}
\tau_+:~&
(x,y,z,\theta_{1},\theta_{2},\theta_{3},\theta_{4})\mapsto \left(x,q^{-\frac{1}{2}}(xy-q^{-\frac{1}{2}}z-\theta_3),y,\theta_{1},\theta_{3},\theta_{2},\theta_{4}\right)~,\cr
\tau_-:~&
(x,y,z,\theta_{1},\theta_{2},\theta_{3},\theta_{4})\mapsto \left(q^{-\frac{1}{2}}(xy-q^{-\frac{1}{2}}z-\theta_3),y,x,\theta_{3},\theta_{2},\theta_{1},\theta_{4}\right)~,\cr
\sigma:~&
(x,y,z,\theta_{1},\theta_{2},\theta_{3},\theta_{4})\mapsto \left(y,x,q^{-\frac{1}{2}}(xy-q^{-\frac{1}{2}}z-\theta_3),\theta_{2},\theta_{1},\theta_{3},\theta_{4}\right)~.
\end{aligned}
\end{equation}
In particular, the Casimir relation \eqref{q_cubic_eqn} is invariant under the action of $B_3$. Under the limit $q\to 1$ , these actions reduces to 
\bea\label{Classical_braid_action_sdaha}
\tau_+:~
(x,y,z,\theta_{1},\theta_{2},\theta_{3},\theta_{4})&\mapsto(x,xy-z-\theta_3,y,\theta_{1},\theta_{3},\theta_{2},\theta_{4})~,\cr 
\tau_-:~
(x,y,z,\theta_{1},\theta_{2},\theta_{3},\theta_{4})&\mapsto(xy-z-\theta_3,y,x,\theta_{3},\theta_{2},\theta_{1},\theta_{4})~,\cr
\sigma:~
(x,y,z,\theta_{1},\theta_{2},\theta_{3},\theta_{4})&\mapsto(y,x,xy-z-\theta_3,\theta_{2},\theta_{1},\theta_{3},\theta_{4})~.
\eea

Additionally, it is useful to consider a $\bZ_3$ cyclic permutation symmetry for the algebra, generated by
\begin{equation}
\label{A_3action}
s:
(x,y,z,\theta_{1},\theta_{2},\theta_{3},\theta_{4})\mapsto (y,z,x,\theta_{2},\theta_{3},\theta_{1},\theta_{4})~,
\end{equation}
with $s^3=\text{id}$. It is indeed a subgroup of the braid group $B_3$, as $s = \tau_+^{-1} \tau_-$. Indeed, it can be directly verified using the commutation relation \eqref{sDAHA_algebra} that
\begin{align}
    &-q^{-\frac{1}{2}}xyz+q^{-1}x^2+qy^2+q^{-1}z^2+q^{-\frac{1}{2}}\theta_1x+q^{\frac{1}{2}}\theta_2y+q^{-\frac{1}{2}}\theta_3z+\theta_4(q) \cr
    =&-q^{-\frac{1}{2}}yzx+q^{-1}y^2+qz^2+q^{-1}x^2+q^{-\frac{1}{2}}\theta_2y+q^{\frac{1}{2}}\theta_3z+q^{-\frac{1}{2}}\theta_1x+\theta_4(q)~,
\end{align}
thus the Casimir relation is indeed invariant under this permutation symmetry.

\paragraph{Sign-flip group $\bZ_2^{\times2}$:}
Lastly, there is a sign-flip symmetry  $\bZ_2\times \bZ_2=\bZ_2^{\times2}$ for the cubic surface, generated by $\xi_1$ and $\xi_2$:
\begin{equation}
\label{Z2action}
\begin{aligned}
\xi_1:~&
(x,y,z,\theta_{1},\theta_{2},\theta_{3},\theta_{4})\mapsto (-x,y,-z,-\theta_{1},\theta_{2},-\theta_{3},\theta_{4})~,\cr
\xi_2:~&
(x,y,z,\theta_{1},\theta_{2},\theta_{3},\theta_{4})\mapsto (x,-y,-z,\theta_{1},-\theta_{2},-\theta_{3},\theta_{4})~,\cr
\xi_3:~&
(x,y,z,\theta_{1},\theta_{2},\theta_{3},\theta_{4})\mapsto (-x,-y,z,-\theta_{1},-\theta_{2},\theta_{3},\theta_{4})~,\cr
\end{aligned}
\end{equation}
with $\xi_3=\xi_2\circ\xi_1$.

\subsubsection*{\texorpdfstring{$A_1$}{A1} limit of the DAHA}
The spherical DAHA of type $A_1$, which is studied in great detail in  \cite{Gukov:2022gei}, admits generators $x, y, z$ and parameters $q, t$, subject to the following relations:
\bea
    &[x, y]_q = (q^{-1} - q)z~, \cr
    &[y, z]_q = (q^{-1} - q)x~, \cr
    &[z, x]_q = (q^{-1} - q)y~, \cr
    -q^{-\frac{1}{2}}xyz + q^{-1}x^2 & + qy^2 + q^{-1}z^2 = q^{-1}t^2 + qt^{-2} + q + q^{-1}.
\eea
Comparing this with \eqref{sDAHA_algebra}, we observe that the $A_1$ limit is achieved by setting $\theta_j = 0$ for $j = 1, 2, 3$. From the definition of $\theta_j$ \eqref{theta_values}, there are two possible ways to take this limit:
\begin{equation}\label{A1_cond_1}
    \overline{t_k} = \overline{t_l} = \overline{t_m} = 0~,
\end{equation}
or
\begin{equation}\label{A1_cond_2}
    \overline{t_j} = -\overline{t_k} = -\overline{t_l} = -\overline{t_m}~,
\end{equation}
for any permutation $(j, k, l, m)$ of $(1, 2, 3, 4)$.

For the first condition \eqref{A1_cond_1}, the corresponding solutions are given by
\begin{equation}\label{A1_t_cond1}
    (t_j, t_k, t_l, t_m) = \left(\pm i \left(q^{-\frac{1}{2}}t\right)^{\pm1}, \pm i, \pm i, \pm i\right),
\end{equation}
where the signs can be chosen freely.
For the second condition \eqref{A1_cond_2}, the solutions take the form
\begin{equation}\label{A1_t_cond2}
    (t_j, t_k, t_l, t_m) = \left(t_*, -t_*^{\pm1}, -t_*^{\pm1}, -t_*^{\pm1}\right), \qquad t_* = \pm \left(\pm q^{-\frac{1}{2}}t\right)^{\pm\frac{1}{2}},
\end{equation}
where any choice of signs is allowed.

In the $A_1$ limit, the braid group action \eqref{braid_action_sdaha} reduces to the $\PSL(2,\bZ)$ action described in \cite{Gukov:2022gei}. Although the sign-flip symmetry remains, the symmetry of the Weyl group $W(D_4)$ is lost in this limit.

\subsection{Polynomial representation}\label{sec:polyrep}
As previously mentioned, the spherical DAHA serves as the underlying algebra governing $q$-difference operators acting on the Askey-Wilson polynomials. This connection is realized in the \emph{polynomial representation} of $\SH$ \cite{noumi2004askey, sahi1999nonsymmetric, stokman2003difference}, which acts on the space $\scP \equiv \bC_{q,\boldsymbol{t}}[X, X^{-1}]^{\bZ_2}$. Here, $\scP$ represents the space of symmetric Laurent polynomials with coefficients in the ring $\bC_{q,\boldsymbol{t}}$ of rational functions in $q$ and $\boldsymbol{t}$, obtained via the localization \eqref{localization}. The representation is given by
\begin{equation} \label{pol}
    \pol: \SH \to \End(\scP)~,
\end{equation}
and the actions of the generators are given by:
\begin{align}\label{pol2}
\begin{aligned}
    \pol(x) &= X + X^{-1}~, \\
    \pol(y) &= A(X; \boldsymbol{t})(\varpi - 1) + A(X^{-1}; \boldsymbol{t})(\varpi^{-1} - 1) - q^{\frac{1}{2}}t_1 t_3 - q^{-\frac{1}{2}}(t_1 t_3)^{-1}~, \\
    \pol(z) &= q^{\frac{1}{2}} \left[ X A(X) \varpi + X^{-1} A(X^{-1}) \varpi^{-1} \right] \\
    &\quad - \frac{X + X^{-1}}{q^{\frac{1}{2}} + q^{-\frac{1}{2}}} \left[ A(X) + A(X^{-1}) + q^{\frac{1}{2}} t_1 t_3 + q^{-\frac{1}{2}} t_1^{-1} t_3^{-1} \right] - \frac{\theta_3}{q^{\frac{1}{2}} + q^{-\frac{1}{2}}}~.
\end{aligned}
\end{align}
Here, we define the function $A(X; \boldsymbol{t})$ as
\begin{equation}
    A(X; \boldsymbol{t}) = -\frac{(1 + q^{\frac{1}{2}} t_1 t_2 X)(1 + q^{\frac{1}{2}} t_1 X / t_2)(1 + q^{\frac{1}{2}} t_3 t_4 X)(1 + q^{\frac{1}{2}} t_3 X / t_4)}{q^{\frac{1}{2}} t_1 t_3 (1 - X^2)(1 - q X^2)}~,
\end{equation}
and $\varpi$ is the $q$-shift operator, defined by:
\begin{equation}\label{qshift}
    \varpi \cdot f(X) = f(qX)~.
\end{equation}

A family of orthogonal symmetric polynomials of type $C^\vee C_{1}$, known as the Askey-Wilson polynomials, forms a basis for $\scP$. These polynomials $P_n$ ($n\in\bZ_{\ge0}$) can be expressed in terms of the basic hypergeometric series \cite{askey1985some} as
\begin{equation}\label{sym-Macdonald-A1}
P_{n}(X;q,\boldsymbol{t}):=\frac{(qab,qac,qad;q)_{n}}{q^{\frac{n}{2}}a^{n}(abcd q^{n+1};q)_{n}} {_4\phi_3}\left(\left.\begin{matrix}q^{-n}, q^{n+1} abcd,q^{\frac{1}{2}}aX,q^{\frac{1}{2}}aX^{-1}\\ qab,qac,qad\end{matrix}\right|q;q\right)~,
\end{equation}
where 
\begin{equation}\label{abcd-t}
    a=-t_1t_2,\quad b=-\frac{t_1}{t_2},\quad c=-t_3 t_4,\quad d=-\frac{t_3}{t_4}.
\end{equation}
Note that ${}_4\phi_{3}$ is the basic hypergeometric series, which can be expressed as a series expansion at $z=0$
\begin{equation}
    {}_4\phi_{3}\left(\left.\begin{matrix}a_1, a_2,a_3,a_4\\ b_{1},b_{2},b_{3}\end{matrix}\right|q;z\right)=\sum_{k=0}^{\infty} \frac{(a_1,a_2,a_3,a_4 ; q)_k }{(b_1,b_2,b_3,q ; q)_k} z^k .
\end{equation}
Here, we use the notation
\begin{equation}
(a_1, a_2, \ldots, a_m ; q)_n=(a_1 ; q)_n(a_2 ; q)_n\ldots(a_m ; q)_n ~,
\end{equation}
where we use the $q$-Pochhammer symbols
\begin{equation}
(a;q)_n=\prod_{k=0}^{n-1}(1-a q^k) ~.
\end{equation}

Under the polynomial representation \eqref{pol2}, the generators of $\SH$ act on the Askey-Wilson polynomial as
\bea\label{pol action of xyz}
\pol(x)\cdot P_{n}(X;q,\boldsymbol{t})&= P_{n+1}(X;q,\boldsymbol{t}) +B_n P_{n}(X;q,\boldsymbol{t})+C_n P_{n-1}(X;q,\boldsymbol{t})~~, \cr
\pol(y)\cdot P_{n}(X;q,\boldsymbol{t})&= \left( -q^{n+\frac{1}{2}}t_1t_3-q^{-n-\frac{1}{2}}t_1^{-1}t_3^{-1} \right) P_{n}(X;q,\boldsymbol{t})~~, \cr
\pol(z)\cdot P_{n}(X;q,\boldsymbol{t})&= -q^{n+1}t_1t_3 P_{n+1}(X;q,\boldsymbol{t})-q^{-n}t_1^{-1}t_3^{-1}C_n P_{n-1}(X;q,\boldsymbol{t}) \cr
&\quad +\left(B_n-q^{n+\frac{1}{2}}\frac{(t_1-t_3) (1-t_1 t_3) (t_2-t_4) (1-t_2 t_4)}{t_2 t_4 (1-t_1 t_3 q^{n}) (1-t_1 t_3 q^{n+1})} \right) P_{n}(X;q,\boldsymbol{t})~, 
\eea
where
   \begin{small}
\begin{align}
B_n&=\resizebox{0.29\textwidth}{!}{$\frac{-q^{n+\frac{1}{2}} }{t_2 t_4 (t_1^2 t_3^2 q^{2 n}-1) (t_1^2 t_3^2 q^{2 n+2}-1)}$} \Big(t_1 (t_2^2+1) t_4 \left\{1+t_3^2 \left[1+q^{2 n+1} t_1^2 (t_3^2+1) -q^n(q+1) (t_1^2+1) \right]\right\} \cr
& \qquad\qquad\qquad\qquad\qquad\qquad\qquad+t_2 t_3 (t_4^2+1) \left\{1+t_1^2 \left[ 1+ q^{2 n+1} t_3^2 (t_1^2+1) -q^n(q+1) (t_3^2+1) \right]\right\} \Big)~,
\cr
C_n &=\resizebox{0.93\textwidth}{!}{$\frac{(q^n-1) (t_1^2 q^n-1) (t_3^2 q^n-1) (t_1^2 t_3^2 q^n-1) (t_1 t_2 t_3 q^n-t_4) (t_1 t_3 q^n-t_2 t_4) (t_1 t_3 t_4 q^n-t_2) (t_1 t_2 t_3 t_4 q^n-1)}{t_2^2 t_4^2 (t_1^2 t_3^2 q^{2 n}-1)^2 (t_1^2 t_3^2 q^{2 n-1}-1) (t_1^2 t_3^2 q^{2 n+1}-1)}~.$} \nonumber
\end{align}\end{small}
We can define the raising and lowering operators in $\SH$ with respect to this basis \cite{sahi2007raising} as
\bea\label{raising_Lowering_operators}
     R_n=-q^{1-n}x -qt_1t_3 z+ D_n ~, \quad
   L_n=-q^{n+2}x -\frac{q}{t_1t_3} z+E_n~,
\eea
where
\begin{align}
D_n=\frac{q^{3/2}t_1t_3\left(q^{n+1}t_1t_3\theta_1-\theta_3\right)}{q^{2n+2}t_1^2t_3^2 -1}\nonumber ~,\quad
E_n=\frac{q^{3/2}t_1^{-1}t_3^{-1}\left(q^{-n}t_1^{-1}t_3^{-1}\theta_1-\theta_3\right)}{1-q^{-2n}t_1^{-2}t_3^{-2}}~.
\end{align}

\subsubsection*{Raising and lowering operators}
Under the polynomial representation \eqref{pol}, these operators raise and lower the labels $n$ of the Askey-Wilson polynomials as follows:
\begin{align}
    &\pol(R_n)\cdot P_{n}(X;q,\boldsymbol{t})= q^{1-n}\left(t_1^2t_3^2q^{2n+1}-1\right) P_{n+1}(X;q,\boldsymbol{t})~,\label{raising} \\
&\pol(L_n)\cdot  P_{n}(X;q,\boldsymbol{t})= -q^{1-n}\left(q^{n}-t_1^{-1}t_2t_3^{-1}t_4\right)\left(q^{n}-t_1^{-1}t_2t_3^{-1}t_4^{-1}\right)\left(q^{n}-t_1^{-1}t_2^{-1}t_3^{-1}t_4\right)\cr
&\hspace{1.1cm}\times\left(q^{n}-t_1^{-1}t_2^{-1}t_3^{-1}t_4^{-1}\right)\frac{(q^n-1)(q^n-t_1^{-2})(q^{n}-t_3^{-2})(q^{n}-t_1^{-2}t_3^{-2})}{\left(q^{2n}-t_1^{-2}t_3^{-2}\right)^2(q^{2n-1}-t_1^{-2}t_3^{-2})} P_{n-1}(X;q,\boldsymbol{t})~.\label{lowering}
\end{align}

As demonstrated above, the spherical DAHA $\SH$ is invariant under the Weyl group $W(D_4)$. However, the polynomial representation is \emph{not} invariant under $W(D_4)$. In fact, as we will see more explicitly, acting the Weyl group on the polynomial representation \eqref{pol2} yields eight distinct representations. 

Furthermore, incorporating the action of the permutation group $\bZ_3$, generated by \eqref{A_3action}, one can obtain 24 distinct representations in total from the polynomial representation. In this way, we can use the action of $W(D_4)$ to understand the representation theory of $\SH$.

\subsubsection*{Finite-dimensional representations}

A common way to construct a finite-dimensional representation is to find null vectors for the raising and lowering operators. The raising operators can never be null since the Askey-Wilson polynomial $P_{n+1}(X;q,\boldsymbol{t})$ always contains a factor $(t_1^2t_3^2q^{2n+1}-1)$ in the denominator that cancels the raising coefficient in \eqref{raising}. However, there are conditions where the lowering operator annihilates an Askey-Wilson polynomial, say $P_{n}(X;q,\boldsymbol{t})$, which becomes the lowest weight state of a sub-representation. In this way, a finite-dimensional $\SH$-module appears as the quotient $\scP/(P_{n})$ of the polynomial representation by the ideal $(P_{n})$. It is worth noting that not all finite-dimensional representations can be obtained in this way. 

Therefore, we can study finite-dimensional representations by imposing the condition $\pol(L_{n})\cdot P_{n}=0$. Namely, when the lowering coefficient \eqref{lowering} vanishes:
\begin{multline}\label{lowering_coefficient}
    -q^{1-n}(q^{n}-t_1^{-1}t_2t_3^{-1}t_4)(q^{n}-t_1^{-1}t_2t_3^{-1}t_4^{-1})(q^{n}-t_1^{-1}t_2^{-1}t_3^{-1}t_4)\cr
\times(q^{n}-t_1^{-1}t_2^{-1}t_3^{-1}t_4^{-1})\frac{(q^n-1)(q^n-t_1^{-2})(q^{n}-t_3^{-2})(q^{n}-t_1^{-2}t_3^{-2})}{(q^{2n}-t_1^{-2}t_3^{-2})^2(q^{2n-1}-t_1^{-2}t_3^{-2})}=0~.
\end{multline}
This amounts to the seven shortening conditions as follows
\begin{equation}\label{generic_shortening_cond}
q^{n}=1,~ t_1^{-2},~ t_3^{-2},~
t_1^{-1}t_2t_3^{-1}t_4,~
t_1^{-1}t_2t^{-1}_3t_4^{-1},~ t_1^{-1}t_2^{-1}t_3^{-1}t_4,~  t_1^{-1}t_2^{-1}t_3^{-1}t_4^{-1}~.
\end{equation}
Note that the shortening condition $q^n=t_1^{-2}t_3^{-2}$ is ruled out because some lowering operator becomes ill-defined due to the vanishing of the denominator.  

Taking into account the eight distinct representations obtained by the $W(D_4)$ action, the set of shortening conditions can be interpreted within the framework of the $D_4$ root system. The shortening conditions \eqref{generic_shortening_cond} can be identified with a subset of roots of type $D_4$. By using the Weyl group symmetry, all the shortening conditions for finite-dimensional representations obtained in this way can be neatly repackaged as
\be\label{finite dimensional representation}
    q^n = t_1^{-r_1} t_2^{-r_2} t_3^{-r_3} t_4^{-r_4}=:\boldsymbol{t}^{-r}, \qquad r \in \sfR(D_4)\cup\{(0,0,0,0)\}~.
\ee
(See Appendix \ref{app:notation} for the convention of the roots.)

When $q$ is not a root of unity, all finite-dimensional representations are rigid. 
In this regime, the finite-dimensional modules of the full DAHA $\HH$ were classified in
\cite{oblomkov2009finite}, where it was shown that every such module can be realized as a
quotient of a polynomial representation. 
Moreover, the algebras $\HH$ and $\SH$ are Morita equivalent, so their representation
theories are equivalent \cite{oblomkov2004cubic}. 
The precise correspondence between the polynomial representations used in
\cite{oblomkov2009finite} and those appearing in the present work is spelled out at the end
of Appendix~\ref{app:24lines}. 

We will return to this classification of finite-dimensional representations from the
perspective of brane quantization in \S\ref{sec:brane-rep}.

\subsubsection*{Deformation of the polynomial representation}

Let 
\[
\mathbb{K}_{q, \boldsymbol{t}} := \mathbb{C}\bigl(q^{1/2}, t_1, t_2, t_3, t_4 \bigr)
\]
be the field of rational functions in the parameters \( q^{1/2} \) and 
\(\boldsymbol{t} = (t_1, t_2, t_3, t_4)\) over \(\mathbb{C}\). If we enlarge the representation space to the ring of formal power series
\begin{equation}
\PR^{y_1} = \mathbb{K}_{q, \boldsymbol{t}}\llbracket X^{\pm1} \rrbracket
\end{equation}
in the variable \(X\) with coefficients in \(\mathbb{K}_{q, \boldsymbol{t}}\),
then we obtain a family of representations of \(\SH\) on \(\PR^{y_1}\), labeled by \(y_1 \in \mathbb{C}^\times\), deforming the standard polynomial representation \eqref{pol}:
\begin{equation}\label{poly-rep-y1}
\SH \longrightarrow \End\bigl(\PR^{y_1}\bigr).
\end{equation}
where the action of the generators \(x, y, z \in \SH\) on \(\PR^{y_1}\) is given by
\begin{align}\label{gen-pol}
\begin{aligned}
    \pol_{y_1}(x) &= X + X^{-1}, \\[4pt]
    \pol_{y_1}(y) &= A(X; \boldsymbol{t}) \bigl(y_1 \varpi - 1 \bigr) 
    + A(X^{-1}; \boldsymbol{t}) \bigl(y_1^{-1} \varpi^{-1} - 1 \bigr) 
    - q^{1/2} t_1 t_3 - q^{-1/2}(t_1 t_3)^{-1}, \\[4pt]
    \pol_{y_1}(z) &= q^{1/2} \bigl[\, X A(X;\boldsymbol{t})\, y_1 \varpi 
    + X^{-1} A(X^{-1};\boldsymbol{t})\, y_1^{-1} \varpi^{-1} \bigr] \\
    &\quad - \frac{X + X^{-1}}{q^{1/2} + q^{-1/2}} 
    \Bigl[\, A(X;\boldsymbol{t}) + A(X^{-1};\boldsymbol{t}) 
    + q^{1/2} t_1 t_3 + q^{-1/2} t_1^{-1} t_3^{-1} \Bigr] 
    - \frac{\theta_3}{q^{1/2} + q^{-1/2}}.
\end{aligned}
\end{align}

For generic $y_1\in\mathbb C^\times$, the eigenfunction of $y$ under the $\pol_{y_1}$ is constructed in \cite{koelink2001askey}, which is called the Askey-Wilson function:
\begin{multline}\label{AW-function}
    P_{y_1}(X;q,\boldsymbol{t})=\frac{(y_1 X,q(y_1X)^{-1};q)_\infty}{( X,qX^{-1};q)_\infty}\Bigg[{_4\phi_3}\left(\left.\begin{matrix}y_1^{-1}, y_1q abcd,q^{\frac{1}{2}}aX,q^{\frac{1}{2}}aX^{-1}\\ qab,qac,qad\end{matrix}\right|q;q\right)\\
    +K(X)\cdot {_4\phi_3}\left(\left.\begin{matrix}y_1 qbc, (y_1ad)^{-1},q^{\frac{1}{2}}d^{-1}X,q^{\frac{1}{2}}(dX)^{-1}\\ q(ad)^{-1},qbd^{-1},qcd^{-1}\end{matrix}\right|q;q\right)\Bigg]~,
\end{multline}
with
\begin{equation}
K(X)=\frac{\left((ad)^{-1},qbc,qad^{-1},qbd^{-1},qcd^{-1},y_1^{-1},y_1qabcd,q^{\frac{1}{2}}aX,q^{\frac{1}{2}}aX^{-1};q\right)_\infty}{\left(ad,qbc,qab,qac,qad^{-1},y_1qbc,(y_1ad)^{-1},q^{\frac{1}{2}}d^{-1}X,q^{\frac{1}{2}}(dX)^{-1};q\right)_\infty}~.
\end{equation}
Then, the eigenvalue is given by
\begin{equation}
    \pol_{y_1}(y)\cdot  P_{y_1}(X;q,\boldsymbol{t})= -\left( y_1 q^{\frac{1}{2}}t_1t_3+y_1^{-1}q^{-\frac{1}{2}}t_1^{-1}t_3^{-1} \right)  P_{y_1}(X;q,\boldsymbol{t})~~.
\end{equation}
If we set $y_1=q^n$, the function \eqref{AW-function} becomes the polynomial $P_n$ in \eqref{sym-Macdonald-A1} up to a factor
\be 
(-1)^n\frac{(qab,qac,qad;q)_{n}}{q^{\frac{n}{2}}a^{n}(abcd q^{n+1};q)_{n}} \ P_{y_1=q^n}(X;q,\boldsymbol{t})=P_n(X;q,\boldsymbol{t})~.
\ee

\section{Geometry of Coulomb Branch}\label{sec:geometry}

In this paper, we study the representation theory of the spherical DAHA $\SH$ via brane quantization in the
two-dimensional $A$-model whose target is the moduli space of flat $\SL(2,\bC)$-connections on a four-punctured sphere.
The purpose of this section is to describe the relevant geometry of the target space that will be used throughout the paper.

We begin by recalling how $\SH$ arises from the deformation quantization of the coordinate ring of the
$\SL(2,\bC)$ character variety of the four-punctured sphere. We then summarize the brane-quantization framework \cite{Gukov:2008ve}
that we use to relate $\SH$-modules to $A$-branes.
A central objective of this section is to identify compact Lagrangian submanifolds
in the target geometry, as we will identify them to
finite-dimensional $\SH$-modules in \S\ref{Sec:Brane_Finite_Rep}. For this purpose, we adopt the perspective of the Hitchin moduli space 
and analyze the Hitchin fibration, focusing in particular on the classification of its singular fibers.

Compared with the situation studied in \cite{Gukov:2022gei}, the presence of four ramification parameters makes the fibration structure
more involved. Nevertheless, the connection between the moduli space and the $D_4$ root system makes the underlying geometric structure transparent: it allows us to track the relevant cycles and their degenerations in a uniform way, and it will be used repeatedly in later sections.

Metaphorically, the target space of the 2d sigma-model serves as the stage upon which the main actors—branes—will appear. With that in mind, let us proceed by carefully setting the stage for the forthcoming analysis.

\subsection{Revisiting Seiberg-Witten theory of SU(2) with \texorpdfstring{$N_f=4$}{Nf=4}}\label{sec:prelude}

In the construction of DAHA of type $C^\vee C_1$, we observed that a four-punctured sphere emerges naturally. In the class $\cS$ construction for type $A_1$ \cite{Witten:1997sc,Gaiotto:2009we}, the four-punctured sphere corresponds to a 4d $\cN=2$ supersymmetric theory with gauge group $G=\SU(2)$ and four fundamental hypermultiplets ($N_f=4$). This theory has a rich history, dating back to the original Seiberg-Witten paper \cite{Seiberg:1994aj}. If all the hypermultiplets are massless, the theory is superconformal. Since the fundamental representation of SU(2) is pseudo-real, the eight fundamental half-hypermultiplets are subject to the SO(8) flavor symmetry. In fact, the geometry we are interested in naturally appears as the Coulomb branch of this theory compactified on $S^1$. Since the Seiberg-Witten theory plays a crucial role in understanding the geometry of the Coulomb branch, let us recall the Seiberg-Witten analysis of this theory.
 
The construction is associated to a Hitchin system  \cite{hitchin1987self} where the base curve is a four-punctured sphere $C_{0,4}$
\begin{equation}
\label{eq:spectral}
\begin{tikzcd}
\Sigma \arrow[r,hook]
& T^*C_{0,4} \arrow[d] \\
& C_{0,4}
\end{tikzcd}\qquad,
\end{equation}
where $\Sigma$ is a Seiberg-Witten (or spectral) curve.
(Historically, Garnier \cite{garnier1919classe} first introduced the notion of the spectral curve in genus zero.) 
Without loss of generality, the four ramification points can be placed at $\{0,1,q,\infty \}$ on a Riemann sphere where $q=e^{2\pi i \tau}$ can be interpreted as a complexified coupling constant $\tau=\frac{\theta}{\pi}+\frac{8\pi i}{g^2}$ of the 4d $\cN=2$ superconformal theory in the ultra-violet.

The data of the Hitchin system are as follows.  
Let $C_{0,4}$ be a Riemann sphere and $p_j \in C_{0,4}$ ($j=1,2,3,4$) four marked points. Then, the meromorphic Higgs bundles on $C_{0,4}$ are pairs $(E, \varphi)$ with $E$ a holomorphic vector bundle of rank two on $C_{0,4}$ and $\varphi: E \rightarrow E \otimes K_{C_{0,4}}(D)$ a meromorphic Higgs field with polar divisor $D=p_1+p_2+p_3+p_4$ where $K_{C_{0,4}}$ denotes the canonical bundle of $C_{0,4}$. Recall that the slope of a vector bundle $E$ is defined as $\mu(E)=\textrm{deg}(E)/\textrm{rk}(E)$. As usual $(E, \varphi)$ is (semi)stable if $\mu(E^{\prime}) \leqq \mu(E)$ for any $\varphi$ invariant subbundle $E^{\prime} \subset E$. The tame ramification at $p_j$ is described by the following conditions on the SU(2) connection $A$ on $E$ and the Higgs field
\bea\label{tame}
A &=  \alpha_j \,d\vartheta +\cdots \cr
\varphi&= {\frac12}(\beta_j+i\gamma_j){\frac{dz}{z}}+\cdots
\eea
Here, $z=re^{i\vartheta}$ is a local coordinate on a small disk centered at $p_j$, and the ramification data is a triple  $(\alpha_j,\beta_j,\gamma_j)\in \frakt\times \frakt \times \frakt$ where we denote the Cartan subalgebra $\frakt\subset\fraksu(2)$. 
By expressing them as conjugate to diagonal matrices
\bea
\textrm{diag}(\talpha_j, -\talpha_j) &\sim \a_j \in \frakt~, \cr
\textrm{diag}(\tbeta_j, -\tbeta_j) &\sim \b_j \in \frakt~, \cr
\textrm{diag}(\tgamma_j, -\tgamma_j) &\sim \g_j \in \frakt~,
\eea
we will often use the triple $(\a_j, \b_j, \g_j) \in \frakt^3$ interchangeably with the ramification parameters $(\talpha_j, \tbeta_j, \tgamma_j) \in \bR^3$ in the following discussions, by a slight abuse of notation.

The relation between masses $m_j$ of hypermultiplets and ramification parameters of a four-punctured sphere is given by \cite{Alday:2009aq,Yoshida:2025iae}
\begin{equation}\label{mass_monodromy}
(\tbeta_1+i\tgamma_{1},\tbeta_2+i\tgamma_{2},\tbeta_3+i\tgamma_{3},\tbeta_4+i\tgamma_{4})=\left(\frac{m_{1}+m_{2}}{2},\frac{m_{1}-m_{2}}{2},\frac{m_{3}+m_{4}}{2},\frac{m_{3}-m_{4}}{2}\right).
\end{equation}
The ramification parameters $\talpha_j$ are invisible in the 4d $\cN=2$ theory on $\bR^4$, but they show up as gauge holonomies around the corresponding puncture when compactified on $S^1$.

When the four hypermultiplets are massless, the Seiberg-Witten curve is given by \cite{Seiberg:1994aj} 
\begin{equation}
    y^2=\left(x-e_1(\tau) u\right)\left(x-e_2(\tau) u\right)\left(x-e_3(\tau) u\right)~,
\end{equation}
with $u$ the parameter of Coulomb branch and $e_i(\tau)$ are roots of the cubic polynomial $4x^3-g_2(\tau)x-g_3(\tau)$ obeying $e_1+e_2+e_3=0$. Note that $g_2(\tau)$ and $g_3(\tau)$ are appropriately normalized Eisenstein series
\begin{align}
\begin{aligned}
       g_2(\tau)&=\frac{60}{\pi^4}G_4(\tau)=\frac{60}{\pi^4}\sum_{m,n\in \bZ_{\neq0}}\frac{1}{(m\tau+n)^4}~,\cr
    g_3(\tau)&=\frac{140}{\pi^6}G_6(\tau)=\frac{140}{\pi^6}\sum_{m,n\in \bZ_{\neq0}}\frac{1}{(m\tau+n)^6}~.
\end{aligned}
\end{align}
The roots $e_i$ can be expressed by the Jacobi theta functions
\begin{align}
\begin{aligned}
    & e_1-e_2=\vartheta_3{ }^4(0, \tau)~, \cr
& e_3-e_2=\vartheta_2{ }^4(0, \tau)~, \cr
& e_1-e_3=\vartheta_4{ }^4(0, \tau)~,
\end{aligned}
\end{align}
where the Jacobi theta functions are defined by
\begin{align}\label{jacob_theta}
\begin{aligned}
    & \vartheta_2(0, \tau)=\sum_{n \in \bZ} q^{\frac{1}{2}(n+1 / 2)^2} ~, \cr
& \vartheta_3(0, \tau)=\sum_{n \in \bZ} q^{\frac{1}{2} n^2}~,
\cr
& \vartheta_4(0, \tau)=\sum_{n \in \bZ}(-1)^n q^{\frac{1}{2} n^2}~.
\end{aligned}
\end{align}
At generic mass parameters $m_j$, the Seiberg-Witten curve is obtained in \cite[Eqn. (16.38)]{Seiberg:1994aj}, which is
\begin{align}\label{SW-Nf4-generic}
y^2 &=\left(x^2-c_2^2 u^2\right)\left(x-c_1 u\right)-c_2^2\left(x-c_1 u\right)^2 \sum_i m_i^2-c_2^2\left(c_1^2-c_2^2\right)\left(x-c_1 u\right) \sum_{i>j} m_i^2 m_j^2 \cr
& +2 c_2\left(c_1^2-c_2^2\right)\left(c_1 x-c_2^2 u\right) m_1 m_2 m_3 m_4-c_2^2\left(c_1^2-c_2^2\right)^2 \sum_{i>j>k} m_i^2 m_j^2 m_k^2,
\end{align}
where $c_1=\frac{3}{2} e_1$ and $c_2=\frac{1}{2}(e_3-e_2)$.

When compactifying the 4d $\cN=2$ SU(2) gauge theory with $N_f=4$ hypermultiplets on $S^1$, the Coulomb branch becomes the moduli space $\MF(C_{0,4},\SL(2,\bC))$ of flat $\SL(2,\bC)$-connections on the four-punctured sphere \cite{Gaiotto:2009hg,Gaiotto:2010be} (a.k.a. the $\SL(2,\bC)$ character variety). Via the non-abelian Hodge correspondence \cite{hitchin1987self,corlette1988flat,simpson1988constructing,simpson1990harmonic}, the $\SL(2,\bC)$ character variety is diffeomorphic to the moduli space $\MH(C_{0,4},\SU(2))$ of the SU(2) Higgs bundles on $C_{0,4}$, (a.k.a. the Hitchin moduli space). The Hitchin moduli space $\MH(C_{0,4},\SU(2))$ is the space of pairs $(E,\varphi)$ of a holomorphic bundle $E$ and the corresponding Higgs field $\varphi$, imposing a stability condition called Hitchin equations \cite{hitchin1987self}:
\bea
\label{Hitchin-eq}
F-[\varphi, \overline \varphi]&=0~,\cr
\overline D_A\,\varphi&=0~.
\eea
This space is a \HK space \cite{konno1993construction,nakajima1996hyper,biquard2004wild} with three integrable complex structures $I$, $J$, $K$, and the corresponding \K forms are given by
\bea\label{kahler-form}
\omega_I &=
-{\frac{i}{2\pi}}\int_{C_{0,4}}|d^2z|\,\Tr\Bigl(\delta A_{\bar z}\wedge
\delta A_z-\delta \bar\varphi\wedge \delta\varphi\Bigr)~,\cr
\omega_J &={\frac{1}{2\pi}}\int_{C_{0,4}} |d^2z|\,\Tr\Bigl(
\delta\bar\varphi\wedge \delta A_z+\delta\varphi\wedge
\delta A_{\bar z} \Bigr)~~,\cr
\omega_K &=
{\frac{i}{2\pi}}\int_{C_{0,4}}|d^2z|\,\Tr\Bigl(\delta\bar\varphi
\wedge\delta A_z-\delta\varphi\wedge\delta A_{\bar z}\Bigr)~.
\eea

\begin{figure}[ht]
    \centering
    \includegraphics{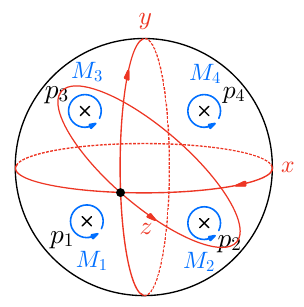}
    \caption{Generators of the fundamental group of a four-punctured sphere and the generators of spherical DAHA of type $C^\vee C_1$.}
    \label{fig:sdaha generators}
\end{figure}

The description of $\MH(C_{0,4}, G)$ as the Hitchin moduli space given above is in complex structure $I$, while the description as character variety $\MF(C_{0,4},\SL(2,\bC))$ arises in complex structure $J$. In this complex structure, a complex combination $A_{\bC} = A + i ( \varphi + \bar{\varphi})$, can be identified with an $\SL(2,{\bC})$-connection. The Hitchin equations then become the flatness condition $F_{\bC} = dA_{\bC} + A_{\bC} \wedge A_{\bC} = 0$ for this $\SL(2,{\bC})$-connection $A_{\bC}$. Part of the data for the tame ramification \eqref{tame} is encoded in a monodromy $M_j$ around the corresponding puncture $p_j$, which is conjugate to
\be 
M_{j}\sim\exp \left(-2 \pi\left(\gamma_j+i \alpha_j\right)\right)~
\ee
while $\tbeta_j$ is a \K parameter in this complex structure.

These monodromy matrices naturally define an $\SL(2,\bC)$ holonomy representation of the fundamental group of a four-punctured sphere (see Figure \ref{fig:sdaha generators}):
\begin{equation}\label{monodromyRep}
  \MF(C_{0,4},\SL(2,\bC)) = \langle M_1, M_2, M_3, M_4 \in \SL(2,\bC) \mid M_1M_2M_3M_4 = \textrm{Id} \rangle / \SL(2,\bC)~,
\end{equation}
where the quotient by $\SL(2,\bC)$ is taken with respect to conjugation. 

To describe the character variety geometrically, we introduce holonomy variables as holomorphic functions on $\MF(C_{0,4},\SL(2,\bC))$:
\begin{align}
   & x = -\Tr(M_1M_2)~, \quad y = -\Tr(M_1M_3)~, \quad z = -\Tr(M_2M_3)~, \cr
   & \overline{t_j} = \Tr(M_j)~, \qquad (j=1,2,3,4)~.
\end{align}
These variables are subject to the trace identity \cite{goldman2009trace}:
\begin{equation}\label{cubic_eqn}
    f(x,y,z) = -xyz + x^2 + y^2 + z^2 + \theta_1x + \theta_2y + \theta_3z + \theta_4 = 0~,
\end{equation}
where
\begin{align}\label{theta_values2}
\begin{aligned}
    \theta_1 &= \overline{t_1} \ \overline{t_2} + \overline{t_3} \ \overline{t_4}~, \cr
    \theta_2 &= \overline{t_1} \ \overline{t_3} + \overline{t_2} \ \overline{t_4}~, \cr
    \theta_3 &= \overline{t_1} \ \overline{t_4} + \overline{t_2} \ \overline{t_3}~, \cr
    \theta_4 &= \overline{t_1}^2 + \overline{t_2}^2 + \overline{t_3}^2 + \overline{t_4}^2 + \overline{t_1} \ \overline{t_2} \ \overline{t_3} \ \overline{t_4} - 4~.
\end{aligned}
\end{align}
As explained in \eqref{theta-SO8characters}, these are the characters of the SO(8) representations, which can be attributed to the SO(8) flavor symmetry present in the 4d $\cN=2$ SU(2) theory with $N_f=4$.
Thus, in complex structure $J$, the character variety $\MF(C_{0,4},\SL(2,\bC))$ is an affine variety described by this cubic equation. The four complex structure parameters $(t_1,t_2,t_3,t_4)$ are identified as
\be \label{tj}
t_j = \exp(-2\pi (\tgamma_j + i\talpha_j))~, \qquad (j=1,2,3,4).
\ee 
It becomes evident that in the classical limit $q\to1$, the spherical DAHA $\SH$ reduces to the coordinate ring of the character variety, as mentioned in \eqref{classical-SH}.

Physically, the holonomy variables $x, y, z$ can be interpreted as vacuum expectation values of loop operators along $S^1$ in the 4d $\cN=2$ theory. In the class $\cS$ construction, two M5-branes wrap the four-punctured sphere $C_{0,4}$, with the punctures realized by intersections of co-dimension two with other M5-branes. A line operator in the 4d $\cN=2$ theory is realized by an M2-brane attaching to a one-cycle on $C_{0,4}$ \cite{Alday:2009fs,Drukker:2009id}, and the one-cycles $x, y, z$ in Figure \ref{fig:sdaha generators} correspond to fundamental Wilson, 't Hooft, and dyonic loop operators, respectively:
\bea 
x \ & \longleftrightarrow \textrm{ Wilson loop}~,\cr 
y \ & \longleftrightarrow  \textrm{ 't Hooft loop}~,\cr
z \ & \longleftrightarrow  \textrm{ dyonic loop}~.
\eea 
Consequently, the algebra of loop operators gives rise to the coordinate ring $\scO(\MF(C_{0,4},\SL(2,\bC)))$ of the Coulomb branch holomorphic in complex structure $J$ \cite{Gaiotto:2010be,Teschner:2013tqy}. 

One can introduce the $\Omega$-background $S^1\times \bR\times_q \bR^2$, which effectively introduces a potential around the origin of the $\Omega$-deformation. As illustrated in Figure \ref{fig:omega-background}, the loop operators are localized along the axis of the $\Omega$-deformation and are forced to cross each other as they exchange positions. Consequently, the algebra of loop operators becomes non-commutative, providing a physical realization of the deformation quantization of the coordinate ring \cite{Nekrasov:2010ka,Ito:2011ea,Yagi:2014toa,Bullimore:2016nji,Dedushenko:2018icp,Okuda:2019emk}.

The deformation quantization we consider is with respect to the holomorphic symplectic form $\Omega_J = \omega_K + i\omega_I$ of complex structure $J$ (also known as the Atiyah-Bott-Goldman symplectic form), which is given by
\begin{equation}
    \Omega_{J}=-\frac{1}{2\pi}\frac{dx\wedge dy}{\partial f/\partial z}=-\frac{1}{2\pi }\frac{dx\wedge dy}{2z-xy+\theta_{3}}.
\end{equation}
In terms of the holomorphic symplectic form $\Omega_{J}$, the Poisson brackets of the generators of the coordinate ring are given by
\begin{equation}
\begin{aligned}
       \{x,y\}= -2\pi (2z-xy+\theta_3)~,\\
       \{y,z\}=-2\pi  (2x-yz+\theta_1)~,\\
       \{z,x\}=-2\pi  (2y-zx+\theta_2)~.\\
\end{aligned}
\end{equation}
Using the algebraic relation of the spherical DAHA $\SH$ in \eqref{sDAHA_algebra}, one can show that, with $q=e^{2\pi i \hbar}$, the Poisson brackets can be obtained by
\begin{equation}
    \eval{\frac{[-,-]}{i\hbar}}_{\hbar\rightarrow0}=\{-,-\}~,
\end{equation}
where $[x,y]=x y-y x$ is the ordinary commutator (not the $q$-commutator). This verifies that the spherical DAHA $\SH$ is indeed the deformation quantization of the coordinate ring with respect to the holomorphic symplectic form $\Omega_J$ \cite{oblomkov2004cubic}:
\begin{equation}\label{deformation of coordinate ring}
    \SH\cong \OO^{q}(\MF(C_{0,4},\SL(2,\bC)))~.
\end{equation}
In other words, the spherical DAHA $\SH$ is the algebra of loop operators in the 4d $\cN=2$ SU(2) gauge theory with $N_f=4$ hypermultiplets on the $\Omega$-background \cite{Cirafici:2020qlf,Yoshida:2025iae}.

Another perspective comes from the Kauffman bracket skein algebra \cite{turaev1991skein,przytycki2006skein} of a Riemann surface, which also provides a deformation quantization of the coordinate ring of the character variety on the surface with respect to $\Omega_J$ \cite{bullock1997rings,przytycki1997skein}. Indeed, the algebraic relations \eqref{sDAHA_algebra} and \eqref{q_cubic_eqn} for $\SH$ can also be derived as the skein algebra $\textrm{Sk}(C_{0,4})$ of a four-punctured sphere in \cite{bullock2000multiplicative}. This perspective has been explored extensively in the literature \cite{berest2016double,berest2018affine,Hikami:2019jaw,cooke2020kauffman,bousseau2023strong,Allegretti:2024svn}; see also \cite[\S2.5]{Gukov:2022gei} for discussions on the relationship between skein algebras/modules and brane quantization. Accordingly, we will not go into the details here, referring instead to these sources.

\begin{figure}
    \centering
    \includegraphics[width=0.9\linewidth]{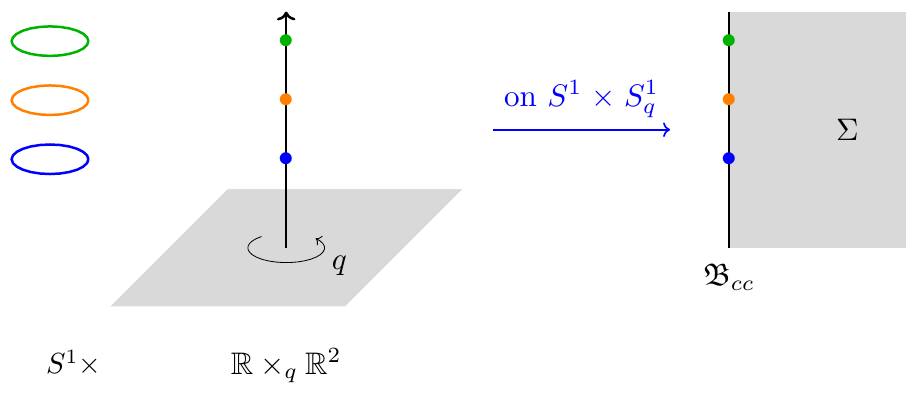}
    \caption{An algebra of line operators (colored circles) in a 4d $\cN=2$ theory becomes non-commutative in the $\Omega$-background $S^1 \times \bR \times_{q} \bR^2$, which provides deformation quantization of the holomorphic coordinate ring of the Coulomb branch. The 4d $\cN=2$ theory compactified on $S^1 \times S_q^1$ is described by 2d $A$-model $\Sigma \rightarrow \mathcal{M}_C$ on the Coulomb branch where the boundary condition at $\partial \Sigma$ is given by $\Bcc$. Here $\bR^2 \supset S_q^1$ is the circle generating the $\Omega$-deformation.}
    \label{fig:omega-background}
\end{figure}

\subsection{Brane quantization}\label{sec:brane-quantization}

As discussed above, the spherical DAHA $\SH$ arises from the deformation quantization of the coordinate ring of
\be \label{target}
\X = \MF(C_{0,4}, \SL(2,\bC))
\ee
with respect to the holomorphic symplectic form $\Omega_J$.  This space is both an affine variety and a \HK manifold. These properties naturally place it into the framework of brane quantization \cite{Gukov:2008ve} in a 2d sigma-model with target space $\X$. Brane quantization provides a geometric approach to the representation theory of $\SH$, which is the main focus of this paper. We will briefly review the brane quantization method here while referring the reader to sections 2.3 and 2.4 of \cite{Gukov:2022gei} for a more detailed treatment in the context of DAHA.

We consider the topological $A$-model on a symplectic manifold $(\X, \omega_\X)$, where quantization is achieved via open strings in the $A$-model. Brane quantization incorporates both deformation quantization and geometric quantization simultaneously, naturally providing the algebra and its representation, respectively. The boundary conditions of open strings are determined by geometric data in the target space, known as $A$-branes.  Depending on whether one is dealing with deformation quantization or geometric quantization, one considers two types of $A$-brane: the canonical coisotropic brane $\Bcc$ and the Lagrangian branes $\brane_\bfL$. 

The deformation quantization, which provides the algebra, is achieved by the canonical coisotropic brane, specifically via open $(\Bcc, \Bcc)$-strings. The canonical coisotropic brane, which can be figuratively described as the ``big $A$-brane,'' is a holomorphic line bundle over the target space itself:
\be
\label{Bcc}\Bcc:\quad
\begin{tikzcd}
\cL \arrow[d ] \\
\X
\end{tikzcd}  \qquad\qquad c_1(\cL)=[F/2\pi]\in H^2(\X,\Z)~.
\ee
As usual, the curvature $F$ forms a gauge invariant combination with the 2-form $B$-field in the 2d sigma-model, given by $F + B$, where
\be\label{B-field}
B \in H^2(\X, \U(1))~.
\ee 
The parameter of the deformation quantization can be written as
\be 
q = e^{2\pi i \hbar}, \qquad \hbar \in \bC^\times.
\ee 
The canonical coisotropic brane $\Bcc$ is parameterized by $\hbar$ on the symplectic manifold $(\X, \omega_\X)$, with the following structure:
\be\label{Bcc-Omega}
\Omega := F + B + i\omega_\X = \frac{\Omega_J}{i\hbar}~.
\ee
At a generic value of $\hbar = |\hbar| e^{i\theta}$, we can express the real and imaginary parts of $\Omega$ as
\be 
\begin{aligned}\label{generic-Bcc}
F + B &= \Re\Omega = \frac{1}{|\hbar|} (\omega_I \cos\theta - \omega_K \sin\theta)~, \cr
\omega_\X &= \Im\Omega = -\frac{1}{|\hbar|} (\omega_I \sin\theta + \omega_K \cos\theta)~,
\end{aligned}
\ee
so the symplectic form $\omega_\X$ depends on the value of $\hbar$.
Thanks to the \HK structure $\omega_{\X}^{-1}(F+B)=J$,  the brane automatically satisfies the condition $(\omega_\X^{-1}(B + F))^2 = -1$, which is required for a coisotropic $A$-brane \cite{Kapustin:2001ij}. With this setup, the space of open $(\Bcc, \Bcc)$-strings gives rise to the deformation quantization of the coordinate ring on $\X$, holomorphic in $J$.

To see the connection between the 4d $\cN=2$ theory and the 2d sigma-model, we compactify the 4d $\cN=2$ theory on $T^2 \cong S^1 \times S^1_q$, as illustrated in Figure \ref{fig:omega-background}. This compactification yields a 2d sigma-model $\bR \times \bR_+ \cong \Sigma \to \MF(C_{0,4}, \SL(2,\bC))$ on the Coulomb branch. Here, $S^1_q \subset \bR^2$ is a circle that encircles the axis of the $\Omega$-background, where the loop operators intersect. By the state-operator correspondence, loop operators in the 4d $\cN=2$ theory map to states in $\Hom(\Bcc, \Bcc)$. Thus, upon the compactification, the canonical coisotropic brane condition $\Bcc$ naturally emerges from the ``axis of the $\Omega$-deformation'' $\partial\Sigma$ (or the tip of the cigar, as described in \cite{Nekrasov:2010ka}).

The study on the representation theory of $\OO^q(\X)$ in the context of the 2d $A$-model originates from a simple idea: given an $A$-brane boundary condition $\brane'$, the space of open strings between $\Bcc$ and $\brane'$ forms a vector space $\Hom(\brane',\Bcc)$. As illustrated on the right side of Figure~\ref{fig:Bcc-algebra-rep}, the joining of $(\Bcc,\Bcc)$ and $(\Bcc,\brane')$-string produces another $(\Bcc,\brane')$-string. This implies that the space of $(\Bcc,\brane')$-strings receives an action of the algebra of $(\Bcc,\Bcc)$ strings \cite{Gukov:2008ve}. In other words, $A$-branes $\brane'$ on~$\X$ correspond to representations of $\OO^q(\X)$:
\be\label{HHrep}
\begin{array}{ccc}
 \OO^q(\X) & = & \Hom(\Bcc,\Bcc) \\
 \rotatebox[origin=c]{180}{$\circlearrowright$} & & \rotatebox[origin=c]{180}{$\circlearrowright$}  \\
 \repB' & = & \Hom(\brane',\Bcc)~.
\end{array}
\ee
In other words, brane quantization naturally proposes a derived functor
\be \label{eq:functor}
\RHom(-,\Bcc): D^b\ABrane(\MS,\omega_\MS) \to D^b \Rep(\OO^q(\MS))~,
\ee 
which conjecturally provides a derived equivalence between the category of $A$-branes and the derived category of $\OO^q(\MS)$-modules.

\begin{figure}[ht]
	 \centering
		 \includegraphics[width=4.5cm]{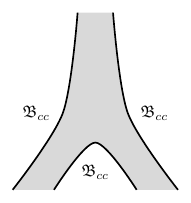} \hspace{3cm}
		 \includegraphics[width=4.5cm]{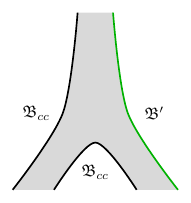}
 \caption{(Left) Open strings that start and end on the same brane $\Bcc$ form an algebra. \\ (Right) Joining a $(\Bcc,\Bcc)$-string with a $(\Bcc,\brane')$-string leads to another  $(\Bcc,\brane')$-string.}
 \label{fig:Bcc-algebra-rep}
\end{figure}

As the target symplectic manifold $(\X,\omega_\X)$ is of quaternionic dimension one, $A$-branes of the other types are all Lagrangian $A$-branes, namely it has a Lagrangian submanifold $\bfL$ as its support, endowed with a flat $\text{Spin}^{c}$-bundle:
\begin{equation}\label{Lag-brane}
\brane_\bfL:\begin{gathered}
\mathcal{L}^{\prime} \otimes K_{\bfL}^{-1 / 2} \\
\downarrow \\
\bfL
\end{gathered}
\end{equation}
where $K_{\bfL}^{-1 / 2}$ is a square root of the canonical bundle of $\bfL$, which gives rise to a $\text{Spin}^{c}$ structure when $\mathcal{L}^{\prime}$ is not a genuine line bundle. (See \cite{Gukov:2008ve,Gukov:2022gei} for more details.) The subtlety of $\text{Spin}^{c}$ structures appears only when we consider bound states of $A$-branes, and both $\cL'$ and $K_\L^{-1/2}$ exist as genuine line bundles in most of the examples in this paper since all the Lagrangian submanifolds considered are of real two dimensions. Additionally, a Lagrangian $A$-brane must satisfy the flatness condition 
\begin{equation}\label{deformed-flat}
    F^{\prime}+B|_{\bfL}=0~,~
\end{equation}
where $F^{\prime}$ is the curvature of $\cL'$, and $F^{\prime}+B$ is the gauge-invariant combination as before.

Furthermore, a Lagrangian $A$-brane carries a natural \emph{grading datum} determined by the \emph{Maslov index}~\cite{seidel2000graded}. Since gradings play a role in our discussion, let us briefly review the definition of the Maslov index.  

The \emph{Lagrangian Grassmannian} $\LGr(2n)$ of the standard symplectic vector space $(\R^{2n},\omega)$ is the set of all Lagrangian subspaces of $\R^{2n}$. The subgroup $\U(n)\subset \mathrm{GL}(2n,\R)$ preserves both the symplectic and orthogonal structures, and it acts transitively on $\LGr(2n)$. The stabilizer of a fixed Lagrangian subspace is $\Or(n)$, so $\LGr(2n)$ can be identified with a homogeneous space
\begin{equation}
  \LGr(2n) = \U(n)/\Or(n)~.
\end{equation}
There is a natural map
\begin{equation}\label{det2}
  {\det}^2 : \LGr(2n) \to \U(1)
\end{equation}
whose induces an isomorphism on the fundamental groups. The \emph{Maslov index}~\cite{ArnoldMaslov} of a loop in $\LGr(2n)$ is defined as its image in $\pi_1(\U(1)) \cong \Z$ under this map. The universal cover $\tLGr(2n)$ has deck transformation group $\Z$, and the Maslov index of a loop records the $\Z$-valued displacement of a lift of the loop to $\tLGr(2n)$.

Given a symplectic manifold $(\X,\omega_\X)$, one can assemble these constructions into the Lagrangian Grassmannian bundle
\begin{equation}
  \LGr(\X) \to \X
\end{equation}
whose fiber over $x \in \X$ is $\LGr(T_x\X)$.  
We similarly define a bundle $\tLGr(\X)$ as a covering space of $\LGr(\X)$ such that the projection map is fiberwise the universal covering map.

For a Lagrangian submanifold $\L \subset \X$, the tangent spaces define a Gauss map
\be \label{Gauss}
g_\L: \L \hookrightarrow \LGr(\X), \quad x \mapsto T_x \L \in \LGr(T_x \X) ~.
\ee 
Using \eqref{det2}, we can pull back the generator $d\theta \in H^1(\U(1),\Z)$ along the composition ${\det}^2 \circ g_\L$, obtaining a cohomology class
\[
  \mu_\L := ({\det}^2 \circ g_\L)^* d\theta \;\in\; H^1(\L,\Z),
\]
called the \emph{Maslov class} of $\L$.  

In fact, the Maslov class $\mu_\L$ precisely measures the obstruction to lifting the Gauss map $g_\L$ to $\tLGr(\X)$. If $\mu_\L = 0$, then $\L$ admits a \emph{grading}, defined by a lift $g : \L \to \tLGr(\X)$ making the following diagram commute:
\begin{equation}\label{grading}
\begin{tikzcd}
\tLGr(\X) \arrow{r}{\cdot/\Z} & \LGr(\X) \ar[d] \\
\L \arrow[dashed]{u}{g} \arrow{r}{\subset} \arrow[hook]{ru}{} & \X
\end{tikzcd}
\end{equation}
The set of such lifts forms a $\Z$-torsor under deck transformations. In particular, there is no canonical choice of grading: different lifts differ by integer shifts, mirroring the shift functor in the $A$-brane category. Thus, given a Lagrangian object in $\ABrane(\X,\omega_\X)$, the choice of graded lift encodes its possible shifts.

Now, we consider the morphism space in the $A$-brane category. 
The space of $(\Bcc,\brane_\bfL)$-open string arises from the geometric quantization of $\bfL$, namely the space of holomorphic sections $\Bcc\otimes\brane_\bfL^{-1}$ over $\bfL$. Hence, when the support is a compact Lagrangian submanifold, one can employ     the $B$-model perspective to compute the dimension of the representation space $\Hom(\brane_\bfL,\Bcc)$ using Hirzebruch-Riemann-Roch formula:
\begin{equation}
\begin{aligned}
\dim \scL &=\dim \Hom(\brane_\bfL,\Bcc)\cr 
&=\dim H^0(\bfL, \Bcc \otimes \brane_{\bfL}^{-1}) \\
&=\int_{\bfL} \operatorname{ch}(\Bcc) \wedge \operatorname{ch}(\brane_{\bfL}^{-1}) \wedge \operatorname{Td}(T \bfL)~,
\end{aligned}
\end{equation}
where $\operatorname{Td}(T \bfL)$ is the Todd class of the tangent bundle of $\bfL$. Since $\bfL$ is real two-dimensional in our case, the Todd class $\operatorname{Td}(T \bfL)$ is equal to $\operatorname{ch}(K_{\bfL}^{-1/2})$.  Thus, the dimension formula is simplified to
\begin{equation}\label{dimformula}
    \dim \scL  = \int_{\bfL} \operatorname{ch}(\Bcc)= \int_{\bfL} \frac{F+B}{2\pi}~,
\end{equation}
for a real two-dimensional Lagrangian $\bfL$. In this way, we can explicitly provide the dimension of a finite-dimensional representation corresponding to a Lagrangian $A$-brane with compact support. 

The main goal of this paper is to study the representation theory of DAHA of $C^\vee C_1$ using brane quantization. This approach gives a geometric perspective on $\SH$-modules. Moreover, as brane quantization is conjectured to establish derived equivalence on the categories through \eqref{eq:functor}, it offers new insights into categorical structures, including the auto-equivalence group of these categories. To study the $A$-model in the character variety $\X$, we proceed to examine the geometry of $\X$ in gory detail.

\subsection{Affine Weyl group action via Picard-Lefschetz}\label{sec:wallcrossing}

The affine cubic surface \eqref{cubic_eqn} has been the subject of study since the 19th century, with a history spanning over a century \cite{vogt1889,fricke_klein_1897}.
Notably, the character variety $\X=\MF(C_{0,4}, \SL(2,\bC))$ naturally arises in the isomonodromic deformation problem of Painlev\'e VI equation \cite{Jimbo1982}.
(See \cite{boalch2005klein,mazzocco2018embedding,Crampe:2020abq} and references therein for more details.) 
The non-abelian Hodge correspondence identifies $\X$ with the corresponding moduli space of Higgs bundles. Following developments in Seiberg-Witten theory \cite{Seiberg:1994aj}, this space has been intensively investigated even from physics perspective. Nonetheless, this paper further investigates $\X$, uncovering its new aspects. As demonstrated in \cite{Gukov:2022gei}, Hitchin fibration is the key to understanding compact Lagrangian submanifolds in $\X$. Moreover, the remainder of this subsection aims to reveal intrinsic connections between the geometry of $\X$ and the root system of type $D_4$.

The root system of type $D_4$ appears naturally in the classification of du Val singularities of the character variety $\X$. For specific values of the monodromy parameters $\boldsymbol{t}$, the character variety $\X$ develops du Val singularities (a.k.a. $ADE$ singularities) around which the space is locally $\bC^2/\Gamma_{ADE}$ with $\Gamma_{ADE}$ being a finite subgroup of $\SU(2)$ classified by the $ADE$ types.  These singularities arise when the discriminant of the cubic equation  \eqref{cubic_eqn} vanishes \cite{iwasaki2002modular}:
\begin{equation}\label{discriminant_cubic_surface}
\Delta(\boldsymbol{t})=\left[\prod_{\epsilon_1\epsilon_2\epsilon_3=1}\left(\overline{t_4}+\sum_{j=1}^3\epsilon_j\overline{t_j}\right)-\prod_{j=1}^3(\overline{t_j} \ \overline{t_4}-\overline{t_k} \ \overline{t_l})\right]^2\prod_{j=1}^4(\overline{ t_j}^2-4)~,
\end{equation}
where $\epsilon_j=\pm 1$, and $(j,k,l)$ is the positive permutation of $(1,2,3)$ for any $j$. In fact, the discriminant can be expressed in a remarkably concise way, using the set $\sfR(D_4)$ of the $D_4$ roots:
\begin{equation}
   \Delta(\boldsymbol{t})=\prod_{r\in\sfR(D_4)} (\boldsymbol{t}^{r}-1)~,
\end{equation}
where $\boldsymbol{t}^{r}=t_1^{r_1}t_2^{r_2}t_3^{r_3}t_4^{r_4}$. (See Appendix \ref{app:notation} for the notation of the $D_4$ roots.) 
Furthermore, the specific $ADE$ type of each singularity can be identified by embedding the corresponding root system into $\sfR(D_4)$. Taking into account the multiplicity of singularities, a singularity of type $\frakg$ emerges when
\be \label{ADE-condition}
\boldsymbol{t}^{r} =1 ~, \qquad \forall r \in \sfR(\frakg) \hookrightarrow \sfR(D_4) 
\ee 
where $\frakg$ is a Lie subalgebra of $D_4$ so that $\sfR(\frakg)$ is the subset of $\sfR(D_4)$ \cite{Iwasaki2007FiniteBS}. In this manner, we can consistently classify the conditions and types of du Val singularities using the root system of type $D_4$, as summarized in Table \ref{tab:surface_sing_classification}.

\begin{table}[ht]
\centering
\renewcommand{\arraystretch}{1.3}
\begin{tabular}{c|c|c}
 Types& Conditions of $\boldsymbol{t}$ &Examples\\[.5em]
\hline 
 \multirow{1.4}*{$A_1$}&\multirow{1.4}*{$\boldsymbol{t}^{r}=1,\ \ \forall r\in \sfR(A_1)\hookrightarrow \sfR(D_4)$} &\multirow{1.4}*{$t_1^2=1$ \text{or} $t_1t_2t_3 t_4=1$}\\[.5em]
\hline
\multirow{1.4}*{$A_1^{\oplus2}$}&\multirow{1.4}*{$\boldsymbol{t}^{r}=1,\ \ \forall r\in \sfR(A_1^{\oplus2})\hookrightarrow \sfR(D_4)$} &\multirow{1.4}*{$t_1^2=t_2^2=1$}\\[.5em]
\hline
\multirow{1.4}*{$A_2$}&\multirow{1.4}*{$\boldsymbol{t}^{r}=1,\ \ \forall r\in \sfR(A_2)\hookrightarrow \sfR(D_4)$} &\multirow{1.4}*{$t_1^2=t_1t_2t_3t_4=1$}\\[.5em]
\hline
\multirow{1.4}*{$A_1^{\oplus 3}$}&\multirow{1.4}*{$\boldsymbol{t}^{r}=1,\ \ \forall r\in \sfR(A_1^{\oplus3})\hookrightarrow \sfR(D_4)$} &\multirow{1.4}*{$ t_1^2=t_1t_2t_3t_4=t_1t_2t_3t_4^{-1}=1$}\\[.5em]
\hline
 \multirow{1.4}*{$A_3$}&\multirow{1.4}*{$\boldsymbol{t}^{r}=1,\ \ \forall r\in \sfR(A_3)\hookrightarrow \sfR(D_4)$} &\multirow{1.4}*{$t_1^2= t_2^2=t_1t_2t_3 t_4=1$}\\[.5em]
\hline
 \multirow{1.4}*{$A_1^{\oplus 4}$}&\multirow{1.4}*{$\boldsymbol{t}^{r}=1,\ \ \forall r\in \sfR(A_1^{\oplus4})\hookrightarrow \sfR(D_4)$} &\multirow{1.4}*{$ t_1^2=t_2^2=t_3^2=t_4^2=1$,~ $t_1t_2t_3t_4=-1$}\\[.5em]
\hline
 \multirow{1.4}*{$D_4$}&\multirow{1.4}*{$\boldsymbol{t}^{r}=1,\ \ \forall r\in\sfR(D_4)$} &\multirow{1.4}*{$ t_1=t_2=t_3=t_4=1$}\\[.5em]\end{tabular}
 \caption{Classification of du Val singularities and the corresponding conditions on $\boldsymbol{t}$. Each row represents a type of singularity characterized by the embedding of a sub-root system into $\sfR(D_4)$. The second column specifies the conditions on $\boldsymbol{t}$, where $\boldsymbol{t}^r = 1$ holds for all roots $r$ in the respective sub-root system. The third column provides illustrative examples of these conditions. The type of the sub-root system matches that of the du Val singularity. } \label{tab:surface_sing_classification}
\end{table}

Certainly, there are more intricate connections between the geometry of $\X$ and the $D_4$ root system. 
In complex structure $I$, the Hitchin moduli space admits the Hitchin fibration over $\cB_H$ is an affine space, called the Hitchin base \cite{hitchin1987self}:
\be \label{Hitchin-fibration}
\begin{aligned}
    h:\cM_H(C_{0,4},\SU(2)) &\to \cB_H=H^{0}(C_{0,4},K_C^{\otimes2})~.\\
(E,\varphi)&\mapsto \Tr(\varphi^2)
\end{aligned}
\ee 
It is a completely integrable system so that generic fibers are abelian varieties (sometimes called ``Liouville tori'') and they are holomorphic Lagrangian with respect to $\Omega_I$, namely Lagrangian submanifolds of type $(B,A,A)$. For our case where $G=\SU(2)$ and the base curve is $C_{0,4}$, the Hitchin moduli space is of quaternionic dimension one, and a generic fiber $\bfF$ of the Hitchin fibration is topologically a two-torus. 
The Hitchin base $\cB_H$ is identified with the $u$-plane that parametrizes the complex structure $u$ of the Seiberg-Witten curve \eqref{SW-Nf4-generic}.

At special points of the $u$-plane where the Seiberg-Witten curve becomes singular, a Hitchin fiber degenerates into a singular fiber. The geometry of the moduli space undergoes significant changes as the ramification parameters \eqref{tame} at the four punctures are varied, leading to changes in the types of singular fibers. A classification of these singular fibers is presented in the next subsection.

Among physics literature, the geometry of the Hitchin moduli space for a curve with tame ramifications has been studied in great detail in \cite{Gukov:2006jk}. One of the key insights there is the action of the affine Weyl group on the second (co)homology group of the Hitchin moduli space. 
Consider the case where all ramification points correspond to full punctures, and the triples $(\alpha_j, \beta_j, \gamma_j) \in (\mathfrak{t}_j)^3$ take generic values. Following the notation of \cite{Gukov:2006jk}, we denote the corresponding moduli space as $\MH(\alpha_1, \beta_1, \gamma_1; \ldots; \alpha_s, \beta_s, \gamma_s;G)$. In this setting, the second integral homology group can be identified with the (generalized) affine root lattice:
\be
H_2(\MH(\alpha_1, \beta_1, \gamma_1; \ldots; \alpha_s, \beta_s, \gamma_s;G), \bZ) \cong \bZ\langle e\rangle \oplus \bigoplus_{j=1}^s \sfQ_j(\frakg)~,
\ee
where $e$ is the generator of the second homology group $H_2(\bfV, \bZ) \cong \bZ\langle e\rangle$ of the moduli space $\bfV$ of $G$-bundles.
Consequently, the (generalized) affine Weyl group $\AWeyl^s \cong (\Weyl \ltimes \mathsf{t}(\sfQ^\vee))^s$ acts naturally on the second homology group. Here, $\sfQ^\vee$ represents the coroot lattice, which acts via affine translations $\mathsf{t}$, given explicitly by:
\be\label{affine-translation}
\mathsf{t}(r_j^\vee)(\tilde{r}_j) = \tilde{r}_j - (\tilde{r}_j, r_j^\vee)e~, \quad j \in \{1,\ldots,s\}~.
\ee

In our case, where $G = \SU(2)$ and the Riemann surface is a four-punctured sphere $C_{0,4}$, it is well known that the flavor symmetry is enhanced from $\SU(2)^4$ to the $\SO(8)$ group \cite{Seiberg:1994aj}. As a result, the space of ramification parameters $(\alpha_j, \beta_j, \gamma_j) \in (\mathfrak{t}_j)^3$ at the four punctures, where each parameter takes values in the Cartan subalgebra $\mathfrak{t}_j$ of $\SU(2)$, can be naturally identified with the Cartan subalgebra of the Lie algebra of type $D_4$:
\be 
\prod_{j=1}^4 \mathfrak{t}_j^{(\alpha)}\times  \mathfrak{t}_j^{(\beta)} \times  \mathfrak{t}_j^{(\gamma)}\cong \mathfrak{t}_{D_4}^{(\alpha)}\times\mathfrak{t}_{D_4}^{(\beta)}\times\mathfrak{t}_{D_4}^{(\gamma)} ~ .
\ee 
Moreover, the second integral homology group of the moduli space can be identified with the affine root lattice of type $D_4$:
\be\label{Homology_Lattice}
\bZ\langle e\rangle \oplus \bigoplus_{j=1}^4 \sfQ_j(A_1) \cong \dt{\sfQ}(D_4)\cong H_2(\MH(C_{0,4},\SU(2)),\bZ)~,
\ee 
on which the affine Weyl group of type $D_4$ acts:
\be 
\dt{W}(D_4) \cong W(D_4) \ltimes \mathsf{t}(\sfQ^\vee_{D_4})~.
\ee 
The affine Weyl group also acts on the ramification parameters $(\alpha,\beta,\gamma)$
\be \label{dtWD4onabc}
\frac{\mathfrak{t}_{D_4}^{(\alpha)}\times\mathfrak{t}_{D_4}^{(\beta)}\times\mathfrak{t}_{D_4}^{(\gamma)}}{\dt{W}(D_4)}\cong \frac{T_{D_4}^{(\alpha)}\times\mathfrak{t}_{D_4}^{(\beta)}\times\mathfrak{t}_{D_4}^{(\gamma)} }{W(D_4)},
\ee 
where the affine translation $\mathsf{t}(\sfQ^\vee_{D_4})$ acts non-trivially only on the $\alpha$-space. The quotient of the affine translation \eqref{affine-translation} can be identified with the periodicity 
\be 
\talpha_j \sim \talpha_j +1~, 
\ee 
due to \eqref{tj}. As we will see below, a point of the space $T_{D_4} \times \mathfrak{t}_{D_4}\times \mathfrak{t}_{D_4}$ represents the cohomology classes of the \K forms in the second cohomology group $H^2(\MH,\bR)$ of the corresponding moduli space:
\be 
\left(\left[\frac{\omega_I}{2\pi}\right],\left[\frac{\omega_J}{2\pi}\right],\left[\frac{\omega_K}{2\pi}\right]\right)\in T_{D_4} \times \mathfrak{t}_{D_4}\times \mathfrak{t}_{D_4}~.
\ee

In fact, the connection between the affine Weyl group of type $D_4$ and this system also has an extensive history. It was first discovered in \cite{okamoto1986studies} that a group of Bäcklund transformations of Painlevé VI equation is isomorphic to $\dt{W}(D_4)$. For our purposes, the modern treatment of the relation between the geometry of the cubic surface \eqref{cubic_eqn} and the affine Weyl group $\dt{W}(D_4)$ is provided in \cite{Iwasaki2007FiniteBS}.

\begin{figure}[ht]
    \centering    \includegraphics{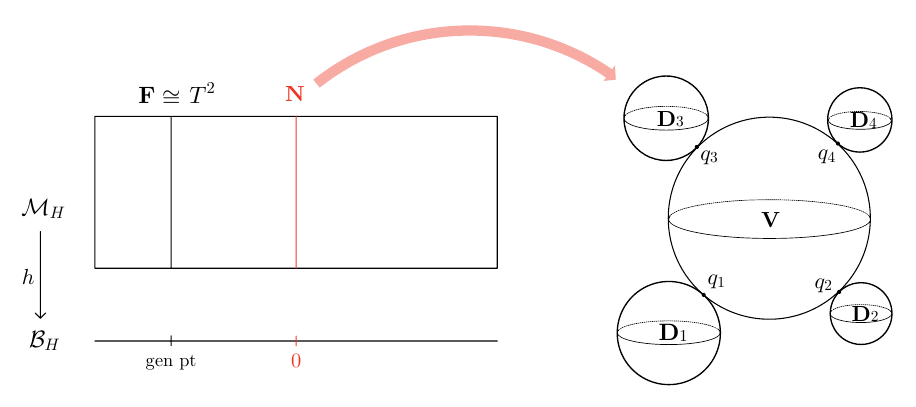}
    \caption{The schematic figure of the Hitchin fibration $\MH\to \cB_H$ when the ramification parameters $\talpha_j$ are generic, and the others are zero $\tbeta_j=\tgamma_j=0$, $(j=1,2,3,4)$. A generic fiber $\bfF$ is topologically a two-torus, and the global nilpotent cone $\bfN$ at the origin of the Hitchin base $\cB_H$ is a singular fiber of Kodaira type $I_0^*$, which is illustrated on the right. }
    \label{fig:nilpotent_cone}
\end{figure}

To understand the action of the affine Weyl group $\dt W(D_4)$, let us examine the geometry of the Hitchin moduli space under the condition that the ramification parameters $\talpha_j$ are generic while the others are set to zero, $\tbeta_j = \tgamma_j = 0$ $(j = 1, 2, 3, 4)$. Under these assumptions, the Hitchin fibration $h:\MH(C_p, \SU(2))\to\cB_H$ in \eqref{Hitchin-fibration} possesses a singular fiber only at the origin of the base, $\bfN=h^{-1}(0)$, which is referred to as the \emph{global nilpotent cone}. 
The global nilpotent cone $\bfN$ is a singular fiber of Kodaira type $I_0^*$ \cite{Gukov:2007ck}, characterized by a configuration of five irreducible components, each of which is topologically $\bC\bP^1$, arranged in the shape of the affine $D_4$ Dynkin diagram:
\be
\bfN = \bfV \cup \bigcup_{j=1}^4 \bfD_j~,
\ee
where $\bfV$ represents the moduli space of $\SU(2)$-bundles on the four-punctured sphere $C_{0,4}$, and $\bfD_j$ $(j = 1, 2, 3, 4)$ are known as the exceptional divisors. (See Figure \ref{fig:nilpotent_cone}.)

At generic values of $\talpha_j$, each irreducible component of the global nilpotent cone serves as a generator of the second integral homology group, $H_2(\MH, \bZ) \cong \bZ^{\oplus 5}$. Using the basis $\{[\bfD_1], [\bfD_2], [\bfD_3], [\bfD_4], [\bfV]\}$, the intersection pairing between these homology classes is represented by the Cartan matrix of the affine $D_4$ Dynkin diagram, up to an overall minus sign:
\be \label{intersection-form}
Q = 
\begin{pmatrix}
-2 & 0 & 0 & 0 & 1 \\
0 & -2 & 0 & 0 & 1 \\
0 & 0 & -2 & 0 & 1 \\
0 & 0 & 0 & -2 & 1 \\
1 & 1 & 1 & 1 & -2
\end{pmatrix}.
\ee
With this bilinear pairing defined by the intersection form, the second homology group $H_2(\MH, \bZ)$ can be identified with the root lattice of the affine $D_4$ Lie algebra. In this identification, the irreducible components correspond to the simple roots of the affine $D_4$ root system. Furthermore,
these components also satisfy a fiber-class relation in $H_2(\MH,\bZ)$:
\begin{equation}\label{homology-relation}
    [\bfF] = 2[\bfV] + \sum_{j=1}^4 [\bfD_j]~,
\end{equation}
where $[\bfF]$ denotes the homology class of a generic fiber of the Hitchin fibration \cite{Gukov:2007ck}.

\begin{figure}[ht]
    \centering
    \includegraphics{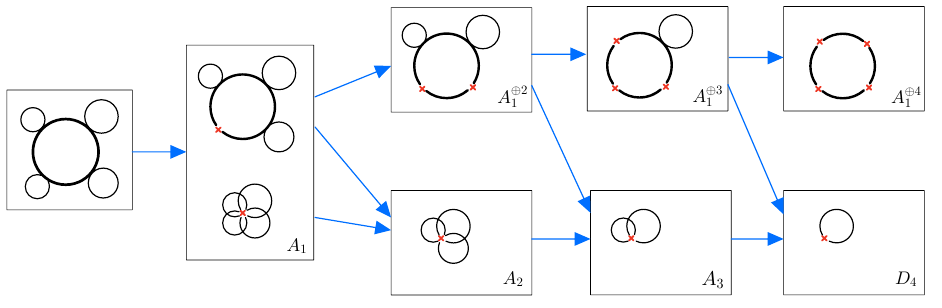}
    \caption{When $\tbeta_j = \tgamma_j = 0$, du Val singularities emerge as the volumes of $\bfV$ or $\bfD_j$ shrink to zero. These vanishing volumes are indicated by red crosses on the corresponding diagrams. The blue arrows represent transitions where a single volume contracts to zero, leading to successive enhancements of the singularity type. The progression illustrates the hierarchy of du Val singularities, starting from $A_1$ and culminating in higher types such as $A_1^{\oplus 4}$ and $D_4$.}
    \label{fig:du_Val}
\end{figure}

In complex structure $I$, the ramification parameters $\talpha_j$ serve as the \K parameter of the Hitchin moduli space. Since each Hitchin fiber is a holomorphic Lagrangian of type $(B, A, A)$, variations in $\talpha_j$ correspond to changes in the ``sizes'' or volumes of the irreducible components of the global nilpotent cone while the Kodaira type remains the same. 
 For our applications to branes and representations, it is essential to identify the relation between the \K moduli parameter $\talpha_j$ and the periods of the \K forms over the five compact two-cycles:
\be 
\textrm{vol}_I(\bfV):=\int_{\bfV}\frac{\omega_I}{2\pi}~, \qquad \textrm{vol}_I(\bfD_j)=\int_{\bfD_j}\frac{\omega_I}{2\pi}~.
\ee 
As a characteristic of the Hitchin fibration, due to its complete integrability, the period of $\omega_I/2\pi$ over a general fiber $\bfF$ is one: 
\be \label{Fibervolume}
\textrm{vol}_I(\bfF)=\int_{\bfF}\frac{\omega_I}{2\pi}=1~.
\ee
As a result of the relation \eqref{homology-relation}, we have the constraint
\begin{equation}\label{VolumeRelation}
    2\textrm{vol}_I(\bfV)+\sum_{j=1}^4\textrm{vol}_I(\bfD_j)= 1~.
\end{equation}

The relation between the ramification parameters $\talpha_j$ and the volumes of the irreducible components in the global nilpotent cone is subtle and involves the wall-crossing phenomenon. 
As illustrated in Figure \ref{fig:du_Val}, when an irreducible component $\bfV$ or $\bfD_j$ shrinks to zero, the Higgs bundle is no longer stable, and the Hitchin moduli space develops a du Val singularity.  These singularities correspond to special points in the \K moduli space, which are located at codimension-one loci, referred to as walls. As seen in \eqref{discriminant_cubic_surface}, du Val singularities appear precisely at the zeros of the discriminant
\be \label{wall}
 \boldsymbol{t}^{r}=1~,\quad\forall r\in\sfR(D_4)~.
\ee 
Since we are now considering the case $\tgamma_j = 0$, with $t_j = e^{-2\pi (\tgamma_j + i\talpha_j)}$, the locations of the walls \eqref{wall} in the \K moduli space are therefore specified by
\be\label{wallalpha}
\sum_{j=1}^4 r_j \boldsymbol{\talpha}_j=n,\qquad\forall r\in\sfR(D_4),\quad \forall n\in\bZ~.
\ee

These walls align with the set of reflection hyperplanes for the affine $D_4$ weight lattice at level one \cite{Iwasaki2007FiniteBS}. In an affine root system, the reflection of an affine weight $\lambda$ at level one with respect to an affine root $\dt{r} = r - n\delta$
\be 
s_{\dt{r}}(\lambda) = s_r\left(\lambda - \frac{n}{2}r\right)
\ee
corresponds to a reflection over the hyperplane specified by \eqref{wallalpha}, generating the affine Weyl group $\dt W(D_4)$.
Therefore, a chamber surrounded by the walls in \eqref{wallalpha} corresponds to a Weyl alcove of type $D_4$.

A du Val singularity arises when a two-cycle shrinks to zero volume, whose condition \eqref{wallalpha} is linear in the \K parameters $\talpha_j$. This implies that the volume of a compact two-cycle depends linearly on $\talpha_j$.
 This linearity can be numerically verified by the integration of $\omega_I/2\pi$ over $\bfV$ in the relevant domain where the variables $x,y,z$ in the cubic surface \eqref{cubic_eqn} are real. Starting from this observation, we can make some statements about the volumes of the cycle.

Assuming that the volumes of compact two-cycles are linear in $\talpha_j$, the configurations of the walls in \eqref{wall}, along with the normalization in \eqref{Fibervolume} uniquely determine the volumes of the cycles. Given a point $\talpha_j$ of the \K moduli space, the volumes of the two-cycles $\bfV$ and $\bfD_j$ ($j=1,2,3,4$) are equal to twice the distance from this point to the corresponding walls, as illustrated in Figure \ref{fig:dis_wall}. The factor of two arises from the length of a root $r\in \sfR(D_4)$. 
Since these walls can be identified with reflection hyperplanes associated with the affine $D_4$ Weyl group, the set of the volumes defines a basis $(e^0,\ldots,e^4)$ of the affine $D_4$ root system. Given a point $\boldsymbol{\alpha}$ in the Weyl alcove, the volume corresponding to each basis vector $e^i$ is computed as the inner product $e^i \cdot \boldsymbol{\alpha}$; see \cite[Chapter 1.2]{macdonald2003affine} for further details.

\begin{figure}[ht]
    \centering
    \includegraphics[width=0.55\linewidth]{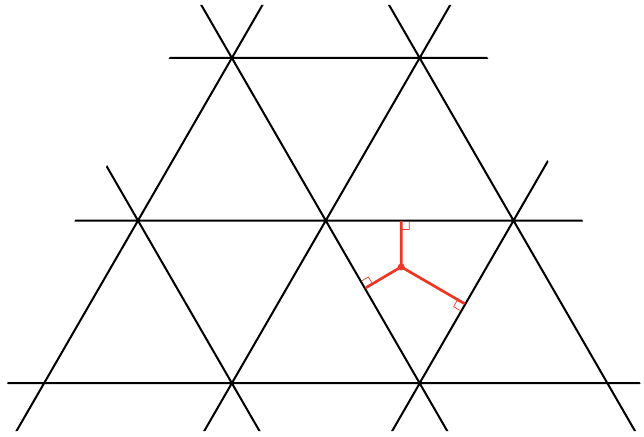}
    \caption{A schematic illustration of the volumes of compact two-cycles in the \K moduli space. The black lines represent reflection hyperplanes, and each triangle corresponds to a Weyl alcove, which is a chamber in the moduli space. The volumes of the two-cycles are twice the distances (depicted by the lengths of the red lines) from a given parameter point to the walls of the chamber containing it.}
    \label{fig:dis_wall}
\end{figure}

This identification immediately indicates a wall-crossing phenomenon of the volume functions in the \K moduli space. Whenever the $\talpha_j$ parameters cross the wall defined by the vanishing loci of an affine root $r$, the basis of roots is transformed by the Weyl reflection $s_{r}$. 
\begin{equation}\label{Weyl_Reflection}
s_{r_a}(r_b)= r_b- \dt A_{ba} r_a 
\end{equation}
where $\dt A_{ba}$ is the affine $ D_4$ Cartan matrix. As a result, the volume functions exhibit a discontinuity, with their dependence on the \K parameters $\talpha_j$ jumping by the affine Weyl reflection. The detailed study of the chamber structures and the appearance of du Val singularities is presented in Appendix \ref{app:chamber-structure}.

\begin{figure}[ht]
    \centering
    \includegraphics[width=0.9\linewidth]{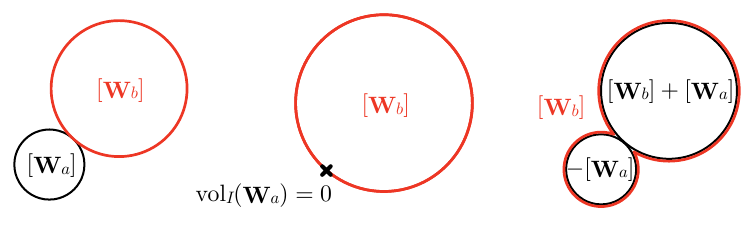}
    \caption{When crossing the wall defined by $\text{vol}_I(\bfW_a)=0$, there is a basis transformation of the second homology, governed by the Picard-Lefschetz monodromy transformation $T_{\bfW_a}$ defined in \eqref{mPL}. This figure illustrates the action of $T_{\bfW_a}$ on the homology classes $[\bfW_a]$ and $[\bfW_b]$ when their intersection number satisfies $[\bfW_a]\cdot[\bfW_b]=1$ where $[\bfW_a]$ is inverted and $[\bfW_b]$ shifts by $[\bfW_a]$. The middle diagram represents the configuration at the moment of wall-crossing, where the volume of $\bfW_a$ vanishes. This transformation reflects the geometric change in the homology basis due to the vanishing cycle.}
    \label{fig:PL}
\end{figure}

This wall-crossing phenomenon admits geometric interpretation. 
When crossing a wall defined by $\text{vol}_I(\bfW_a) = 0$ for $[\bfW_a]\in H_2(\MH,\bZ)$, the Picard-Lefschetz (PL) monodromy transformation  \cite{Seidel1999,Seidel2008} changes a base of the second homology due to the vanishing cycle $[\bfW_a]$ as
\begin{equation}\label{mPL}
T_{\bfW_a}([\bfW_b])= [\bfW_b]+([\bfW_a] \cdot [\bfW_b])[\bfW_a]~,\qquad [\bfW_b]\in H_2(\MH,\bZ)~.
\end{equation}
The sign difference from \eqref{Weyl_Reflection} results from the sign difference between the intersection form \eqref{intersection-form} and the affine Cartan matrix $\dt A_{ba}$. 
Therefore, writing $T_j \equiv T_{\bfD_j}$ for $(j = 1, 2, 3, 4)$, the PL transformation geometrically realizes the affine Weyl reflection in \eqref{Weyl_Reflection} in the second homology classes, subject to the following relations:
\be\label{affine-Weyl-D4}
\begin{aligned}
    T_j T_k &= T_k T_j~, \\
    T_j T_{\bfV} T_j &= T_{\bfV} T_j T_{\bfV}~, \\
    T_j^2 &= T_{\bfV}^2 = \operatorname{id}~,
\end{aligned}
\ee
for any pair $(j, k = 1, 2, 3, 4)$.
 In conclusion, the PL transformation gives a concrete example of the affine Weyl group action on the second integral homology of the Hitchin moduli space \cite{Gukov:2006jk}.

The remainder of this subsection focuses on visualizing the wall inside the \K moduli space, specifying a chamber, and explicitly writing down the volume function. Using the periodicity $\talpha_j \to \talpha_j + 1$, we can restrict $\talpha_j$ to the range $[-\frac{1}{2}, \frac{1}{2}]$.  In addition, the Weyl group invariance $\talpha_j\to-\talpha_j$ of the cubic equation implies that parameter space is symmetric with respect to four $\talpha_j=0$ walls.
In total, all the walls of \eqref{wall} divide the parameter space into $24\times 2^5$ chambers, where the factor $2^5$ corresponds to the sign changes of $\talpha_j$.  Consequently, it suffices to study the region where $\talpha_j$ is restricted to $[0, \frac{1}{2}]$, which contains 24 chambers. This subset of the parameter space forms a 4-dimensional hypercube, which we denote as $\mathsf{Cube}$.

\begin{figure}[ht]
    \centering
    \includegraphics[width=7.5cm]{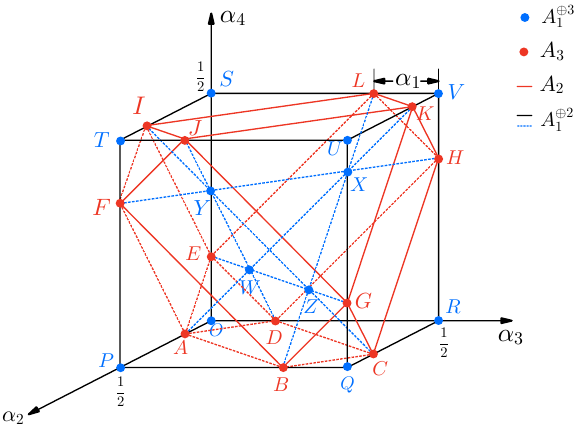}
    \hspace{1cm}
		 \includegraphics[width=7cm]{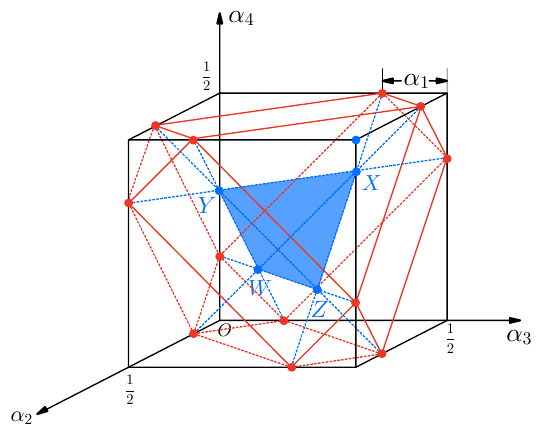}
    
    \caption{(Left): The cross section of the \K moduli space $\mathsf{Cube}$ at a fixed $\talpha_1\in[0,\frac{1}{4}]$, and its chamber structure. $W$ is at $(\talpha_1,\talpha_1,\talpha_1)$ and $X,Y,Z$ have one coordinate $\talpha_1$ and two $\frac{1}{2}-\talpha_1$; $A,D,E$ have two coordinates 0 one  $\talpha_1$ and $B,C,F,G,H,I,J,K,L$ have one coordinate $0$, one $\frac{1}{2}$ and one $\frac{1}{2}-\talpha_1$. Every plane in the 3-cube corresponds to a wall where an $A_1$ type singularity develops, and other singularities are marked out. The $A_1^{\oplus4}$ and $D_4$ singularities can be seen only for specific $\talpha_1$, which is presented in Figure \ref{fig:chambers_special_alpha1}.\\
    (Right):  The center chamber with $\talpha_1\in[0,\frac{1}{4}]$, shaded in blue, is the tetrahedron $WXYZ$ determined by the constraints \eqref{wall_center}.}
    \label{fig:chambers}
\end{figure}

To visualize the 24 chambers in $\mathsf{Cube}$, we can fix the value of $\talpha_1$ and examine a 3-dimensional cross-section of the hypercube, represented as a 3-cube in Figure \ref{fig:chambers}. This 3-cube is subdivided into 23 distinct regions, each corresponding to a unique chamber in $\mathsf{Cube}$. The mapping between these regions and the chambers is injective, meaning no two regions correspond to the same chamber. However, one chamber remains ``invisible'' in this visualization. If we instead cut along $\frac{1}{2} - \talpha_1$, the central chamber $WXYZ$ is replaced by the previously hidden chamber. A detailed discussion of the chamber structure can be found in Appendix \ref{app:chamber-structure}.

In the subsequent analysis, we focus on the central chamber $WXYZ$, as depicted in Figure \ref{fig:chambers}, with $\talpha_1 \in [0, \frac{1}{4}]$. The other chambers can be analyzed similarly. The central chamber is the region defined by the following constraints:
\begin{equation}\label{wall_center}
    \begin{aligned}
        \talpha_1 + \talpha_2 + \talpha_3 + \talpha_4 &\leq 1~, \\
        -\talpha_1 - \talpha_2 + \talpha_3 + \talpha_4 &\geq 0~, \\
        -\talpha_1 + \talpha_2 - \talpha_3 + \talpha_4 &\geq 0~, \\
        -\talpha_1 + \talpha_2 + \talpha_3 - \talpha_4 &\geq 0~,
    \end{aligned}
\end{equation}
together with the restriction $\talpha_1 \in [0, \frac{1}{4}]$. The four constraints in \eqref{wall_center} correspond to the four walls $XYZ,WYZ,WZX,WXY$ of this tetrahedron. 
From the identification between the volume functions and the basis of the affine root system, we deduce the explicit volume functions in this chamber:
\be\label{vol-middle-chamber2}
\begin{aligned}
    \textrm{vol}_I(\bfD_j) &= \Big(1 - \talpha_1 - \talpha_2 - \talpha_3 - \talpha_4, ~-\talpha_1 - \talpha_2 + \talpha_3 + \talpha_4, \cr
    &\qquad -\talpha_1 + \talpha_2 - \talpha_3 + \talpha_4, ~-\talpha_1 + \talpha_2 + \talpha_3 - \talpha_4\Big), \cr
    \textrm{vol}_I(\bfV) &= 2 \talpha_1~,
\end{aligned}\ee
where the normalization in \eqref{VolumeRelation} is appropriately applied.

In fact, we can define the simple roots of the $D_4$ root system with a suitable choice of positive roots as
\begin{equation}\label{simpleroots}
    \{e^1, e^2, e^3, e^4\} = \{(-1, -1, 1, 1), (-1, 1, -1, 1), (-1, 1, 1, -1), (2, 0, 0, 0)\}~,
\end{equation}
Then, the highest root is expressed by $\theta=e^1+e^2+e^3+2e^4 = (1, 1, 1, 1)$.
Using these roots, the volume functions can be concisely written as
\be\label{vol-middle-chamber3}
\Big(\textrm{vol}_I(\bfD_1),\textrm{vol}_I(\bfD_2),\textrm{vol}_I(\bfD_3),\textrm{vol}_I(\bfD_4),\textrm{vol}_I(\bfV)\Big)=\Big(1-\theta\cdot\boldsymbol{\talpha},~e^1\cdot\boldsymbol{\talpha},~e^2\cdot\boldsymbol{\talpha},~e^3\cdot\boldsymbol{\talpha}~,e^4\cdot\boldsymbol{\talpha}\Big)~,
\ee
where $r\cdot\boldsymbol{\talpha}=\sum_{j=1}^4 r_j\talpha_j$ is the Euclidean inner product. As drawn in Figure \ref{fig:affineD4}, the homology class of each irreducible component of the $I_0^*$ singular fiber corresponds to an affine $D_4$ root as
\be \label{homology-root}
[\bfD_1]\leftrightarrow e^0=\delta-\theta~,\quad [\bfD_2]\leftrightarrow e^1~,\quad [\bfD_3]\leftrightarrow e^2~,\quad  [\bfD_4]\leftrightarrow e^3~,\quad [\bfV]\leftrightarrow e^4~,\quad 
\ee 
where $\delta$ is the imaginary root and the case $\delta\cdot\talpha =1$ is assumed for the volume function. Using the relation of the second homology classes, the homology class $[\bfF]$ of a generic Hitchin fiber indeed corresponds to the imaginary root $\delta$.

\begin{figure}[ht]
    \centering
    \includegraphics[width=0.3\linewidth]{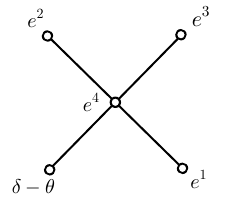}
    \caption{Simple roots of affine $D_4$ root system.}
    \label{fig:affineD4}
\end{figure}

\subsection{Classification of Kodaira singular fibers}\label{sec:classification}

Up to this point, we have considered the case where only the ramification parameters $\talpha_j$ are non-zero, while the other parameters are set to zero, i.e., $\tbeta_j = 0$ and $\tgamma_j = 0$. However, when the ramification parameters $\tbeta_j$ and $\tgamma_j$—which correspond to the masses in the SU(2) SQCD, as in \eqref{mass_monodromy}—are turned on, the configuration of the Hitchin fibration \eqref{Hitchin-fibration} undergoes a significant change. 

In our example of $\MH(C_{0,4},\SU(2))$, the Hitchin fibration is an elliptic fibration over an affine base $\cB_H$, with some singular fibers. The introduction of the parameters $\tbeta_j$ and $\tgamma_j$ modifies the types and configurations of these singular fibers. However, since the holomorphic symplectic form $\Omega_J$ does not depend on the ramification parameters $\tbeta_j$, the representation theory of $\SH$ is insensitive to these parameters. Thus, for simplicity, we set $\tbeta_j=0$. Hence, this subsection will focus on classifying the singular fibers and the second homology classes in the Hitchin fibration when the parameters $\tgamma_k$ are turned on. Although our motivation comes from the brane quantization on the Hitchin moduli space, this subsection also serves as the detailed analysis of low-energy dynamics of the SU(2) $N_f=4$ Seiberg-Witten theory.

The possible types of singular fibers in an elliptic fibration have been systematically classified by Kodaira \cite{kodaira1964structure,kodaira1966structure}.
The types of Kodaira singular fibers in the Hitchin fibration can be determined from the Seiberg-Witten curve  \eqref{SW-Nf4-generic} through two steps. The first step is to convert the Seiberg-Witten curve \eqref{SW-Nf4-generic} into the Weierstrass form $y^2=x^3+a(u) x+b(u)$  \cite{kodaira1964structure,neron1964modeles}. The second step is to find a vanishing locus $u=u_*$ of its discriminant 
\be \label{SW_discriminant}
\frakD(u_*):=-16\left(4 a(u_*)^3+27 b(u_*)^2\right),
\ee 
where $u_*$ is the location of a singular fiber at the Hitchin base $\cB_H$.
The vanishing orders of $a(u)$, $b(u)$, and $\frakD(u)$ at this specific $u_*$ determine the Kodaira type of the singular fiber. The concrete method is well-known in the literature (see, for instance, \cite[\S4.1]{Weigand:2018rez}), so we omit the details.

The appearance of these singular fibers on the Coulomb branch can be understood physically as the result of certain charged BPS particles becoming massless in the low-energy effective theory of SU(2) $N_f=4$ SQCD \cite{Seiberg:1994rs,Seiberg:1994aj}. Each singular fiber has an associated monodromy matrix in $\SL(2,\bZ)$, which encodes how the homology cycles of the elliptic curve transform around the singularity. The specific type of massless BPS particle is determined by the monodromy matrix around the singular fiber, once an electromagnetic frame is chosen. More precisely, if the particle becoming massless at $u_*$ has charge $(n_m, n_e)$, then its charge remains invariant under the monodromy matrix $M$ associated with $u_*$:
\begin{equation}\label{charge_monodromy}
    (n_m,n_e)\cdot M=(n_m,n_e)~.
\end{equation}

\begin{table}[ht]
\centering
\begin{adjustbox}{max width=\textwidth}
\renewcommand{\arraystretch}{1.3}
\begin{tabular}{c|c|c|c|c}
 {ord($\frakD$)} & {Kodaira types} & {Conditions} & {Examples in  $\tgamma_j$} & {Examples in mass $m_j$} \\
\hline 
 {$(1,1,1,1,1,1)$} & {$6I_1$} & {$\Ker _{\operatorname{ev}_{\tgamma}}=\emptyset$} & 
{$(\tgamma_1,\tgamma_2,\tgamma_3,\tgamma_4)$} &{$(m_1,m_2,m_3,m_4)$} \\
 \hline
 \multirow{2}*{$(2,1,1,1,1)$} & \multirow{2}*{$(I_2,4I_1)$} &  \multirow{2}*{$\Ker _{\operatorname{ev}_{\tgamma}}\cong\sfR(A_1)$}
& {$(\tgamma_1,\tgamma_2,\tgamma_3, 0)$}  & \multirow{2}*{$(m_1,m_2,m_3,m_3)$} \\
 & & & {$(\tgamma_1,\tgamma_2,\tgamma_3, \tgamma_1+\tgamma_2+\tgamma_3)$} & \\
\hline
\multirow{2}*{$(2,2,1,1)$} & \multirow{2}*{$(2I_2,2I_1)$} &  \multirow{2}*{$\Ker _{\operatorname{ev}_{\tgamma}}\cong\sfR(A_1)^{\oplus2}$}
& {$(\tgamma_1,\tgamma_2,\tgamma_1, \tgamma_2)$}  & {$(m_1,m_2,m_1,m_2)$} \\
& & & {$(\tgamma_1,\tgamma_2,0,0)$}  & {$(m_1,m_2,0,0)$} \\
\hline
\multirow{2}*{$(2,2,2)$} & \multirow{2}*{$3I_2$} &  \multirow{2}*{$\Ker _{\operatorname{ev}_{\tgamma}}\cong\sfR(A_1)^{\oplus3}$}
& {$(\tgamma_1,0,0,0)$}  &\multirow{2}*{$(m_1,m_1,0,0)$} 
\\
& & & {$(\tgamma_1,\tgamma_1,\tgamma_1,\tgamma_1)$}  & \\
\hline
{$(3,1,1,1)$} & {$(I_3,3I_1)$} &  {$\Ker _{\operatorname{ev}_{\tgamma}}\cong\sfR(A_2)$}
& {$(\tgamma_1,\tgamma_2,\tgamma_1+\tgamma_2,0)$} &{$(m_1,m_2,m_2,m_2)$} \\
\hline
\multirow{2}*{$(4,1,1)$} & \multirow{2}*{$(I_4,2I_1)$} &  \multirow{2}*{$\Ker _{\operatorname{ev}_{\tgamma}}\cong\sfR(A_3)$}
& \multirow{2}*{$(\tgamma_1,\tgamma_1,0,0)$}  &{$(m_1,m_1,m_1,m_1)$} \\
& & & &{$(m_1,0,0,0)$} \\
\hline
{$(6)$} & {$I_0^*$} & {$\Ker _{\operatorname{ev}_{\tgamma}}\cong\sfR(D_4)$}& {$(0,0,0,0)$}  &{$(0,0,0,0)$} \\
\end{tabular}
\end{adjustbox}
\caption{The ``generic'' configurations of Kodaira singular fibers in the Hitchin fibration of $\MH(C_{0,4}, \SU(2))$, along with conditions on the evaluation map kernel $\ker \operatorname{ev}_\tgamma$. The table also lists the associated multiplicity ord($\frakD$) of the discriminant \eqref{SW_discriminant}, and examples of the corresponding parameters $\tgamma_j$ and $m_j$ where the mass parameters and monodromy parameters are related by \eqref{mass_monodromy}, where $\beta_i = 0$ is assumed for simplicity. The evaluation map $\operatorname{ev}_\tgamma$ is defined in \eqref{evaluation}. The singular fibers can be read off from affine root system $\Ker _{\operatorname{ev}_{\tgamma}}$, as detailed in the conditions column, with specific examples of $\tgamma_j$ and $m_j$ values provided for each case.}
\label{tab:fiber_sing_classification2}
\end{table}

As discussed in the previous subsection, the singular fiber of type $I_0^*$ corresponds to the affine $D_4$ Dynkin diagram. This correspondence is a general feature of Kodaira singular fibers, where each type of singular fiber forms a chain of $\mathbb{CP}^1$ components (also known as rational curves or $(-2)$-curves) connected according to the structure of an affine $ADE$ Dynkin diagram (except for Kodaira type $II$). Consequently, the intersection matrix of a singular fiber matches the corresponding affine Cartan matrix, up to an overall sign difference. While the analysis of the Seiberg-Witten curve \eqref{SW-Nf4-generic} with the ramification parameters $\tgamma_j$, as described above, determines the configurations of singular fibers, the relationship between singular fibers and affine root systems provides a systematic way to classify the types of singular fibers. To see that, given the ramification parameters $\tgamma_j$, we define the evaluation map as
\be \label{evaluation}
\operatorname{ev}_{\tgamma}:\sfR(D_4)\rightarrow\bR \ ; \ r=(r_1,r_2,r_3,r_4)\mapsto \tgamma_1 r_1+\tgamma_2r_2+\tgamma_3r_3+\tgamma_4r_4~.
\ee 
The kernel of this linear map precisely gives the root system associated with the singular fibers. Physically, it encodes the breaking patterns of the flavor symmetry group $\SO(8)$ of the 4d SQCD for given mass parameters.
By combining the analysis of the Seiberg-Witten curve with this root system approach, we can straightforwardly classify the possible configurations of the Hitchin fibrations. In the case where, given the kernel condition, the positions of singular fibers are generic, seven configurations appear as $\tgamma_j$ vary. These configurations are summarized in Table \ref{tab:fiber_sing_classification2}.

However, Table \ref{tab:fiber_sing_classification2} is \emph{not} exhaustive as the root system does \emph{not} uniquely determine the Kodaira types. To classify all possible configurations of the Hitchin fibration, it is helpful to consider the physical interpretation \eqref{charge_monodromy} of the singular fibers in Seiberg-Witten theory. For generic mass parameters, the Hitchin fibration contains six $I_1$ singular fibers ($6I_1$), which can be divided into two classes based on their distinct monodromies. 
In a preferred electro-magnetic frame, four of these singular fibers, with base points denoted by $p_i$ for $i=1, 2, 3, 4$, have monodromies given by
\be 
M_q=T=\begin{pmatrix}
    1 & 1 \\
0 & 1
\end{pmatrix}~.
\ee 
At these points, known as quark singularities, the $(n_m, n_e) = (0, 1)$-cycle $S^1\subset T^2$ is pinched in a Hitchin fiber so that a quark with charge $(n_m, n_e) = (0, 1)$ becomes massless. 
The remaining two singular fibers, with base points denoted by $p_i$ for $i=5, 6$, have monodromies given by
\be
M_d=T^2STS^{-1}T^{-2}=\begin{pmatrix} 
-1 & 4 \\
-1 & 3
\end{pmatrix}~,
\ee 
where the standard generators of $\SL(2,\bZ)$ are given by  
\begin{equation}\label{modular_matrices}
    T=\begin{pmatrix}
1 & 1 \\
0 & 1 
\end{pmatrix},~\qquad S=\begin{pmatrix}
0 & 1 \\
-1 & 0 
\end{pmatrix}~.
\end{equation}
At these points, known as dyon singularities, the $(n_m, n_e) = (1,-2)$-cycle $S^1\subset T^2$ is pinched in a Hitchin fiber so that a dyon with charge $(n_m, n_e) = (1, -2)$ becomes massless. 
Then, direct calculation shows that the combined monodromy satisfies  
\begin{equation}
    M_q^2M_dM_q^2M_d=\begin{pmatrix}
        -1&0\\
        0&-1
    \end{pmatrix}
\end{equation}
which corresponds to the monodromy for the $I_0^*$ singular fiber, as expected at infinity.\footnote{In \cite{Seiberg:1994aj}, an electro-magnetic frame is chosen such that there are four quark ($M_q$), one monopole ($M_m$), and one dyon ($M_d$) singularities, satisfying:
\be 
M_q^4 M_m M_d = 
\begin{pmatrix}
    -1 & 0 \\
    0 & -1
\end{pmatrix}.
\ee 
For convenience, we have adopted a different frame here.}

\begin{figure}
    \centering
    \includegraphics[width=0.8\linewidth]{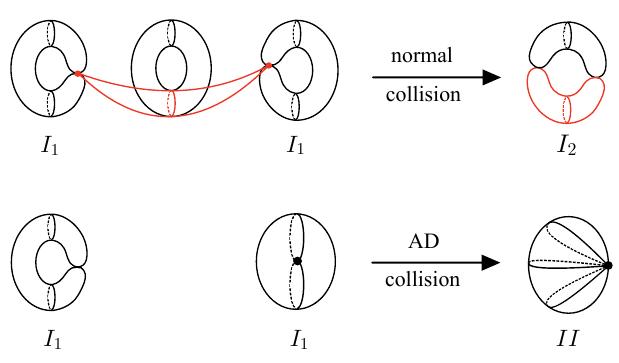}
    \caption{Two types of collisions of $I_1$ fibers. When two $I_1$ are of the same type, there is a cycle suspended between them, drawn in red. After the collision, an $I_2$ fiber appears where the suspended cycle becomes a (red) fiber cycle. When two $I_1$ are of distinct types, there is no cycle directly suspended between them, and collision will result in a $II$ type singular fiber.  }
    \label{fig:collide I1}
\end{figure}

There are two possible scenarios when two $I_1$ fibers collide to form a singular fiber. The first possibility is a collision of two $I_1$ fibers with monodromies of the same type. We refer to this as a \emph{normal collision}.  Since the same one-cycle $S^1\subset T^2$ becomes trivial for both singular fibers, there exists a two-cycle suspended between these singular fibers, as illustrated in Figure \ref{fig:collide I1}. Generally, when such fibers collide, an additional fiber cycle emerges. As a result, a normal collision of two $I_1$ fibers produces an $I_2$ singular fiber. By appropriately tuning the parameters $\tgamma_j$, one can successively collide additional $I_1$ fibers, leading to higher singular types as listed in Table \ref{tab:fiber_sing_classification2}.

The second possibility involves a collision of two $I_1$ fibers with monodromies of different types. 
With a specific choice of the UV coupling constant $\tau$ and the mass parameters of the Seiberg-Witten curve \eqref{SW-Nf4-generic}, a quark singularity and a dyon singularity can collide, resulting in the monodromy  
\be 
M_d \cdot M_q =
\begin{pmatrix}
-1 & 3 \\
-1 & 2
\end{pmatrix} \sim
\begin{pmatrix}
1 & 1 \\
-1 & 0
\end{pmatrix},
\ee  
which corresponds to the monodromy of Kodaira type $II$. At this point, mutually non-local degrees of freedom become massless simultaneously, as discussed in \cite{Argyres:1995jj, Argyres:1995xn}. The resulting low-energy theory is known as the $(A_1, A_2)$ Argyres-Douglas (AD) theory. 
Geometrically, AD points are characterized by the collision of singular fibers with monodromy of distinct types. We refer to this as an \emph{AD collision}.

Further collisions involving the $(A_1, A_2)$ AD point and quark singularities give rise to singular fibers of Kodaira types $III$ and $IV$, corresponding to the $(A_1, A_3)$ and $(A_1, D_4)$ AD theories \cite{Argyres:1995xn}, respectively. These account for the remaining configurations of the Hitchin fibration. 
Figure \ref{fig:Kodaira} classifies all possible Kodaira singular fibers in $\MH(C_{0,4},\SU(2))$ and illustrates how, starting from the $6I_1$ configuration, singular fibers undergo collisions and evolve into fibers of higher types as the ramification parameters $\tgamma_j$ are tuned, ultimately culminating in the $I_0^*$ singular fiber.  
Note that the kernel of the evaluation map \eqref{evaluation} is the root system associated with the singular fibers even with these ``exceptional'' cases. Nonetheless,
such collisions require a specific choice of the UV coupling constant $\tau$. A detailed analysis of these AD points in the Hitchin moduli space is presented in Appendix \ref{app:monodromy}.

\begin{figure}[ht]
    \centering
    \includegraphics{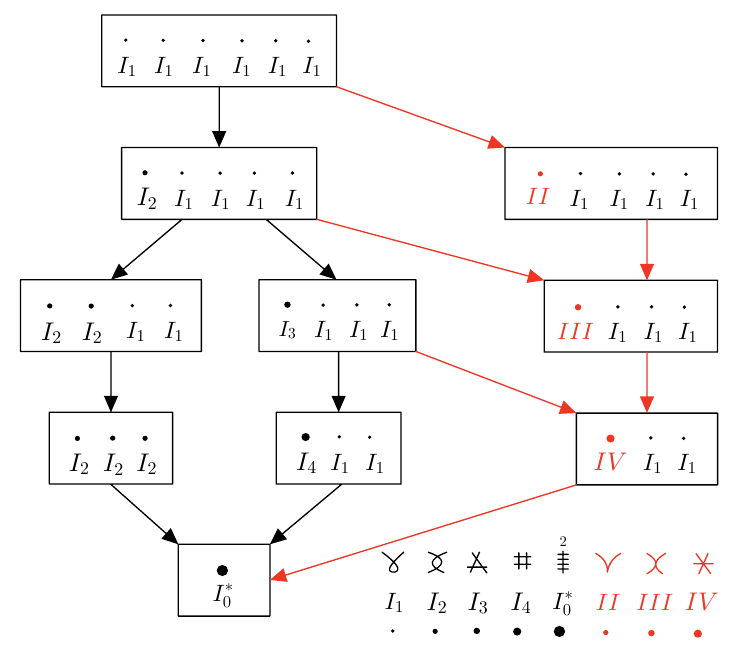}
    \caption{Classification of Kodaira singular fibers and their transitions under continuous variations of mass parameters. Seven configurations contain only $I_k$-type or $I^*_0$-type singular fibers, referred to as ``generic'' configurations, which are interconnected by black arrows representing normal collisions. The remaining three configurations involve $II$, $III$, and $IV$ types of singular fibers, corresponding to Argyres-Douglas points, and are denoted as ``exceptional'' configurations. These transitions are depicted with red arrows, indicating the presence of collisions of Argyres-Douglas type. The geometry of singular fibers is schematically illustrated at the bottom right, where each line represents a $\mathbb{CP}^1$ component.}
    \label{fig:Kodaira}
\end{figure}

\subsection{Generators of second homology groups and their volumes at generic masses}\label{sec:volume}

As the parameters $\tgamma_j$ vary, the Hitchin fibration over the $u$-plane exhibits increasingly intricate behavior. In this subsection, we aim to identify the generators of the second homology group and compute their volumes. This investigation also provides insights into how the geometry of the target space $\MS$ evolves as the parameters $(\talpha_j, \tgamma_j)$ are varied. We focus specifically on the central chamber $WXYZ$, defined in \eqref{wall_center}, where the parameters $(\talpha_j, \tgamma_j)$ can be freely adjusted. 
By selecting sufficiently generic parameters, we ensure that the geometry avoids developing du Val singularities. Consequently, the second homology group $H_2(\MH, \bZ)$ and its intersection form $Q$ remain invariant as the $\tgamma_j$ parameters are turned on. Therefore, $H_2(\MH, \bZ)$ can still be identified with the affine $D_4$ root lattice, as described in \eqref{Homology_Lattice}.

When the parameters $\tgamma_j$ are turned on, the global nilpotent cone of type $I_0^*$ splits into other singular fibers, and the homology generators $[\bfV]$, $[\bfD_j]$ ($j=1,2,3,4$) are not manifest in the Hitchin fibration. Nevertheless, they still span a basis of the second homology group at generic $\tgamma_j$. We describe the volume of a homology class $[\bfW]$ with respect to $\Omega_J / 2\pi i$ as:
\be\label{holomorphic-volume}
\operatorname{vol}(\bfW) := \int_{\bfW} \frac{\Omega_J}{2\pi i}~.
\ee
Thus, the volumes in the central chamber are obtained by replacing $\talpha_j$ in \eqref{vol-middle-chamber3} with $\talpha_j - i \tgamma_j$, yielding:
\be
\begin{aligned}\label{vol-middle-chamber4}
    \Big(\text{vol}(\bfD_1), &\text{vol}(\bfD_2), \text{vol}(\bfD_3), \text{vol}(\bfD_4), \text{vol}(\bfV)\Big) \\
    &= \Big(1 - \theta \cdot (\boldsymbol{\talpha} - i\boldsymbol{\tgamma}),~e^1 \cdot (\boldsymbol{\talpha} - i\boldsymbol{\tgamma}),~e^2 \cdot (\boldsymbol{\talpha} - i\boldsymbol{\tgamma}),~e^3 \cdot (\boldsymbol{\talpha} - i\boldsymbol{\tgamma}),~e^4 \cdot (\boldsymbol{\talpha} - i\boldsymbol{\tgamma})\Big),
\end{aligned}
\ee
where $e^i$ and $\theta$ represent the simple roots and the highest root, respectively, as defined in \eqref{simpleroots}.

An $I_k$ singular fiber consists of $k$ $\mathbb{CP}^1$ components arranged in a necklace shape where we denote irreducible components by $\bfU_a$ for $a = 1, \dots, k$. For instance, see Figure \ref{fig:411} for the $I_4$ singular fiber. To determine the volume function for these cycles, it is necessary to find the relationship between $[\bfU_a]$ and the basis $\{[\bfD_j], [\bfV]\}$. This can be done by simply using the structure of the affine root system.

The intersection form between the cycles $\bfU_a$ in an $I_k$ singular fiber corresponds to the Cartan matrix of the affine root system $A_{k-1}$, up to an overall minus sign. This correspondence allows us to identify the second homology group of an $I_k$ singular fiber with the affine $A_{k-1}$ root lattice, where the homology classes $[\bfU_a]$ serve as the simple roots of the root lattice. Consequently, specifying the relationships between $[\bfU_a]$ and the basis $\{[\bfD_j], [\bfV]\}$ reduces to finding an embedding of an affine sub-root lattice $\dt{\sfR}(\frakg)$ into the affine $D_4$ root lattice $\dt{\sfR}(D_4)$.

Importantly, the embedding of $\dt{\sfR}(\frakg) \hookrightarrow \dt{\sfR}(D_4)$ is not unique. This non-uniqueness arises because different choices of $\tgamma$ parameters can produce the same Hitchin fibration pattern. However, each embedding of the root system uniquely determines the $\tgamma$ parameters that generate the fibration pattern, up to cyclic permutations. Conversely, a fixed set of $\tgamma$ parameters uniquely determines the corresponding embedding of the affine subroot system.

To illustrate this interplay, we will conduct a detailed case study of the fibration configurations listed in Table \ref{tab:fiber_sing_classification2}. This analysis will clarify how specific embeddings of $\dt{\sfR}(\frakg) \hookrightarrow \dt{\sfR}(D_4)$ correspond to the $\tgamma$ parameters and their corresponding Hitchin fibration patterns.

\subsubsection*{Type $(I_4,2I_1)$}\label{I_4 singular fiber}
\begin{figure}
    \centering
    \includegraphics[width=0.7\linewidth]{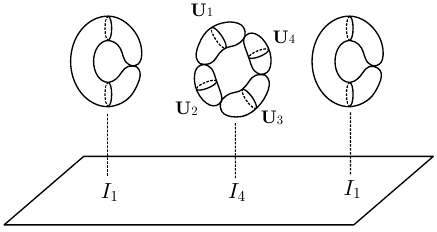}
    \caption{Cycles in Kodaira type $(I_4,2I_1)$.}
    \label{fig:411}
\end{figure}

Let us consider the fibration pattern $(I_4, 2I_1)$, as shown in Figure \ref{fig:411}. In this case, the singular fiber of type $I_1$ consists of a single irreducible component, and its homology class is simply the fiber class $[\bfF]$. So a generic fiber has volume $1$ for any choice of ($\tgamma_j,\talpha_j$). On the other hand, the $I_4$ singular fiber has four components: $\{[\bfU_1], [\bfU_2], [\bfU_3], [\bfU_4]\}$. The intersection form of these components coincides with the Cartan matrix of the affine $A_3$ root system up to an overall sign, which is:
\be
Q_{I_4}=
\begin{pmatrix}
-2 & 1 & 0 & 1 \\
1 & -2 & 1 & 0 \\
0 & 1 & -2 & 1 \\
1 & 0 & 1 & -2
\end{pmatrix}
\ee
Therefore, the components $[\bfU_a]$ correspond to the simple roots of the affine $A_3$ root system. To compute the volumes of these cycles, we need to express their homology classes $[\bfU_a]$ in terms of the basis of $H_2(\MH,\bZ)$. Although the data available now do not allow for a unique determination of these homology classes, we can infer that this setup specifies an embedding of the $A_3$ root system into the $D_4$ root system. One such possible embedding is given by:
\be\label{HomologyI42I1}
[\bfU_1] = [\bfV] + [\bfD_1], \quad
[\bfU_2] = [\bfD_2], \quad
[\bfU_3] = [\bfV] + [\bfD_3], \quad
[\bfU_4] = [\bfD_4].
\ee

\begin{figure}
    \centering
    \includegraphics[width=1\linewidth]{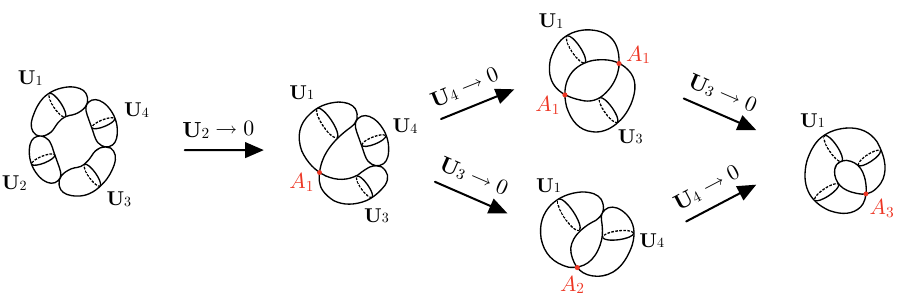}
    \caption{du Val singularities developed on $I_4$ fiber by shrinking cycles to points. $\bfU_a\to 0$ means that $\bfU_a$ shrinks to a point.}
    \label{fig:du_Val_I4}
\end{figure}

Next, since all of these components $[ \bfU_a ]$ are Lagrangian with respect to the symplectic form $\omega_K$, consistency requires that we choose the parameters $\tgamma_j$ such that the integral of $\omega_K$ over each component vanishes. This condition is encoded in the following system of linear equations:
\be
\begin{aligned}
-\tgamma_1 + \tgamma_2 + \tgamma_3 + \tgamma_4 &= 0~, \\
-\tgamma_1 - \tgamma_2 + \tgamma_3 + \tgamma_4 &= 0~, \\
\tgamma_1 + \tgamma_2 - \tgamma_3 + \tgamma_4 &= 0~, \\
-\tgamma_1 + \tgamma_2 + \tgamma_3 - \tgamma_4 &= 0~.
\end{aligned}
\ee
The solution to this system is $(\tgamma_1, \tgamma_2, \tgamma_3, \tgamma_4) = (\tgamma_1, 0, \tgamma_1, 0)$, which matches the conditions for the fibration configuration $(I_4, 2I_1)$ derived from the Seiberg-Witten curve, as classified in Table \ref{tab:fiber_sing_classification2}.
From the homology class relation \eqref{HomologyI42I1}, we deduce the explicit volume formulas for the irreducible components of the cycles:
\begin{equation}\label{I_4-volume}
\begin{aligned}
     \textrm {vol}_I(\bfU_1)&=1-(-\talpha_1+\talpha_2+\talpha_3+\talpha_4),\\
    \textrm {vol}_I(\bfU_2)&=-\talpha_1-\talpha_2+\talpha_3+\talpha_4,\\
    \textrm {vol}_I(\bfU_3)&=\talpha_1+\talpha_2-\talpha_3+\talpha_4,\\
    \textrm {vol}_I(\bfU_4)&=-\talpha_1+\talpha_2+\talpha_3-\talpha_4.
\end{aligned}
\end{equation}
When one or more of these volumes vanish, du Val singularities emerge. The conditions for such singularities align precisely with the classifications in Table \ref{tab:surface_sing_classification}, and these phenomena can be visualized geometrically in Figure \ref{fig:du_Val_I4}.

We can approach the analysis from the opposite direction by first specifying a choice of $\tgamma$ parameters.
For illustration, we select the $\tgamma$ parameters as $(\tgamma_1, \tgamma_2, \tgamma_3, \tgamma_4) = (\tgamma_1, 0, \tgamma_1, 0)$. Since this choice of $\tgamma_j$ gives $\Ker (\text{ev}_{\tgamma}) = \sfR(A_3)$, this implies that the homology classes of the Lagrangian cycles with respect to $\omega_K$ span an affine $A_3$ root lattice in $H_2(\MH,\bZ)$. 
Given a choice \eqref{wall_center} of the chamber in the $\talpha$ parameter space, (equivalently a choice of positive roots in the root lattice), $\Ker (\text{ev}_{\tgamma}) = \sfR(A_3) \subset  \sfR(D_4)$ uniquely specifies an embedding of $\dt\sfR(A_3)\hookrightarrow \dt\sfR(D_4)$. This will give the relation between the homology classes $[\bfU_a]$ and the basis of $H_2(\MH,\bZ)$, which is exactly \eqref{HomologyI42I1}. A similar analysis is applied to the other cases.

\subsubsection*{Type $(I_3,3I_1)$}
\begin{figure}
    \centering
    \includegraphics[width=0.7\linewidth]{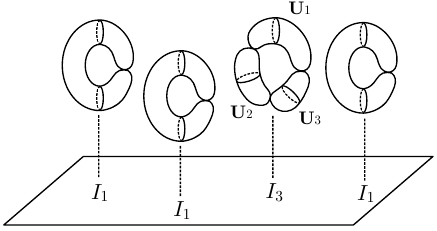}
    \caption{Cycles in Kodaira type $(I_3,3I_1)$. }
    \label{fig:3111}
\end{figure}
Let us now consider the $(I_3,3I_1)$ case, as in Figure \ref{fig:3111}.  Denote the  components of $I_3$ singular fiber as ${[\bfU_1],[\bfU_2],[\bfU_3]}$. The intersection form coincides with the affine $A_2$ Cartan matrix.
\begin{equation}
Q_{I_3}=
 \begin{pmatrix}
-2 & 1 & 1 \\
1 & -2 & 1 \\
1 & 1 & -2 
\end{pmatrix}  
\end{equation}
Thus, the components $[\bfU_a]$ can be interpreted as the simple roots of the affine $A_2$ root system. To compute their volumes, we express $[\bfU_a]$ in the basis of $H_2(\MH,\bZ)$. Although the data do not uniquely determine these homology classes, it specifies an embedding of the affine $A_2$ root system into the affine $D_4$ root system, such as:
\begin{equation}
     [\bfU_1]=[\bfV] + [\bfD_1]+[\bfD_2]~,\qquad 
    [\bfU_2]=[\bfV]+[\bfD_4]~,\qquad 
    [\bfU_3]=[\bfD_3]~.
\end{equation}
The Lagrangian conditions of cycles impose the following conditions on $\tgamma_j$
\begin{equation}
\begin{aligned}
    2\tgamma_2&=0,\\
   \tgamma_1+\tgamma_2+\tgamma_3-\tgamma_4&=0,\\
    -\tgamma_1+\tgamma_2-\tgamma_3+\tgamma_4&=0,
\end{aligned}
\end{equation}
with solution taking the form
$(\tgamma_1,\tgamma_2,\tgamma_3,\tgamma_4)=(\tgamma_1,0,\tgamma_3,\tgamma_1+\tgamma_3)$ . It coincides with one of the conditions when $(I_3,3I_1)$ singular type occurs as in Table \ref{tab:fiber_sing_classification2}. Integrating the volume form over these cycles, one finds that
\begin{equation}\label{volume_I3}
\begin{aligned}
     \textrm{vol}_I(\bfU_1)&=1-2\talpha_2\\
    \textrm{vol}_I(\bfU_2)&=\talpha_1+\talpha_2+\talpha_3-\talpha_4,\\
    \textrm{vol}_I(\bfU_3)&=-\talpha_1+\talpha_2-\talpha_3+\talpha_4.
\end{aligned}
\end{equation}
We can further analyze the singularities through the behavior of the volume function. When a single volume vanishes, the cubic surface develops an $A_1$ singularity. In contrast, when two volumes vanish simultaneously, the singularity enhances to $A_2$.

\subsubsection*{Type $(I_2,4I_1)$, $(2I_2,2I_1)$ and $3I_2$}
\begin{figure}
    \centering
    \includegraphics[width=0.7\linewidth]{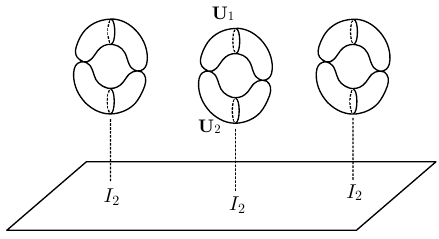}
    \caption{Cycles in Kodaira type $3I_2$.}
    \label{fig:222}
\end{figure}

We now analyze the cases $(I_2, 4I_1)$, $(2I_2, 2I_1)$, and $3I_2$, as depicted in Figures \ref{fig:2211} and \ref{fig:222}. Each $I_2$ fiber consists of two components, denoted as $\{[\bfU_1], [\bfU_2]\}$. These components correspond to the simple roots of the affine $A_1$ root system, with their intersection matrix matching the Cartan matrix of affine $A_1$ (up to an overall sign):
\begin{equation}Q_{I_2}=
    \begin{pmatrix}
-2 & 2 \\
2 & -2 
\end{pmatrix} 
\end{equation}
The homology classes of the components $[\bfU_a]$ can be expressed in terms of the basis of $H_2(\MH, \bZ)$, providing an embedding of the affine $A_2$ root system into the affine $D_4$ root system:
\begin{equation}
     [\bfU_1]=[\bfV] + [\bfD_1]+[\bfD_2]~,\qquad 
    [\bfU_2]=[\bfV]+[\bfD_3]+[\bfD_4]~,\qquad 
\end{equation}
In this case, the Lagrangian condition enforces $\tgamma_2 = 0$, which corresponds to the $(I_2, 4I_1)$ singularity type. Integrating the volume form over these cycles yields:
\begin{equation}\label{volume_I2_1}
\begin{aligned}
     \textrm{vol}_I(\bfU_1)&=1-2\talpha_2,\\
    \textrm{vol}_I(\bfU_2)&=2\talpha_2.
\end{aligned}
\end{equation}
When $\talpha_2=0~\text{or}~ \frac{1}{2} $, one $A_1$ singularity is developed.

For the $(2I_2, 2I_1)$ case, the singular fiber corresponds to a reducible affine ${A}_1 \oplus {A}_1$ root system. Each component $[\bfU_a]$ corresponds to a simple root of this system. By an embedding of the affine ${A}_1\oplus {A}_1$ root system into the affine $D_4$ root system, we can express $[\bfU_a]$ in terms of the basis of $H_2(\MH,\bZ)$:
\bea    
&[\bfU^{(1)}_1]=[\bfV] + [\bfD_1]+[\bfD_2]~,\qquad 
    [\bfU^{(1)}_2]=[\bfV]+[\bfD_3]+[\bfD_4]~,\qquad \cr 
   & [\bfU^{(2)}_1]=[\bfV] + [\bfD_1]+[\bfD_3]~,\qquad 
    [\bfU^{(2)}_2]=[\bfV]+[\bfD_2]+[\bfD_4]~,\qquad 
\eea
The lagrangian condition imposes $\tgamma_2=\tgamma_3=0$. The volume for the first $I_2$ cycles is the same as \eqref{volume_I2_1}. The volume cycles of the second fiber can also be achieved by straightforward calculation via \eqref{vol-middle-chamber2}.
\begin{equation}\label{volume_I2_2}
\begin{aligned}
     \textrm{vol}_I(\bfU^{(2)}_1)&=1-2\talpha_3,\\
    \textrm{vol}_I(\bfU^{(2)}_2)&=2\talpha_3.
\end{aligned}
\end{equation}

For the $3I_2$ case, the singular fibers correspond to a reducible affine root system ${A}_1\oplus{A}_1\oplus{A}_1$, which is the same as the one taken in \cite{Gukov:2022gei}. With the same analysis as before, we can take an embedding of the root lattice such that
\bea 
\relax  [\bfU^{(1)}_1]=[\bfV] + [\bfD_1]+[\bfD_2]~,\qquad 
    [\bfU^{(1)}_2]=[\bfV]+[\bfD_3]+[\bfD_4]~,\qquad \cr 
\relax    [\bfU^{(2)}_1]=[\bfV] + [\bfD_1]+[\bfD_3]~,\qquad 
    [\bfU^{(2)}_2]=[\bfV]+[\bfD_2]+[\bfD_4]~,\qquad \cr 
\relax [\bfU^{(3)}_1]=[\bfV] + [\bfD_1]+[\bfD_4]~,\qquad 
    [\bfU^{(3)}_2]=[\bfV]+[\bfD_2]+[\bfD_3]~,\qquad 
\eea
Then, the Lagrangian condition implies that $\tgamma_2 = \tgamma_3 = \tgamma_4 = 0$. The volume of the components in the first two singular fibers is the same as the one in \eqref{volume_I2_2}. For the third singular fiber, the volume is given by:
\begin{equation}\label{volume_I2_3}
\begin{aligned}
     \textrm{vol}_I(\bfU^{(3)}_1)&=1-2\talpha_4,\\
    \textrm{vol}_I(\bfU^{(3)}_2)&=2\talpha_4.
\end{aligned}
\end{equation}
Finally, as a consistency check, we consider the limit corresponding to the type $A_1$ DAHA, specifying the $\boldsymbol{t}$-parameters as in \eqref{A1Parameters}. In this limit, the volume formulas for the cycles precisely match the analysis presented in \cite{Gukov:2022gei}.

\begin{figure}
    \centering
    \includegraphics[width=0.7\linewidth]{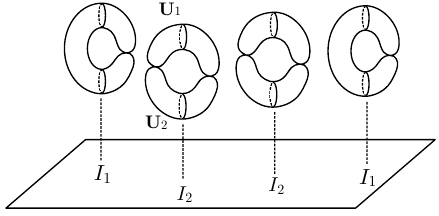}
    \caption{Cycles in Kodaira type $(2I_2,2I_1)$.}
    \label{fig:2211}
\end{figure}

\subsubsection*{Volume for suspended cycles}

\begin{figure}
    \centering
    \includegraphics[width=0.75\linewidth]{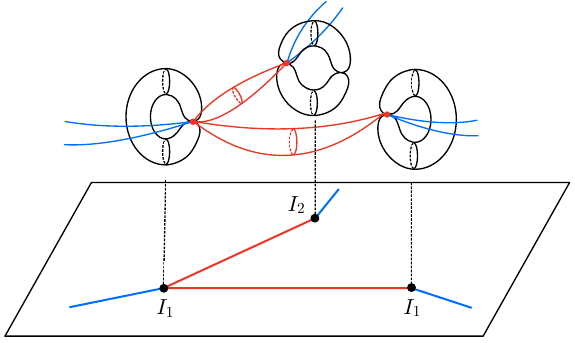}
    \caption{The cycles in the second homology of the Hitchin total space can be represented by a diagram on the $u$-plane. Singular fibers are represented as marked points along with their Kodaira types. Suspended cycles are represented by the lines connecting marked points.}
    \label{fig:suspended_cycle}
\end{figure}

Up to this point, we study homology classes of irreducible components in Kodaira singular fibers. However, as illustrated in Figure \ref{fig:collide I1}, there exist two-cycles suspended between two Kodaira singular fibers at a generic $\tgamma_j$. 
We study two-cycles suspended between two Kodaira singular fibers in more detail below, and we also identify generators of the second integral homology groups in each configuration of the Hitchin fibrations. This analysis will be useful for the $A$-model approach to the representation theory of $\SH$.

First, let us consider the case $6I_1$ with generic ramification parameters. 
The homology class of an $I_1$ singular fiber in $H_2(\MH, \bZ)$ is equivalent to the class $[\bfF]$ of a generic fiber, while the second homology satisfies $H_2(\MH, \bZ) \cong \bZ^{\oplus 5}$, as previously established. Therefore, there must exist additional homology generators in the Hitchin moduli space in the case of $6I_1$ singularities. To identify these generators, we begin by showing that for any class $[\bfW] \in H_2(\MH, \bZ)$, there exists a suitable representative $\bfW$ such that its projection $h(\bfW)$ onto the Hitchin base $\mathcal{B}_H$, via the fibration described in \eqref{Hitchin-fibration}, satisfies one of the following conditions:
\begin{enumerate}[nosep]
    \item[(1)] a point
    \item[(2)] a (piece-wise) line between Kodaira singular points 
\end{enumerate}
If the projection $h(\bfW)$ is two-dimensional, then $\bfW$ intersects the Hitchin fibers at discrete points due to dimensionality. However, since the Hitchin base $\cB_H$ is a complex plane $\bC$, a representative of $[\bfW]$ can be homotopically deformed such that its projection to $\cB_H$ reduces to either a point or a one-dimensional curve (with or without endpoints) in $\cB_H$. 
In the former case, where the projection is a point, a representative of $[\bfW]$ is supported entirely on a single Hitchin fiber, corresponding to condition (1) above. In the latter case, a representative of $[\bfW]$  intersects a generic Hitchin fiber along $S^1 \subset T^2\cong \bfF$, again by dimensionality. If we denote the projection to $\cB_H$ again by $h(\bfW)$ (abusing notation), this curve $h(\bfW)$ can end at points where the $S^1$ collapses to a trivial cycle in a Hitchin fiber, which occurs only at singular fibers. Note that for a singular fiber with monodromy $M$, the ($n_m,n_e$)-cycle subject to \eqref{charge_monodromy} collapses to a trivial cycle. 
As a result, a representative of $[\bfW]$ can always be deformed such that its projection to $\cB_H$ is either a composition of lines connecting Kodaira singular points (see Figure \ref{fig:suspended_cycle}) or a point in $\cB_H$. Therefore, the statement holds.

\begin{figure}
    \centering
    \includegraphics[width=1\linewidth]{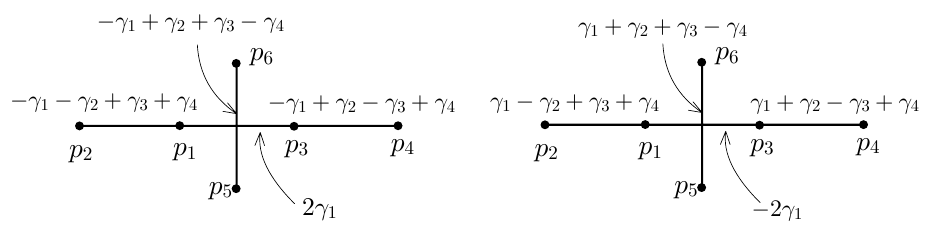}
    \caption{The generators of the second homology of the Hitchin total space for the $6I_1$ case. Singular fibers at $p_i$ are all of type $I_1$, and suspended cycles are projected onto lines between these points. The cycle suspended between $p_1$ and $p_3$ intersects with the one suspended between $p_5$ and $p_6$ once. Two examples of length assignments are related by a PL transformation \eqref{mPL} with respect to the cycle suspended $[\bfW_{13}]$ between $p_1$ and $p_3$.}
    \label{fig:6I1 cycles}
\end{figure}

The two-cycles satisfying condition (1) will be referred to as \emph{fiber} (two-)cycles, while those satisfying condition (2) will be called \emph{suspended (two-)cycles}. 
As seen in Figure \ref{fig:collide I1}, there exists a suspended cycle between the two $I_1$ singular fibers of the same type.  In the case of the configuration $6I_1$, these suspended cycles as well as the homology class $[\bfF]$ of a generic fiber generate the second homology group $H_2(\MH,\bZ)$, forming the affine $D_4$ root lattice \eqref{Homology_Lattice}:  
\be
\dt{\sfQ}(D_4) \cong \sfQ(D_4) \oplus \bZ\langle \delta \rangle,
\ee
where the homology class of a generic fiber serves as the imaginary root $\delta$. 
Hence, the homology classes represented by the suspended cycles span the root lattice $\sfQ(D_4)$ with the intersection form given by the Cartan matrix of type $D_4$ (up to an overall sign).

From the analysis above, suspended cycles can be represented as line segments connecting six marked points $p_i$ on the $u$-plane, corresponding to the loci of the $I_1$ singular fibers. We denote the cycle associated with the line segment between $p_i$ and $p_j$ as $\bfW_{ij}$. Recalling that four marked points $p_i$ ($i=1,2,3,4$) are quark singularities while two marked points $p_i$ ($i=5,6$) are dyon singularities. Since suspended cycles can exist between singular fibers of the same type, a consistent method to assign the cycles is illustrated in Figure \ref{fig:6I1 cycles}. 
Note that additional cycles, which wind around singular fibers, also exist. These winding cycles are discussed in Appendix \ref{app:monodromy}, where they are explicitly shown to be homologous to the suspended cycles depicted in Figure \ref{fig:6I1 cycles}.

The homology classes of the suspended cycles are determined by embedding the $D_4$ root system into the affine $D_4$ root lattice, along with specifying a chamber in the $\tgamma$-space. Let us begin by considering a specific chamber:
\begin{equation} \label{gamma-chamber}
\{\tgamma \in \mathfrak{t}^{(\tgamma)}_{D_4} \mid e^i \cdot \tgamma > 0, \ i=1,2,3,4\}.
\end{equation}
In this chamber, the homology classes of the suspended cycles are identified as follows:
\be\label{suspended-homology}
    [\bfW_{13}] = [\bfV], \quad 
    [\bfW_{12}] = [\bfD_2], \quad 
    [\bfW_{34}] = [\bfD_3], \quad 
    [\bfW_{56}] = [\bfD_4].
\ee
Therefore, \eqref{vol-middle-chamber4} tells us the volumes of these suspended cycles defined in \eqref{holomorphic-volume}  are
\be\label{Volume_suspended_cycles}
\Big(\textrm{vol}(\bfW_{12}),\textrm{vol}(\bfW_{34}),\textrm{vol}(\bfW_{56}),\textrm{vol}(\bfW_{13})\Big)
    =\Big(e^1\cdot(\boldsymbol{\talpha}-i\boldsymbol{\tgamma}),
e^2\cdot(\boldsymbol{\talpha}-i\boldsymbol{\tgamma}),~e^3\cdot(\boldsymbol{\talpha}-i\boldsymbol{\tgamma})~,e^4\cdot(\boldsymbol{\talpha}-i\boldsymbol{\tgamma})\Big)
\ee

The imaginary part of the volume in \eqref{Volume_suspended_cycles} has a geometric interpretation. By comparing these volumes with the conditions for the Kodaira singular fibers (as detailed in Table \ref{tab:fiber_sing_classification2}), we observe that a Kodaira singular fiber of new type appears precisely when the imaginary part of a volume vanishes. Geometrically, this corresponds to the collision of two $I_1$ fibers, as shown schematically in Figure \ref{fig:collide I1}. Consequently, the imaginary part of the volume serves as a measure of the ``distance'' between two $I_1$ singular fibers. The left-hand panel of Figure \ref{fig:6I1 cycles} visually represents these distances in the central chamber.

It is important to emphasize that the embedding of the $D_4$ root system into the affine $D_4$ root lattice does not uniquely specify the homology classes of the suspended cycles. Wall-crossing phenomena, described in \S\ref{sec:wallcrossing}, allow transitions between different bases of their homology classes. As seen in \eqref{dtWD4onabc}, the $\tgamma$-space receives only the Weyl group $W(D_4)$ rather than the affine Weyl group  $\dt W(D_4)$ since the affine translation acts on the $\tgamma$-space trivially. 
Thus, the walls in the $\tgamma$-space are determined by the discriminant loci given by
\begin{equation}
\sum_{j=1}^4 r_j \boldsymbol{\tgamma}_j = 0, \qquad \forall r=(r_1,r_2,r_3,r_4) \in \sfR(D_4),
\end{equation}
where $\sfR(D_4)$ denotes the $D_4$ root system. These walls correspond to the reflection hyperplanes of the $D_4$ weight lattice. As a result, each chamber in the $\tgamma$-space, bounded by these walls, corresponds to a Weyl chamber of type $D_4$. In Figure \ref{fig:gamma_space}, we schematically represent these walls by black lines.

The volume in \eqref{Volume_suspended_cycles} is proportional to the distance from the chamber boundaries, where the chamber is specified by \eqref{gamma-chamber}.
This relation indicates a wall-crossing phenomenon for the volume functions in the $\gamma$-space. When $\tgamma_j$ crosses a wall defined by a root $r$, the basis of the root undergoes a Weyl reflection $s_r$:  
\be\label{Weyl-reflection}
s_{r_a}(r_b) = r_b - A_{ba} r_a~,
\ee
where $A_{ba}$ is the $D_4$ Cartan matrix. Consequently, the volume functions exhibit a discontinuity, with their dependence on $\tgamma_j$ jumping according to the Weyl reflection. Since the wall-crossing is governed by Weyl reflections, the PL transformation defined by \eqref{mPL} still provides the basis transformation of the second homology classes once crossing a wall. The wall-crossing of the volume functions with respect to the cycle $[\bfW_{13}]$ is depicted in Figure \ref{fig:6I1 cycles}.

To better understand the physical interpretation of the volume function, we examine the $\talpha_j = 0$ limit and express the volumes in terms of the mass parameters $m_j$ as defined in \eqref{mass_monodromy}. In this case, the volumes of the relevant 2-cycles are given by:  
\bea\label{volume-suspended}
    \text{vol}(\bfW_{12}) &= m_1 - m_3, \\
    \text{vol}(\bfW_{13}) &= -m_1 - m_2, \\
    \text{vol}(\bfW_{34}) &= m_2 + m_4, \\
    \text{vol}(\bfW_{56}) &= m_2 - m_4.
\eea  
These expressions illustrate how the mass parameters dictate the geometric volumes of the cycles, encapsulating both their physical and geometric significance.  

This result can also be derived in the string theory framework of Seiberg-Witten theory, using the complex structure $I$ instead of $J$. Using the isomorphism  $\SU(2) \cong \Sp(1)$, the $\Sp(1)$ theory with $N_f = 4$ can be realized by considering a D3-brane in the presence of four D7-branes and an O7${}^-$-plane in Type IIB string theory \cite{Banks:1996nj}. The four D7-branes carry electric charges $(n_m, n_e) = (0, 1)$. A notable result from \cite{sen1996f} shows that the O7${}^-$-plane splits into two 7-branes with charges $(n_m, n_e) = (1, 0)$ and $(n_m, n_e) = (1, -2)$. Remarkably, the charges of these 7-branes agree perfectly with those of the massless BPS particles on the $u$-plane in the 4d $\cN = 2$ theory.  

In this framework, the mass (or equivalently, the central charge) of a BPS string stretched between the $i$-th and $j$-th 7-branes has a clear geometric interpretation. It is determined by integrating the holomorphic symplectic form $\Omega_I$ over the 2-cycle connecting the two $I_1$ singular fibers associated with the 7-branes \cite{sen1996f}:  
\be  
 \left| \int_{\bfW_{ij}} \frac{\Omega_I}{2\pi} \right| = \left| m_a \pm m_b \right|,  
\ee  
where $m$ are the mass parameters of the 4d theory. The precise relation between $\bfW_{ij}$ and the masses depends on a choice of chamber and electromagnetic frame of BPS particles, which we do not specify here. Nevertheless, our analysis in \eqref{volume-suspended} is consistent with the result of \cite{sen1996f} in the complex structure $I$, provided $\tbeta_j = 0$. Consequently, the \HK structure ensures a unified and consistent connection between the analyses in the complex structures $I$ and $J$.

\subsection{Symmetry actions}\label{sec:symmetry}

As seen in \S\ref{sec:SH}, the spherical DAHA $\SH$ enjoys the symmetries that involve maps on the deformation parameters $\boldsymbol{t}$. The symmetry action can be understood as the group of certain diffeomorphisms (or symplectomorphisms as we will see in \S\ref{sec:category}) of the target space $\X$ with different ramification parameters.  Recall that the deformation parameters $\boldsymbol{t}$ are related to the $(\talpha, \tgamma)$ parameters via $t_j = e^{-2\pi(\tgamma_j + i\talpha_j)}$ \eqref{tj}. Here, we analyze the actions of the braid group $B_3$ and the sign-flip group $\bZ_2\times \bZ_2$ at the level of homology classes. To achieve this, we must choose a preferred uplift of the group actions to the $(\talpha, \tgamma)$ parameter space. The chamber structure in the $(\talpha, \tgamma)$ parameter space, as discussed in the previous subsections, introduces ambiguity, as the uplift may map one chamber to another. To resolve the ambiguity, we adopt a preferred uplift where the actions of both the braid group $B_3$ and the sign-flip group $\bZ_2\times \bZ_2$ preserve the chamber structure.

To be explicit, we focus on the chamber $WXYZ$ defined in \eqref{wall_center}. Our goal is to analyze the action of the braid group $B_3$ as described in \eqref{Classical_braid_action_sdaha}. A family of character varieties is parametrized by the ramification parameters $(\talpha, \tgamma)$, and the generators of $B_3$ act as diffeomorphisms between the character varieties at different points in the parameter space. The explicit maps of the ramification parameters are given by:
\bea\label{braid_action_daha2}
\tau_+: \left(\talpha_1, \talpha_2, \talpha_3, \talpha_4\right) &\mapsto \left(\talpha_1, \talpha_2, \talpha_4, \talpha_3\right), \\
\left(\tgamma_1, \tgamma_2, \tgamma_3, \tgamma_4\right) &\mapsto \left(\tgamma_1, \tgamma_2, \tgamma_4, \tgamma_3\right), \\
\tau_-: \left(\talpha_1, \talpha_2, \talpha_3, \talpha_4\right) &\mapsto \left(\talpha_1, \talpha_4, \talpha_3, \talpha_2\right), \\
\left(\tgamma_1, \tgamma_2, \tgamma_3, \tgamma_4\right) &\mapsto \left(\tgamma_1, \tgamma_4, \tgamma_3, \tgamma_2\right), \\
\sigma: \left(\talpha_1, \talpha_2, \talpha_3, \talpha_4\right) &\mapsto \left(\talpha_1, \talpha_3, \talpha_2, \talpha_4\right), \\
\left(\tgamma_1, \tgamma_2, \tgamma_3, \tgamma_4\right) &\mapsto \left(\tgamma_1, \tgamma_3, \tgamma_2, \tgamma_4\right).
\eea
These parameter transformations induce corresponding diffeomorphisms of the character varieties, which we schematically write 
\bea
\tau_+: \X(t_1,t_2,t_3,t_4) \to \X(t_1,t_2,t_4,t_3) \ ;\ (x, y, z) \mapsto (x, xy - z - \theta_3, y), \cr
\tau_-: \X(t_1,t_2,t_3,t_4) \to \X(t_1,t_4,t_3,t_2) \ ;\ (x, y, z) \mapsto (xy - z - \theta_3, y, x), \cr
\sigma: \X(t_1,t_2,t_3,t_4) \to \X(t_1,t_3,t_2,t_4) \ ;\ (x, y, z) \mapsto (y, x, xy - z - \theta_3).
\eea
Furthermore, computing the changes of the volumes using \eqref{vol-middle-chamber4}, we can determine the corresponding transformations of homology cycles:
\bea\label{Braidgroup_Action}
\tau_+: [\bfD_3] \leftrightarrow [\bfD_4], \quad & [\bfD_1]~\text{and}~[\bfD_2]~\text{remain invariant}, \\
\tau_-: [\bfD_2] \leftrightarrow [\bfD_4], \quad & [\bfD_1]~\text{and}~[\bfD_3]~\text{remain invariant}, \\
\sigma: [\bfD_2] \leftrightarrow [\bfD_3], \quad & [\bfD_1]~\text{and}~[\bfD_4]~\text{remain invariant}.
\eea

Next, we analyze the action of the sign-flip symmetry group $\bZ_2^{\times 2}$, as defined in \eqref{Z2action}. This group can be understood as the group of diffeomorphisms between the character varieties at different points in the parameter space. Specifically, the diffeomorphisms $\xi_1$, $\xi_2$, and $\xi_3$ are written as
\bea\label{Z2-diffeo}
    \xi_1: \X(t_1, t_2, t_3, t_4) &\to \X(t_1, -t_2^{-1}, t_3, -t_4^{-1}); \quad (x, y, z) \mapsto (-x, y, -z), \\
    \xi_2: \X(t_1, t_2, t_3, t_4) &\to \X(t_1, t_2, -t_3^{-1}, -t_4^{-1}); \quad (x, y, z) \mapsto (x, -y, -z), \\
    \xi_3: \X(t_1, t_2, t_3, t_4) &\to \X(t_1, -t_2^{-1}, -t_3^{-1}, t_4); \quad (x, y, z) \mapsto (-x, -y, z).
\eea
While we have so far considered the action at the level of the $t_i$-parameters, there is a natural way to lift this action to the ramification parameters $(\talpha_j, \tgamma_j)$. This lifting is done by requiring the action to preserve the chamber structure. Explicitly, the action of $\bZ_2^{\times 2}$ is given as follows:
\bea
    \xi_1: (\talpha_1, \talpha_2, \talpha_3, \talpha_4) &\mapsto (\talpha_1, \tfrac{1}{2} - \talpha_2, \talpha_3, \tfrac{1}{2} - \talpha_4), \\
    (\tgamma_1, \tgamma_2, \tgamma_3, \tgamma_4) &\mapsto (\tgamma_1, -\tgamma_2, \tgamma_3, -\tgamma_4), \\
    \xi_2: (\talpha_1, \talpha_2, \talpha_3, \talpha_4) &\mapsto (\talpha_1, \talpha_2, \tfrac{1}{2} - \talpha_3, \tfrac{1}{2} - \talpha_4), \\
    (\tgamma_1, \tgamma_2, \tgamma_3, \tgamma_4) &\mapsto (\tgamma_1, \tgamma_2, -\tgamma_3, -\tgamma_4), \\
    \xi_3: (\talpha_1, \talpha_2, \talpha_3, \talpha_4) &\mapsto (\talpha_1, \tfrac{1}{2} - \talpha_2, \tfrac{1}{2} - \talpha_3, \talpha_4), \\
    (\tgamma_1, \tgamma_2, \tgamma_3, \tgamma_4) &\mapsto (\tgamma_1, -\tgamma_2, -\tgamma_3, \tgamma_4).
\eea
Therefore, the maps \eqref{Z2-diffeo} provide diffeomorphisms between the character varieties at these two points.
To understand the action of $\bZ_2^{\times 2}$ on the homology classes, we compute the induced transformations using the volume formula. The action on the homology classes $[\bfD_i]$ is given by:
\bea\label{signflip_Action}
    \xi_1^*: [\bfD_1] \leftrightarrow [\bfD_3], \quad &[\bfD_2] \leftrightarrow [\bfD_4], \\
    \xi_2^*: [\bfD_1] \leftrightarrow [\bfD_2], \quad &[\bfD_3] \leftrightarrow [\bfD_4], \\
    \xi_3^*: [\bfD_1] \leftrightarrow [\bfD_4], \quad &[\bfD_2] \leftrightarrow [\bfD_3].
\eea
By analyzing the combined action of the braid group $B_3$ and the sign-flip group $\bZ_2^{\times 2}$ on the homology classes, we observe that they together generate the outer automorphism group of the affine $D_4$ Lie algebra.

\section{Branes vs Representations}\label{sec:brane-rep}

Having thoroughly studied the geometry of the target space of the $A$-model, we are now ready to introduce the main actors of our story—the branes. As outlined in \S\ref{sec:brane-quantization}, our motivation stems from brane quantization on the $\SL(2,\bC)$-character variety $(\X,\omega_\X)$ and its connection to the representation theory of the spherical DAHA $\SH$.

The main goal of this section is to establish a precise correspondence between $A$-branes in the sigma-model
and $\SH$-modules.
We first identify a set of non-compact
$(A,B,A)$-branes that correspond to polynomial representations of $\SH$. We then
construct the correspondence between compactly supported Lagrangian $A$-branes
and finite-dimensional representations by comparing defining properties,
including dimensions, shortening conditions, and morphism spaces.

In addition, from the perspective of the 2d $A$-model, we describe the action of the affine braid group on the category.

\subsection{Affine braid group action on category}\label{sec:category}

The focus of this section is the derived equivalence proposed by brane quantization, given explicitly by the functor
\begin{equation}\label{functor-explicit}
\RHom(-,\Bcc): D^b \ABrane(\X, \omega_{\X}) \rightarrow D^b \Rep(\SH), \quad \brane_\bfL \mapsto \scL = \Hom(\brane_\bfL,\Bcc)~.
\end{equation}
Before we explicitly establish the correspondence between $A$-branes and $\SH$-modules, let us first consider the symmetries acting on these two categories, known in category theory as auto-equivalences.
In fact, the equivalence of these categories requires that the auto-equivalences on both sides must correspond to each other.

In this respect, the $A$-model perspective becomes particularly insightful. The $A$-model depends on the quantum (or complexified) \K moduli but is independent of the complex structure moduli. As we vary parameters in the complex structure moduli space, the resulting monodromy transformations of $A$-branes induce non-trivial symmetries on the $A$-brane category. 

To describe quantum \K moduli space, let us recall the role of the $B$-field \eqref{B-field} in the 2d sigma-model. Beyond the triple $(\talpha, \tbeta, \tgamma)$ of ramification parameters, the $B$-field introduces an additional ``quantum'' parameter, denoted by $\teta$, to the 2d sigma-model \cite{Aspinwall:1995zi}. This inclusion enriches the parameter space of the model. Extending the construction in \eqref{dtWD4onabc}, the quadruple of parameters takes the value
\be 
(\talpha, \tbeta, \tgamma, \teta) \in \frac{\frakt_{D_4} \times \frakt_{D_4} \times \frakt_{D_4} \times \frakt_{D_4}^\vee}{\dt W(D_4)} \cong T_{D_4} \times \frakt_{D_4} \times \frakt_{D_4} \times T_{D_4}^\vee~.
\ee 
Under geometric Langlands duality, $\talpha$ and $\teta$ are exchanged, while $\tbeta$ and $\tgamma$ are also exchanged \cite{Gukov:2006jk}.
In our setup, it is natural to normalize the period of the $B$-field over a generic Hitchin fiber $\bfF$ as:
\be 
\int_{\bfF} \frac{B}{2\pi} = 1~.
\ee 
With this normalization, we can define the volume of a submanifold $\bfW$ as:
\be 
\operatorname{vol}_B(\bfW) \equiv \int_\bfW \frac{B}{2\pi}~.
\ee 
The integrals of the $B$-field over the basis elements of $H_2(\MH, \bZ)$ are then expressed as:
\be 
(\operatorname{vol}_B(\bfD_1), \operatorname{vol}_B(\bfD_2), \operatorname{vol}_B(\bfD_3), \operatorname{vol}_B(\bfD_4), \operatorname{vol}_B(\bfV)) = (1 - \theta \cdot \teta, e^1 \cdot \teta, e^2 \cdot \teta, e^3 \cdot \teta, e^4 \cdot \teta)~,
\ee 
where the Euclidean inner product is assumed as $r \cdot \teta = \sum_{j=1}^4 r_j \teta_j$.

As discussed in \S\ref{sec:wallcrossing}, one of the key insights from \cite{Gukov:2006jk} is the action of the affine Weyl group $\dt W(D_4)$ on $H_2(\MH, \bZ)$. This action is explicitly realized through the Picard-Lefschetz monodromy transformations. Another important insight in \cite{Gukov:2006jk} is the affine braid group action on the category of $A$-branes.

To illustrate the affine braid group action explicitly, let us consider the case where $\hbar$ is real, such that the target space is the symplectic manifold $(\X, \omega_K)$. In this setting, the $A$-model on $(\X, \omega_K)$ depends only on the quantum \K parameter
\be 
v = \exp(2\pi (\tgamma + i\teta))~,
\ee 
and is independent of the complex structure parameters $(\talpha, \tbeta)$. This implies that, for the study of $A$-branes, we should fix the value of $v$. 
At the same time, since the $A$-model is locally independent of $\talpha$ and $\tbeta$, the group of monodromies arising from varying $\talpha$ and $\tbeta$ generates auto-equivalences on the $A$-brane category.

Since the $A$-model depends on $v$, monodromies that change $v$ are irrelevant. The relevant part of the Weyl group is therefore the subgroup that leaves $v$ fixed, which keeps invariant the kernel of the evaluation map: 
\be \label{evaluation2}
\textrm{ev}_{\tgamma + i\teta}: \sfR(D_4) \to \bC \ ; \ r=(r_1,r_2,r_3,r_4)\mapsto \sum_{j=1}^4 r_j (\tgamma_j + i\teta_j)~.
\ee 
We write its affine extension by $\dt W_v\subset \dt W(D_4)$.

In this context, a pair $(\talpha, \tbeta)$ is said to be $\dt W_v$-\emph{regular} if it is not fixed by any element of $\dt W_v$ other than the identity. Non-regular quadruples $(\talpha, \tbeta, \tgamma, \teta)$ correspond to singularities in the target space $(\X, \omega_K)$. The group $\dt W_v$ acts freely on the space $(\mathfrak{t}_{D_4}^{(\alpha)} \times \mathfrak{t}_{D_4}^{(\beta)})^{\dt{W}_v\text{-reg}}$ of regular pairs, and a family of smooth sigma-models is parametrized by the quotient 
\be
\frac{(\mathfrak{t}_{D_4}^{(\alpha)} \times \mathfrak{t}_{D_4}^{(\beta)})^{\dt{W}_v\text{-reg}}}{\dt{W}_v}~.
\ee 
On the other hand, the singularities, arising from non-regular pairs, have a co-dimension of at least two in the $(\talpha, \tbeta)$-space.  As a result, the group of monodromies that acts non-trivially on the $A$-model as graded symplectomorphisms is identified with the orbifold fundamental group:
\be 
\dt\Br_v \cong \pi_1\left(\frac{(\mathfrak{t}_{D_4}^{(\alpha)} \times \mathfrak{t}_{D_4}^{(\beta)})^{\dt{W}_v\text{-reg}}}{\dt{W}_v}\right)~.
\ee 
This group can be understood through the following short exact sequence:
\begin{equation}
1 \rightarrow \pi_1\left((\mathfrak{t}_{D_4}^{(\alpha)} \times \mathfrak{t}_{D_4}^{(\beta)})^{\dt{W}_v\text{-reg}}\right) \rightarrow \dt\Br_v \rightarrow \dt{W}_v \rightarrow 1~.
\end{equation}

For the specific case where $v=1$, the group $\dt W_v$ becomes the affine braid group $\dt W(D_4)$ type $D_4$. In this situation, the derived category receives a natural $\dt\Br(D_4)$-action:
\be
\begin{tikzpicture}[x=0.75pt,y=0.75pt,yscale=-1,xscale=1]
% Text Node
\draw (273,83) node [anchor=north west][inner sep=0.75pt]    {$\dt\Br(D_4)$};
% Text Node
\draw (130,120.4) node [anchor=north west][inner sep=0.75pt]    {$D^b\ABrane(\X ,\omega _{\X})_{\tgamma = 0 = \teta} \cong D^b\Rep(\SH)_{|q|=1=|t_j|}$};
% Text Node
\draw (335,112) node [anchor=north west][inner sep=0.75pt]  [rotate=-141.53]  {$\circlearrowright$};
% Text Node
\draw (270,120) node [anchor=north west][inner sep=0.75pt]  [rotate=-212.33]  {$\circlearrowright$};
\end{tikzpicture}
\ee 
The affine braid group $\dt\Br(D_4)$ is explicitly given by
\be
\dt\Br(D_4) \cong \pi_1\left(\frac{(\mathfrak{t}_{D_4}^{(\alpha)} \times \mathfrak{t}_{D_4}^{(\beta)})^{\dt{W}(D_4)\text{-reg}}}{\dt{W}(D_4)}\right)~.
\ee
In this parameter space, the locus where a two-cycle $\bfW_a$ collapses to zero size defines a singular divisor. The monodromy around this divisor gives rise to a generator $T_a$ of $\dt\Br(D_4)$. The defining relations of $\dt\Br(D_4)$ are obtained from those of the affine Weyl group~\eqref{affine-Weyl-D4} by omitting the quadratic relations $T_j^2 = 1$ and $T_{\bfV}^2 = 1$. 
Explicitly, the braid relations take the form
\be\label{braid-relation}
\begin{aligned}
    T_j T_k &= T_k T_j, \\
    T_j T_{\bfV} T_j &= T_{\bfV} T_j T_{\bfV}, \\
\end{aligned}
\ee
for any $j, k = 1, 2, 3, 4$. As we will see in \S\ref{sec:I0*}, $T_j^2$ and $T_{\bfV}^2$ give rise to grading shifts $-2$ to compact Lagrangian $A$-branes so that they are non-trivial graded symplectomorphisms.

So far, we have considered the case where $\hbar$ is real. However, the story remains analogous for a generic $\hbar = |\hbar|e^{i\theta}$, where the symplectic form $\omega_\X$ is given in \eqref{generic-Bcc}. In this more general setting, the $A$-model on $(\X, \omega_\X)$ depends on the parameter 
\be
v = \exp\left[2\pi \left(\Im\left(\frac{\tgamma_j + i\talpha_j}{-i e^{i\theta}}\right) + i\teta\right)\right]~,
\ee
while it remains independent of the pair
\be 
\left(\Re\left(\frac{\tgamma_j + i\talpha_j}{-i e^{i\theta}}\right), \tbeta\right)~.
\ee 
Similarly, we can define $\dt W_v$ using $(\log v)/2\pi$ in \eqref{evaluation2}. The quotient of the space of regular pairs by $\dt W_v$ then parametrizes the $A$-model. Consequently, the orbifold fundamental group of this space corresponds to the group of auto-equivalences of the $A$-brane category $\ABrane(\X, \omega_\X)$ as well as the representation category $\Rep(\SH)$:
\be
\begin{tikzpicture}[x=0.75pt,y=0.75pt,yscale=-1,xscale=1]
% Text Node
\draw (255,85) node [anchor=north west][inner sep=0.75pt]    {$\dt{\Br}_v$};
% Text Node
\draw (140,120.4) node [anchor=north west][inner sep=0.75pt]    {$D^b\ABrane(\X ,\omega _{\X}) \cong D^b\Rep(\SH)$};
% Text Node
\draw (295,112) node [anchor=north west][inner sep=0.75pt]  [rotate=-141.53]  {$\circlearrowright$};
% Text Node
\draw (252,120) node [anchor=north west][inner sep=0.75pt]  [rotate=-212.33]  {$\circlearrowright $};
\end{tikzpicture}
\ee 
This leads to Claim \ref{claim:3}.
Note that the affine braid group action discussed here is mirror to the braid group action on $B$-branes generated by twist functors along spherical objects, as constructed in \cite{seidel2001braid}.

\bigskip

Beyond the affine braid group action, the symmetry action described in \S\ref{sec:symmetry} also induces auto-equivalences. Specifically, all the diffeomorphisms discussed in \S\ref{sec:symmetry} are holomorphic symplectomorphisms, meaning they satisfy  
\be 
f:\X \to \X', \qquad f^*\Omega_J = \Omega_J~, 
\ee  
where $\Omega_J$ is the holomorphic symplectic structure. As explained earlier, the generators of the braid group $B_3$ and the sign-flip group define symplectomorphisms between character varieties at different parameter points, even though the cubic equation \eqref{cubic_eqn} remains unchanged. Consequently, these symplectomorphisms induce auto-equivalences
\be 
f: \ABrane(\X, \omega_\X) \xrightarrow[]{\cong} \ABrane(\X', \omega_{\X'})~.
\ee  

To observe the symmetry action on $\Rep(\SH)$ discussed in \S\ref{sec:symmetry} from the perspective of the $A$-brane category, we focus on the symmetry action that preserves the parameters of the target space $\X$. A subgroup of $B_3$ that keeps the parameters $t_j$ invariant is generated by $\langle \tau_+^2, \tau_-^2, \sigma^2 \rangle$. Since this subgroup action preserves the holomorphic symplectic structure $\Omega_J$, it induces the group of auto-equivalences of the $A$-brane category $\ABrane(\MS, \omega_\MS)$, and hence also of $\Rep(\SH)$.

\subsection{Non-compact \texorpdfstring{$(A,B,A)$}{(A,B,A)}-branes and polynomial representations}\label{sec:brane-poly}

Having discussed the symmetry actions on the category, we demonstrate the explicit correspondence of objects and morphisms between the two categories under the functor \eqref{functor-explicit}.
In this subsection, we study a brane for the polynomial representation \eqref{pol}. The polynomial representation, being infinite-dimensional, is naturally expected to correspond to a non-compact Lagrangian $A$-brane with infinite volume.  Furthermore, we extend the polynomial representation from both the $D_4$ root system and a geometric perspective. Specifically, we relate the corresponding branes to the 24 lines in the target space $\X$, further elucidating the structure of this correspondence.

To identify the $A$-brane associated with the polynomial representation in \eqref{pol action of xyz}, we consider the classical limit $q \to 1$. This procedure is justified in \cite{Gukov:2022gei} for the quantum torus algebra
and for the spherical DAHA of type $A_1$.  In this limit, the lowering operator $L_n$ becomes independent of $n$, so we denote $L^{(c)} = L_n |_{q \to 1}$. In the classical limit, certain operators act as scalar multiplications:
\be\label{line}
\begin{aligned}
    \text{pol}(y) &= -t_1 t_3 - t_1^{-1} t_3^{-1}, \\
    -\text{pol}(L^{(c)}) &=0= \text{pol}(x) + \frac{\text{pol}(z)}{t_1 t_3} + t_1^{-1}t_2^{-1}+ t_1^{-1}t_2 + t_3^{-1} t_4^{-1} +  t_3^{-1} t_4.
\end{aligned}
\ee
The second equation holds because the lowering operator becomes null in the classical limit $q \to 1$, as demonstrated in \eqref{lowering}. Geometrically, this describes the support of the brane $\brane_\bfP$, for the polynomial representation:
\be\label{Brane_P}
\bfP=\{y=-t_1t_3 - t_1^{-1}t_3^{-1},z=-t_1t_3 x-t_1 t_4^{-1}-t_1 t_4 -t_2^{-1}t_3-t_2t_3\}
\ee
Since it is holomorphic in complex structure $J$, $\Omega_J|_\bfP=0$. Consequently, the brane $\brane_\bfP$ associated to the polynomial representation is an $(A, B, A)$-brane.

From the perspective of representation theory, we can construct new polynomial representations by applying the symmetries of $\SH$ to \eqref{pol2}. The supports of these branes can thus be found as the images of $\bfP$ in \eqref{Brane_P} under the symmetry actions. In this way, we obtain additional examples of the correspondence between non-compact $(A, B, A)$-branes and polynomial representations.

\subsubsection*{Weyl group $W(D_4)$ and $\bZ_3$ actions on the line}

We begin by applying the Weyl group action $W(D_4)$ and the cyclic group $\bZ_3$ (see \eqref{A_3action}) to $\bfP$. Since both group actions act linearly on the coordinates $(x, y, z)$, they map the line $\bfP$ to other lines. As a result, 24 distinct lines can be generated from $\bfP$, as detailed in Appendix \ref{app:24lines}. 

A notable feature of these lines is that they all lie in planes where one of the coordinates $x$, $y$, or $z$ remains constant. Thus, the slope of each line can be described by a single complex number. To formalize this, we define the slope as follows: if $x$ is constant along the line, the slope is given by $\frac{dy}{dz}$; if $y$ is constant, the slope is $\frac{dz}{dx}$; and if $z$ is constant, the slope is $\frac{dx}{dy}$.

Furthermore, we denote these slopes by $\bS_x$, $\bS_y$, and $\bS_z$, corresponding to lines in planes where $x$, $y$, or $z$ is constant, respectively. Interestingly, these sets have a natural interpretation in terms of $\SO(8)$ representation theory:
\be\label{Slope}
\begin{aligned}
        \bS_x &= \{-\boldsymbol{t}^w \mid w \in \sfP(\mathbf{8}_V)\}, \\
        \bS_y &= \{-\boldsymbol{t}^w \mid w \in \sfP(\mathbf{8}_S)\}, \\
        \bS_z &= \{-\boldsymbol{t}^w \mid w \in \sfP(\mathbf{8}_C)\},
\end{aligned}
\ee
where $\boldsymbol{t}^w \equiv t_1^{w_1} t_2^{w_2} t_3^{w_3} t_4^{w_4}$, and $\sfP(\mathbf{8}_V)$, $\sfP(\mathbf{8}_S)$, and $\sfP(\mathbf{8}_C)$ are the weights of the $\SO(8)$ vector, spinor, and cospinor representations, respectively. These weights correspond to the shortest weights in the $D_4$ weight lattice (see \eqref{D4weightsystem} for the conventions used). 

Consequently, the 24 non-zero roots of the $D_4$ root system are organized into three distinct sets, each in one-to-one correspondence with the weights of the $\SO(8)$ vector, spinor, and cospinor representations. Furthermore, a similar argument based on the pullback of $\Omega_J$ confirms that the branes corresponding to these 24 lines are $(A, B, A)$-branes in the cubic surface.

The remaining question is whether all lines of type $(A, B, A)$ have been identified in the target space. A classical result in algebraic geometry establishes that a smooth \emph{projective} cubic variety contains exactly 27 lines \cite{cayley1849}. Among them, three lines lie at infinity $xyz=0$ and intersect each other to form a triangle. (For a modern and elementary derivation of the result, see \cite[Section 7]{reid1988}.) To focus on the lines in the \emph{affine} cubic surface, we exclude these three lines at infinity. This exclusion leaves precisely 24 lines in the affine cubic surface, which are precisely the ones obtained via the symmetry actions described above. Therefore, we show Claim \ref{claim:1}.

As a result, the lines and their corresponding infinite-dimensional representations can be labeled using the shortest weights $w \in \sfP(D_4)$. We denote the lines as $\bfP_w$ and the polynomial representations as $\scP^w$. In addition, the raising and lowering operators in the representation $\scP^w$ are denoted as $R^w_n$ and $L^w_n$, respectively, while the associated Askey-Wilson polynomials are denoted by $P^w_n$.

Using this convention, the polynomial representation discussed in \S\ref{sec:polyrep} is written as $\mathcal{P}^{w=(1,0,1,0)}$, and the line described in \eqref{line} is labeled as $\mathbf{P}_{(1,0,1,0)}$, corresponding to its slope $-t_1 t_3$. However, for simplicity, we often omit the explicit weight $w = (1,0,1,0)$ when referring to the polynomial representation in \S\ref{sec:polyrep}, unless it is necessary for clarity.

\paragraph{Truncation criterion for finite-dimensional representations}
The relation between lines and the shortest weights provides a useful criterion to determine whether a finite-dimensional representation, as given in \eqref{finite dimensional representation}, can be obtained by truncating an infinite-dimensional representation supported on a line.

Recall that finite-dimensional representations obtained by truncating the infinite-dimensional representation $\scP^{(1,0,1,0)}$ with weight $w = (1,0,1,0)$ are classified by the conditions in \eqref{generic_shortening_cond}. These representations correspond to six roots of the $D_4$ root system, explicitly given by  
\be
\big\{ (2, 0, 0, 0), (0, 0, 2, 0), (1, 1, 1, 1), (1, 1, 1, -1), (1, -1, 1, 1), (1, -1, 1, -1) \big\} \subset \sfR(D_4).
\ee  
A direct calculation shows that these are the only roots that meet the condition $\langle r, w \rangle > 0$. Since the actions of the Weyl group preserve the inner product, applying these actions to both the weight $w$ and the identified roots guarantees that the result holds for any other weight $w^\prime$.
Therefore, we obtain the following claim:
\begin{claim}\label{claim-truncation}
    The finite-dimensional representation, labeled by a root $r \in \sfR(D_4)$, can be obtained by truncating the polynomial representation $\scP^w$ if and only if $\langle r, w \rangle > 0$, where $\langle \cdot, \cdot \rangle$ is the standard Euclidean inner product. 
\end{claim}

\subsubsection*{Braid group action on the line}

Next, we apply the action of the braid group $B_3$ to $\mathbf{P}_w$ in order to identify additional $(A,B,A)$-branes. As discussed in \S\ref{sec:symmetry}, the parameters of the cubic surface can change under the braid group action, and therefore, the braid group action is \emph{not}, in general, a symplectomorphism from the surface to itself. However, as shown below, the braid group action determines an algebraic equation that specifies the support of a non-compact $(A,B,A)$-brane. Thus, we can consistently regard such $(A,B,A)$-branes as objects in the same target space $\mathfrak{X}$. 

Let us first consider the braid group action on $\bfP_{(1,0,1,0)}$.
We denote the image of the line $\bfP_{(1,0,1,0)}$ after applying $\tau_+^n$ for some $n\in \bZ$ as
\be
    \tau_+^n(\bfP_{(1,0,1,0)})=\big\{(x,y,z)\in\bC^3|y=f_n(x),z=g_n(x)\big\},
\ee
where $f_n(x)$ and $g_n(x)$ are functions parametrized by $t_j$. The initial conditions are 
\begin{equation}\label{f0_g0}
\begin{aligned}
    f_0 &= -t_1t_3 - t_1^{-1}t_3^{-1},\\
    g_0 &= -t_1t_3x - t_1t_4^{-1} - t_1t_4 - t_2^{-1}t_3 - t_2t_3,
\end{aligned}
\end{equation}
as specified in \eqref{Brane_P}. Since it is holomorphic in $J$, the action of the braid group generates additional $(A,B,A)$-branes.

To obtain a generic expression for the family of polynomials $f_n(x)$ and $g_n(x)$, we define an operator $\hat{\zeta}$ by 
\be
\hat{\zeta}:f(x;t_1,t_2,t_3,t_4)\mapsto f(x;t_1,t_2,t_4,t_3),
\ee
where $f$ is any function of $x$ parametrized by $t_j$, and $\hat{\zeta}^2 = \mathrm{id}$. Then, using \eqref{Classical_braid_action_sdaha}, the recursive relations for $f_n$ and $g_n$ are given by
\begin{equation}\label{recursive_relations}
\begin{aligned}
        f_{n+1}(x)&=x \hat{\zeta}(f_{n})(x)-\hat{\zeta}(g_{n})(x)-\hat{\zeta}(\theta_3),\\
        g_{n+1}(x)&=\hat{\zeta}(f_{n})(x).
\end{aligned}
\end{equation}
Eliminating $f_n$, we find
\begin{equation}
    \hat{\zeta}^{n+1} (g_{n+1})(x)=x \hat{\zeta}^{n} (g_{n})(x)-\hat{\zeta}^{n-1} (g_{n-1})(x)-\hat{\zeta}^{n+1}(\theta_3).
\end{equation}
This recursion relation can be solved by constructing the generating function for $\hat{\zeta}^n(g_n)$:
\begin{equation}
    G(x,u)=\sum_{n=0}^{\infty} \hat{\zeta}^{n}(g_{n}) u^n,
\end{equation}
where $u$ is a formal variable. Using standard techniques for solving recursion relations, we find
\begin{equation}
    G(x,u)=\frac{\theta _3 u^2+\theta _2 u+(1-u^2)(uf_0-g_0)
}{(u^2-1) (u^2-u x+1)},
\end{equation}
with $f_0,g_0$ given by \eqref{f0_g0}.
with $f_0$ and $g_0$ given in \eqref{f0_g0}. Expanding in powers of $u$, one can obtain $\hat{\zeta}^n(g_n)$, and hence $g_n$, for any $n \in \mathbb{Z}_{\geq 0}$. For example, the first few $g_n$ are
\be
\begin{aligned}
    g_1&= \hat{\zeta}(g_0)x-\hat{\zeta}(f_0)  - \theta_3,\\
    g_2&=g_0x^2-(f_0+\theta _2)x-\theta_3,\\
    g_3&=\hat{\zeta}(g_0) x^3- (\hat{\zeta}(f_0)+\theta _3)x^2- (2 \hat{\zeta}(g_0)+\theta
   _2)x+\hat{\zeta}(f_0),\\
   g_4&=g_0
   x^4- (f_0+\theta _2)x^3- (3 g_0+\theta _3)x^2+ (2 f_0+\theta _2)x+g_0.
\end{aligned}
\ee
The corresponding expressions for $f_n$ can be derived using the second equation in \eqref{recursive_relations}. By applying the Weyl group and the $\mathbb{Z}_3$ permutation symmetries, the loci of the $A$-branes can be determined starting from $\mathbf{P}_w$ for any weight $w$.

As a remark, in the limit \eqref{A1Parameters} of the DAHA of type $A_1$, the generating function simplifies even more. In this limit, the action of $\hat{\zeta}$ becomes trivial, and the generating function reduces to
\begin{equation}
   G(x,u) = -\frac{(u - x)t^2 + u}{t (u^2 - ux + 1)}.
\end{equation}
Moreover, by taking the additional limit $t \to i$, the generating function further simplifies to:
\begin{equation}
    G(x,u) = i x G_C(x,u),
\end{equation}
where $G_C(x,u) = \frac{1}{u^2 - ux + 1}$ is the generating function for the Chebyshev polynomials of the second kind, as mentioned in \cite{Gukov:2022gei}.

\subsection{Compact \texorpdfstring{$A$}{A}-branes and finite-dimensional representations}\label{Sec:Brane_Finite_Rep}

In this subsection, we present evidence for Claim \ref{claim:2}, which asserts the derived equivalence between the category of compact $A$-branes and the category of finite-dimensional representations of $\SH$. This equivalence involves matching both objects and morphisms.

To match objects, we establish a correspondence between compact Lagrangian $A$-branes and finite-dimensional representations of $\SH$ by comparing the shortening conditions described in \eqref{finite dimensional representation} with the existence conditions for $A$-branes (referred to as $A$-brane conditions) derived from the dimension formula \eqref{dimformula}. 

To illustrate this correspondence, we consider a case study that covers all possible ``generic'' configurations of Hitchin fibration, as classified in Table \ref{tab:fiber_sing_classification2}.  We begin by considering the case where $\hbar$ is real, with the symplectic form given by $\omega_\mathfrak{X} = \frac{\omega_K}{\hbar}$. Since Hitchin fibration \eqref{Hitchin-fibration} is completely integrable, each fiber is Lagrangian of type $(B, A, A)$. Consequently, every component of a Hitchin fiber is Lagrangian with respect to $\omega_\mathfrak{X}$ for any value of $\tgamma$. For all configurations of Hitchin fibration, we show that the shortening conditions in \eqref{finite dimensional representation} have a precise geometric interpretation, corresponding to $A$-branes supported on the components of a Hitchin fiber.

We then turn to matching morphisms. Specifically, two distinct $A$-branes supported on components of a singular fiber exhibit non-trivial morphisms when they intersect at a singular fiber, forming bound states. Our analysis identifies these bound states of compact $A$-branes and their corresponding $\SH$-modules using an extension. This provides additional evidence for the equivalence of morphism structures as described in \eqref{eq:functor}.

Lastly, we generalize our analysis to the case where $\hbar$ is arbitrary. In this broader context, although a Hitchin fiber is no longer Lagrangian with respect to $\omega_\mathfrak{X}$, cycles suspended between singular fibers can become Lagrangian under appropriate choices of $(\talpha, \tgamma)$. The corresponding matching of objects and morphisms in this setting is carried out analogously, reinforcing Claim \ref{claim:2}.

\subsubsection{Generic fibers of the Hitchin fibration}\label{sec:genericfiber}

We begin by considering the case where $\hbar$ is real. In this case, a generic fiber $\bfF$ is Lagrangian with respect to $\Im \Omega =\frac{\omega_K}{\hbar}$,
and thus $\bfF$ can serve as the support of an $A$-brane $\brane_{\bfF}$. Thus, we consider a brane $\brane_\bfF$ supported on a generic fiber and the corresponding representation $\scF$. The story is parallel to \cite[\S2.6.1]{Gukov:2022gei}.

The branes are labeled by positions on the Hitchin base $\cB_H$. Moreover, the flatness condition \eqref{deformed-flat} for the line bundle $\cL'$ supporting the $A$-brane $\brF$ takes the form
\begin{equation}
F'_{\bfF} + B\big|_{\bfF} = 0~.
\end{equation}
Since $\F$ is topologically a two-torus, a flat $\Spin^c$ bundle $\cL' \otimes K_\L^{-1/2}$ on $\brF$ can have nontrivial $\U(1)^2$ holonomies with spin structure \cite{Gukov:2008ve}.  

Consequently, the branes $\brF^\lambda$ are parametrized by 
\be 
\lambda = (x_m,y_m) \in \C^\times \times \C^\times~,
\ee 
where the absolute values $(|x_m|,|y_m|)$ specify the position on the base $\cB_H$, while the angular phases encode the $\U(1)^2$ holonomies and spin structures. Concretely, the angular phase of each $\U(1)$ factor determines both the holonomy and the choice of spin structure $\Z_2$ along a one-cycle of $\bfF$ via the exact sequence
\begin{equation}
1 \longrightarrow \Z_2 \longrightarrow \U(1) \longrightarrow \U(1) \longrightarrow 1~.
\end{equation}
We assign the Ramond spin structure to the element $+1 \in \Z_2$ and the Neveu--Schwarz spin structure to the element $-1 \in \Z_2$.

\paragraph{Object Matching}
Combining the dimension formula \eqref{dimformula} and the volume formula for a generic fiber \eqref{Fibervolume}, we have
\be\label{A-brane_condition_F}
m:=\dim \Hom(\brane_{\bfF}, \Bcc) = \int_{\bfF} \frac{\omega_I}{2\pi\hbar} = \frac{1}{\hbar}.
\ee
Consequently, the $A$-brane $\brane_{\bfF}$ can exist if and only if $\frac{1}{\hbar} = m \in \bZ_{>0}$. Physically, this condition can be interpreted as the Bohr-Sommerfeld quantization condition. More specifically, if one interprets $(\bfF, \omega_I)$ as a symplectic manifold representing a phase system in classical mechanics, the dimension formula \eqref{A-brane_condition_F} implies that the volume of the symplectic manifold $\bfF$, measured in units of $\hbar$, determines the dimension of the quantum Hilbert space. Consequently, to ensure the existence of a well-defined Hilbert space, the volume of the symplectic manifold must be quantized.

To ensure the existence of $(\Bcc, \brF^\lambda)$-strings, we require that $q$ be a primitive $m$-th root of unity, while $t$ may remain generic. Under this assumption, the action of $\SH$ in the generalized polynomial representation $\pol_{y_1}$ in \eqref{poly-rep-y1} commutes with the element $X^m - x_m$ for any $x_m \in \C^\times$, since the shift operator $\varpi$ acts trivially on this polynomial.

As a result, the ideal $(X^m - x_m)$ is invariant under $\pol_{y_1}$, and the quotient space
\be 
\scF^{\lambda}_{m} = \scP^{y_1}/(X^m - x_m)
\ee 
inherits a representation of $\SH$.
Moreover, when $X^m = x_m$, the Taylor expansion of any denominator in the multiplicative system $\wt M$ truncates automatically. Hence, $\scF^{\lambda}_{m}$ becomes a finite-dimensional representation of dimension $m$, parametrized by $\lambda = (x_m, y_m)$, where $y_1$ is any $m$-th root of $y_m$.
Therefore, when $q$ is a primitive $m$-th root of unity, this family of finite-dimensional modules $\scF_m^{\lambda}$ corresponds to a family of branes $\brane^{\lambda}_\bfF$ supported on generic Hitchin fibers:
\be 
\Hom(\Bcc, \brF^\lambda) \cong \scF^{\lambda}_{m}~.
\ee 
As mentioned above, the parameters $\lambda=(x_m,y_m)$ encode the position of the Hitchin base and the $\U(1)^2$ holonomies of the brane $\brF^\lambda$.

Setting $y_1 = 1$, the representation restores the $\bZ_2$ symmetry $X\leftrightarrow X^{-1}$, and a family of finite-dimensional representations labeled by $x_m$ can be constructed as the quotient of the ordinary polynomial representation \eqref{pol2}:
\be \label{general_AW}
\scF_m^{(x_m)}=\scP/(X^m + X^{-m} - x_m - x_m^{-1})~.
\ee
As $Q_m(X, x_m)$ is of degree $m$, $\scF_m^{(x_m)}$ is indeed $m$-dimensional, which matches the analysis of the $A$-brane condition \eqref{A-brane_condition_F}. 
This corresponds to the brane $\brF^{(x_m,+)}$ where the flat $\Spin^c$-bundle has trivial holonomy and a Ramond spin structure along the other $(0,1)$-cycle.

A special case of this family can be obtained by using Askey-Wilson polynomials with raising/lowering operators. In fact, 
the condition $q^m = 1$ precisely matches with the shortening condition \eqref{finite dimensional representation} with $r=(0,0,0,0)$. Under this shortening condition, we obtain the finite-dimensional $\SH$-module
\be 
\scF_m=\scP/(P_m)~,
\ee 
where 
\bea\label{AskeyWilson_nth_pol}
P_m(X,q=e^{\frac{2\pi i}{m}},\boldsymbol{t})&=X^m+X^{-m}+ F_m(\boldsymbol{t}),\cr 
F_m(\boldsymbol{t})&=\frac{(t_3^mt_4^m+t_3^m/t_4^m)(t_1^{2m}-1)+(t_1^mt_2^m+t_1^m/t_2^m)(t_3^{2m}-1)}{(t_1t_3)^{2m}-1}.
\eea
Moreover, such representations can be obtained from the other polynomial representations $\scP^w$ (see Claim \ref{claim-truncation}) in a similar way 
\be 
\scF^w_m=\scP^w/(P^{w}_m)~.
\ee 
As a remark, for the generic choice of $\boldsymbol t$, the $\SH$-modules $\scF^w_m$ are \emph{not} isomorphic to each other, as each representation $\scF^w_m$ uses different raising/lowering operators. These representations are special cases of \eqref{general_AW}.

\subsubsection{At global nilpotent cone of type \texorpdfstring{$I_{0}^{*}$}{I0*}}\label{sec:I0*}
 
In the following sections, we will establish the more intricate and intriguing correspondence between $A$-branes supported on components of singular fibers and finite-dimensional $\SH$-modules.
As we continue to assume that $\hbar$ is real, here we begin with the $I_0^{*}$ singular fiber that appears when $\tbeta_j = 0=\tgamma_j$, as considered in \S\ref{sec:wallcrossing}.  As illustrated in Figure \ref{fig:nilpotent_cone}, the irreducible components of the $I_0^{*}$ singular fiber consist of both the moduli space $\bfV$ of $\SU(2)$-bundles on $C_{0,4}$ and the exceptional divisors $\bfD_j$ ($j=1,2,3,4$).

\paragraph{Object Matching}
We first establish the correspondence between objects by comparing the $A$-brane conditions for the irreducible components in the $I_0^*$ singular fiber with the shortening conditions of \SH, as classified in \eqref{finite dimensional representation}. 

We start with the brane $\brane_\bfV$ supported on the irreducible component $\bfV$. By combining the dimension formula \eqref{dimformula} with the volume formula for $\brane_{\bfV}$ given in \eqref{vol-middle-chamber2}, we find that the dimension of the corresponding representation is:
\begin{equation}\label{dimV}
    k := \dim\Hom(\brane_{\bfV}, \Bcc) = \int_{\bfV} \frac{F+B}{2\pi} = \int_{\bfV} \frac{\omega_I}{2\pi\hbar} = \frac{2\talpha_1}{\hbar}.
\end{equation}
The $A$-brane $\brane_{\bfV}$ can exist if and only if $k$ is a positive integer. Using \eqref{tj} and noting that $q = e^{2\pi i \hbar}$, this requirement translates to the condition:
\be 
q^k = t_1^{-2}~,
\ee
which corresponds to the shortening condition \eqref{finite dimensional representation}, with the $D_4$ root $r = (2, 0, 0, 0)$. When the shortening condition is satisfied, the corresponding $k$-dimensional representation, denoted by $\scV_k$, can be explicitly constructed as the quotient:
\begin{equation}\label{Vk_rep}
    \scV_k = \scP / (P_k).
\end{equation}

Using the same approach, one can identify an $\SH$-module $\scD^{(j)}_{\ell_j}$ corresponding to a brane $\brane_{\bfD_j}$ supported on each exceptional divisor $\bfD_j$. The dimension of the morphism space is given by:
\begin{equation}
\ell_j := \dim\Hom(\brane_{\bfD_j}, \Bcc) = \int_{\bfD_j} \frac{F+B}{2\pi} = \int_{\bfD_j} \frac{\omega_I}{2\pi\hbar} = \frac{\text{vol}(\bfD_j)}{\hbar}.
\end{equation}
The results are summarized in Table \ref{Tab:Matching_A-brane_Rep_I0}, based on the volume formulas provided in \eqref{vol-middle-chamber2}. However, these representations \emph{cannot}  be constructed as quotients of a single polynomial representation $\scP^w$. Instead, one must apply Claim \ref{claim-truncation} to determine a compatible polynomial representation for each case.

For example, as shown in Table \ref{Tab:Matching_A-brane_Rep_I0}, the $A$-brane condition for $\bfD_1$ leads to the shortening condition $q^{\ell_1} = \boldsymbol{t}^{-r}$, with $r = (-1, -1, -1, -1)$. This condition does not align with the shortening conditions \eqref{generic_shortening_cond} associated with the polynomial representation $\scP^{(1,0,1,0)}$. However, since $\langle r, w \rangle > 0$ for $w = (-1, 0, -1, 0)$, Claim \ref{claim-truncation} ensures that the $\SH$-module $\scD^{(1)}_{\ell_1}$ can be constructed as the quotient module:
\begin{equation}
\scD^{(1)}_{\ell_1} := \scP^{(-1,0,-1,0)} / (P^{(-1,0,-1,0)}_{\ell_1}).
\end{equation}

\begin{table}[ht]
    \centering
    \renewcommand{\arraystretch}{1.3}
    \begin{tabular}{c|c|c|c}
\hline \text { finite-dim rep } & \text { shortening condition }  &\text {$A$-brane }& \text {$A$-brane condition } \\
\hline $\scF_m^{(x_m)}$ & $q^m=1$  &$\brane^{(x_m)}_{\bfF}$& $m=\frac{1}{\hbar}$ \\
\hline $\scV_{k}$ & $q^{k}=t_1^{-2}$ &$\brane_{\bfV}$& $k=\frac{2 \talpha_1}{ \hbar}$ \\
\hline $\scD^{(1)}_{\ell_1}$ & $q^{\ell_1}=t_1t_2t_3t_4$  &$\brane_{\bfD_1}$& $\ell_1=\frac{1}{\hbar}-\frac{\talpha_1+\talpha_2+\talpha_3+\talpha_4}{\hbar}$ \\
\hline $\scD^{(2)}_{\ell_2}$ & $q^{\ell_2}=t_1t_2t_3^{-1}t_4^{-1}$  &$\brane_{\bfD_2}$& $\ell_2=\frac{-\talpha_1-\talpha_2+\talpha_3+\talpha_4}{\hbar}$ \\
\hline $\scD^{(3)}_{\ell_3}$ & $q^{\ell_3}=t_1t_2^{-1}t_3t_4^{-1}$  &$\brane_{\bfD_3}$& $\ell_3=\frac{-\talpha_1+\talpha_2-\talpha_3+\talpha_4}{\hbar}$ \\
\hline $\scD^{(4)}_{\ell_4}$ & $q^{\ell_4}=t_1t_2^{-1}t_3^{-1}t_4$  &$\brane_{\bfD_4}$& $\ell_4=\frac{-\talpha_1+\talpha_2+\talpha_3-\talpha_4}{\hbar}$ \\
\hline
    \end{tabular}
    \caption{A summary of finite-dimensional $\protect\SH$-modules  with their shortening conditions and the corresponding $A$-brane configurations at the $I_0^*$ singular fiber, under the assumption $|q|=1$.}
    \label{Tab:Matching_A-brane_Rep_I0}
\end{table}

Under the braid group action, the homology cycles transform as described in \eqref{Braidgroup_Action}. From the perspective of representation theory, the corresponding $\SH$-modules are exchanged as follows:
\be
\begin{array}{rllll}
\tau_{+}: \scD_{\ell_3}^{(3)} & \leftrightarrow \scD_{\ell_4}^{(4)} & \text { and } & \scD_{\ell_1}^{(1)}, \scD_{\ell_2}^{(2)} & \text { are invariant, } \\
\tau_{-}: \scD_{\ell_2}^{(2)} & \leftrightarrow \scD_{\ell_4}^{(4)} & \text { and } & \scD_{\ell_1}^{(1)}, \scD_{\ell_3}^{(3)} & \text { are invariant, } \\
\sigma: \scD_{\ell_2}^{(2)} & \leftrightarrow \scD_{\ell_3}^{(3)} & \text { and } & \scD_{\ell_1}^{(1)}, \scD_{\ell_4}^{(4)} & \text { are invariant. }
\end{array}
\ee
In particular,  both $\scV_k$ and $\scD^{(1)}_{\ell_1}$ are invariant under the entire braid group action. This phenomenon has been observed in \cite{Gukov:2022gei}, where the braid group action reduces to the $\PSL(2,\bZ)$ action in the case of DAHA of type $A_1$. This action leads to the relationship of the DAHA modules with the modular tensor categories \cite{Shan:2023xtw,Shan:2024yas}.

Beyond a single irreducible component of the $I_0^*$ singular fiber, one can consider a brane supported on a union of several components, denoted by $\bfN_{r}$, which is Lagrangian with respect to $\omega_\X=\omega_K/\hbar$. The homology class of $\bfN_{r}$ is represented with the standard basis by
\be\label{generic_cycle}
[\bfN_{r}] = a_0 [\bfD_1] + a_1 [\bfD_2] + a_2 [\bfD_3] + a_3 [\bfD_4] + a_4 [\bfV],
\ee
where $a=(a_0, a_1, a_2, a_3, a_4)$ and the subscript $r$ are related by
\be \label{root-basis}
r = - a_0 \theta + \sum_{i=1}^4 a_i e^i~,\qquad   (a_i=\{0,1,2\})~.
\ee
Here, $e^i$ are the simple roots, and $\theta$ is the highest root in the $D_4$ root system \eqref{simpleroots}.
Since $D_4$ is a simply-laced root system, every root of $D_4$ has squared length two $\langle r,r\rangle=2$. Therefore, the cycle $\bfN_{r}$ corresponding to $r\in \sfR(D_4)$ also has a self-intersection number of minus two $[\bfN_{r}]\cdot [\bfN_{r}]=-2$. 
Using the same approach, we determine the $A$-brane condition for $\bfN_{r}$. Applying the dimension formula \eqref{dimformula} and the volume formula \eqref{vol-middle-chamber3}, the $A$-brane condition is given by
\be\label{GenericVolume}
m := \dim\Hom(\brane_{\bfN_{r}}, \Bcc) = \int_{\bfN_{r}} \frac{F + B}{2\pi} = \int_{\bfN_{r}} \frac{\omega_I}{2\pi\hbar} = \frac{a_0 + \left(\sum_{i=1}^4 a_i e^i - a_0 \theta\right) \cdot \talpha}{\hbar}.
\ee
The corresponding shortening condition takes the form
\be\label{General_A_brane}
q^{m} = \boldsymbol{t}^{-r},
\ee
where $\boldsymbol{t}^{r} = t_1^{r_1} t_2^{r_2} t_3^{r_3} t_4^{r_4}$. Thus, we write the corresponding $\SH$-module
\be \label{scN}
\scN^{r,w}_m =\scP^w /(P^w_m)~,
\ee 
where $\langle r, w\rangle>0$ due to Claim \ref{claim-truncation}.
As a result, if we consider all the cycles with self-intersection number minus two, all 24 roots in the $D_4$ root system $\sfR(D_4)$ are exhausted. In this way, the $A$-brane conditions in \eqref{GenericVolume} recover all the shortening conditions classified in \eqref{finite dimensional representation}.

If the self-intersection number of a cycle $\bfN_{r}$ is \emph{not} minus two, the condition \eqref{General_A_brane}  does not fit into the shortening conditions for polynomial representations. For example, in the case of the $A$-brane $\brane_{\bfN_{r}}$ with $a = (0,2,0,0,0)$, the condition $q^{m} = \boldsymbol{t}^{-2e^1}$ does not align with any shortening condition for polynomial representations because $2e^1$ is not an element of the $D_4$ root system.\footnote{Here, we assume that $m$ is odd. Otherwise, it reduces to $q^{m/2} = \boldsymbol{t}^{-e^1}$, which implies that the corresponding representation is simply a direct sum $\scD^{(2)}_{m/2}\oplus \scD^{(2)}_{m/2}$.} This suggests that the existence of an $A$-brane requires the self-intersection number of its supporting cycle to be precisely minus two.

\paragraph{Morphism Matching}

We now compare the morphism spaces in the two categories under consideration.
In the category of $A$-branes (or the Fukaya category), the endomorphism
algebra of a compact Lagrangian $A$-brane is isomorphic to the de~Rham
cohomology of its support:
\be
\operatorname{Hom}^*\!\left(\brane_{\bfL}, \brane_{\bfL}\right)
\;\cong\;
H^*(\bfL;\bC).
\ee

For two distinct $A$-branes $\brane_{\bfL_0}$ and $\brane_{\bfL_1}$, nontrivial morphisms
arise only when their supports intersect transversely. Such morphisms admit a physical
interpretation as bound states of open strings stretching between the two branes,
and are mathematically captured by Floer cohomology
\cite{floer1988morse,floer1989witten}.

Since the target space $\X$ is hyperk\"ahler of complex dimension two, the morphism
spaces satisfy a Serre duality of the form
\be
\operatorname{Hom}^k\!\left(\brane_{\bfL_0}, \brane_{\bfL_1}\right)
\;\cong\;
\operatorname{Hom}^{2-k}\!\left(\brane_{\bfL_1}, \brane_{\bfL_0}\right)^{*}.
\ee
In the examples relevant to our discussion, two distinct compact irreducible $A$-branes intersect
at most at a single point. This geometric constraint forces the degree to obey
$k=2-k$, and hence only degree-one morphisms can be nonvanishing. As a result,
\be
\operatorname{Hom}^1\!\left(\brane_{\bfL_0}, \brane_{\bfL_1}\right)
\;\cong\;
\bC\langle q\rangle,
\ee
where $q=\bfL_0\cap\bfL_1$ denotes the unique transverse intersection point.
In all other degrees, the morphism space is trivial.

In the following, we analyze the bound states of $A$-branes and identify the
corresponding morphisms in the $\SH$-module category. As we will show, there is a
precise correspondence between morphisms in the $A$-brane category and those in the
associated representation category.

\bigskip

Let $\brane_{\bfN_r}$ and $\brane_{\bfN_{r^\prime}}$ be $A$-branes supported on components of the $I_0^*$ singular fiber. The $A$-brane conditions are given by:
\be
    m = \dim\Hom(\brane_{\bfN_r}, \Bcc), \quad m^\prime = \dim\Hom(\brane_{\bfN_{r^\prime}}, \Bcc).
\ee
As in \eqref{scN}, we denote the corresponding $\SH$-modules by $\scN^{r,w}_m$ and $\scN^{r^\prime,w^\prime}_{m^\prime}$ where the shortening conditions are given by $q^m=\boldsymbol{t}^{-r}$ and $q^{m^\prime} = \boldsymbol{t}^{-r^\prime}$, respectively.
\begin{claim}
    If $\brane_{\bfN_r}$ and $\brane_{\bfN_{r^\prime}}$ intersect at a point $ q$, then the morphism space is one-dimensional
    \be
\Hom^1(\brane_{\bfN_r},\brane_{\bfN_{r^\prime}})=\bC\langle q\rangle~.
    \ee 
Then, the roots $r,r^{\prime}\in\sfR(D_4)$ satisfy
$\langle r, r^\prime \rangle = -2$. Moreover, there exists a weight $w$ such that  $\langle w, r \rangle > 0$, $\langle s_{r}(w), r^\prime \rangle > 0$. In this setting, there exists a short exact sequence of finite-dimensional $\SH$-modules
  \be\begin{tikzcd}[cramped, sep=small]
    0\arrow[r] &\scN^{r^\prime,s_r(w)}_{m^\prime}\arrow[r] & \scN^{r+r^\prime,w}_{m+m^\prime} \arrow[r] & \scN^{r,w}_m \arrow[r]&0
   \end{tikzcd}~.
\ee
This exact sequence is uniquely determined up to isomorphism, independent of the choice of $w$. Thus, $\Ext^1(\scN^{r,w}_m,\scN^{r^\prime,s_r(w)}_{m^\prime})$ is one-dimensional.
\end{claim}

\begin{proof} 

From \eqref{generic_cycle} and \eqref{root-basis}, it is straightforward to show that the intersection number $[\bfN_r] \cdot [\bfN_{r^\prime}] = 1$ implies $\langle r, r^\prime \rangle = -2$. This indicates that the root sublattice generated by $r$ and $r^\prime$ forms an $A_2$ root lattice. Consequently, there always exists a weight $w$, as illustrated in Figure \ref{fig:rrw}, such that
\be\label{cond_rrw}
\langle w, r \rangle > 0~, \qquad \langle s_r(w), r^\prime \rangle > 0~,
\ee
where $s_r$ denotes the Weyl reflection \eqref{Weyl-reflection} with respect to $r$.

\begin{figure}
    \centering
    \includegraphics[width=0.3\linewidth]{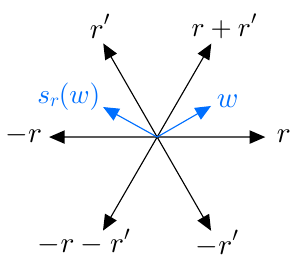}
    \caption{The roots $r$ and $r^\prime$ generate an $A_2$ root lattice, and there always exists a weight $w$ such that \eqref{cond_rrw} holds.}
    \label{fig:rrw}
\end{figure}

As shown in \eqref{scN}, when $q^m = \boldsymbol{t}^{-r}$, the $\SH$-module $\scN^{r,w}_m$ is obtained from the polynomial representation $\scP^w$ with $\langle r, w \rangle > 0$. This leads to a short exact sequence
\be\label{Exact_Sequence_general}
\begin{tikzcd}[cramped, sep=small]
   0 \arrow[r] & \Ker \pi \arrow[r] & \scP^w \arrow[r, "\pi"] & \scN^{r,w}_m \arrow[r] & 0~.
   \end{tikzcd}
\ee

Under the condition $q^m = \boldsymbol{t}^{-r}$ and $\langle r, w \rangle > 0$, the direction calculation using \eqref{raising} and \eqref{lowering} shows that the raising and lowering operators for $\scP^w$ and $\scP^{s_r(w)}$ satisfy the following relations
\bea\label{Relations_Raising_Lowering}
R^w_n &= q^{-m} R^{s_r(w)}_{n-m}~, \\
L^w_n &= q^{m} L^{s_r(w)}_{n-m}~,
\eea
for any $n\geq m$. Consequently, the raising and lowering operators in $\Ker \pi$ coincide with those in the polynomial representation $\scP^{s_r(w)}$ up to scalar multiplication. This establishes an isomorphism of $\SH$-representations $\Ker \pi \cong \scP^{s_r(w)}$, and the exact sequence \eqref{Exact_Sequence_general} can be rewritten as
\be\label{Exact_Sequence_general2}
\begin{tikzcd}[cramped, sep=small]
   0 \arrow[r] & \scP^{s_r(w)} \arrow[r, "i"] & \scP^w \arrow[r, "\pi"] & \scN^{r,w}_m \arrow[r] & 0,
   \end{tikzcd}~,
\ee
which represents an element in $\Ext^1(\scN^{r,w}_m, \scP^{s_r(w)})$. From the perspective of the $A$ model, it can be verified that the line $\bfP_{s_r(w)}$ intersects the $A$-brane $\bfN_r$ under the conditions $q^m = \boldsymbol{t}^{-r}$ and $\langle r, w \rangle > 0$.

Since $\langle s_r(w), r^\prime \rangle > 0$, the $\SH$-module $\scN^{r^\prime,s_r(w)}_{m^\prime}$ can be constructed as the quotient 
\be
\scN^{r^\prime,s_r(w)}_{m^\prime} := \scP^{s_r(w)} / (P^{s_r(w)}_{m^\prime})~,
\ee
when $q^{m^\prime} = \boldsymbol{t}^{-r^\prime}$. The simultaneous conditions $q^m = \boldsymbol{t}^{-r}$ and $q^{m^\prime} = \boldsymbol{t}^{-r^\prime}$ imply $q^{m + m^\prime} = \boldsymbol{t}^{-r - r^\prime}$. Since $s_r(r^\prime) = r + r^\prime$, it follows that
$\langle r + r^\prime, w \rangle = \langle r^\prime, s_r(w) \rangle > 0$.
By Claim \ref{claim-truncation}, another $\SH$-module can be constructed as
\be 
\scN^{r + r^\prime,w}_{m + m^\prime} = \scP^w / (P^w_{m + m^\prime}),
\ee
which gives rise to the short exact sequence
\be\label{Exact_short}
\begin{tikzcd}[cramped, sep=small]
    0 \arrow[r] & \scN^{r^\prime,s_r(w)}_{m^\prime} \arrow[r, "i_*"] & \scN^{r + r^\prime,w}_{m + m^\prime} \arrow[r, "\pi_*"] & \scN^{r,w}_m \arrow[r] & 0.
   \end{tikzcd}~
\ee
For a graphical representation, see Figure \ref{fig:ext_Nr_Nr}. 

There are always three such weights $w_1,w_2,w_3$ that give rise to the short exact sequence \eqref{Exact_short}, each corresponding to one of the eight-dimensional representations $\boldsymbol{8}_V$, $\boldsymbol{8}_S$, and $\boldsymbol{8}_C$. However, it can be shown using Schur's lemma that the short exact sequences associated with these three different weights are isomorphic.  In fact, the character of the finite-dimensional module $\scN_m^{r,w_1}$ can be explicitly computed from \eqref{pol action of xyz} as
\bea 
\Tr_{\scN_m^{r,w_1}}(\pol(x))&= -\frac{q^{m/2}-q^{-m/2}}{q^{1/2}-q^{-1/2}} \left(q^{m/2}\boldsymbol{t}^{w_1}+q^{-m/2}\boldsymbol{t}^{-w_1}\right),\cr
\Tr_{\scN_m^{r,w_1}}(\pol(y))&= -\frac{q^{m/2}-q^{-m/2}}{q^{1/2}-q^{-1/2}}  \left(q^{m/2}\boldsymbol{t}^{w_2}+q^{-m/2}\boldsymbol{t}^{-w_2}\right),\cr
\Tr_{\scN_m^{r,w_1}}(\pol(z))&=-\frac{q^{m/2}-q^{-m/2}}{q^{1/2}-q^{-1/2}}  \left(q^{m/2}\boldsymbol{t}^{w_3}+q^{-m/2}\boldsymbol{t}^{-w_3}\right),
\eea 
under the shortening condition $ q^m = \boldsymbol{t}^{-r} $. Recalling that $\boldsymbol{8}_V$, $\boldsymbol{8}_S$, and $\boldsymbol{8}_C$ are exchanged under the triality transformation \eqref{A_3action}, the above expression is manifestly symmetric under the triality action on $w_1$, $w_2$, and $w_3$. 
By Schur's lemma, this symmetry leads to an isomorphism between the short exact sequences for $ i \neq j $:
\be\begin{tikzcd}[cramped, sep=small]
   0\arrow[r]&\scN^{r^\prime,s_r(w_i)}_{m^\prime}  \arrow[d, "\sim" {anchor=south, rotate=90, inner sep=1.5mm}] \arrow[r] &\scN^{r + r^\prime,w_i}_{m + m^\prime}  \arrow[d , "\sim" {anchor=south, rotate=90, inner sep=1.5mm}] \arrow[r] & \scN^{r,w_i}_m \arrow[d , "\sim" {anchor=south, rotate=90, inner sep=1.5mm}] \arrow[r] & 0\\
    0\arrow[r] &\scN^{r^\prime,s_r(w_j)}_{m^\prime} \arrow[r] & \scN^{r + r^\prime,w_j}_{m + m^\prime}  \arrow[r] & \scN^{r,w_j}_m \arrow[r]&0\\
   \end{tikzcd}~
\ee
Consequently,  $\Ext^1(\scN^{r,w}_m,\scN^{r^\prime,s_r(w)}_{m^\prime})$ is one-dimensional.
This completes the proof.
\end{proof}

\begin{figure}[ht]
    \centering
    \includegraphics[width=0.95\linewidth]{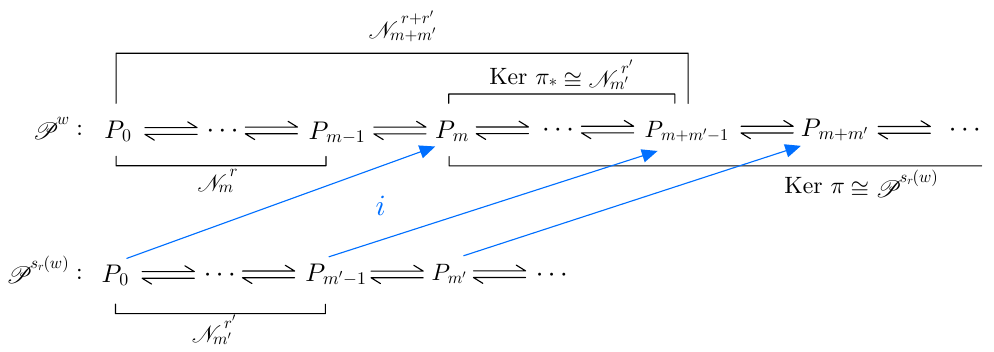}
    \caption{The polynomial representation $\scP^w$ can be understood as the extension of $\scN^{r,w}_m$ by $\scP^{s_r(w)}$ under the condition $q^{m} = \boldsymbol{t}^{-r}$ , as in \eqref{Exact_Sequence_general2}. Imposing the condition $q^{m^\prime} = \boldsymbol{t}^{-r^\prime}$ further, $\scN^{r+r^\prime,w}_{m+m^\prime}$ is constructed as the extension of $\scN^{r,w}_m$ by $\scN^{r',s_r(w)}_{m^\prime}$ as in \eqref{Exact_short}.}
    \label{fig:ext_Nr_Nr}
\end{figure}

\paragraph{Example:}   
As an explicit example, let us consider the morphism space $\Hom^*(\brane_{\bfV}, \brane_{\bfD_1})$. Since $\bfD_1$ and $\bfV$ intersect at a single point $q_1$ (see Figure \ref{fig:nilpotent_cone}), the geometric perspective predicts that the morphism space is one-dimensional:
\be\label{GeometryMorphism}
\Hom^1\left(\brane_{\bfV}, \brane_{\bfD_1}\right) \cong \bC\left\langle q_1 \right\rangle.
\ee

The $A$-brane $\brane_{\bfN_r}$, representing their bound state, is supported on $\bfN_r = \bfV \cup \bfD_1$, with the corresponding root given by:
\be
r = -\theta + e_4 = (1, -1, -1, -1).
\ee
The $A$-brane condition for $\brane_{\bfN_r}$ is evaluated as:
\be
m = \dim\Hom(\brane_{\bfN_r}, \Bcc) = \int_{\bfN_r} \frac{F + B}{2\pi} = \frac{1}{\hbar} - \frac{-\talpha_1 + \talpha_2 + \talpha_3 + \talpha_4}{\hbar},
\ee
which translates to the shortening condition $q^m = t_1^{-1}t_2t_3t_4$.

In the following discussion, we will explain how to construct the corresponding morphism space $\Ext^1(\scV_k, \scD^{(1)}_{\ell})$ as an $\SH$-module. The roots that correspond to $\bfV$ and $\bfD_1$ are $r=(2,0,0,0)$ and $r^\prime=(-1,-1,-1,-1)$, respectively. 
There are three weights that satisfy the condition \eqref{cond_rrw}:
\bea 
&\mathbf{8}_V\ni w_1=(1,-1,0,0)~, \quad
s_r(w_1)=(-1,-1,0,0)~, \cr 
&\mathbf{8}_S\ni  w_2=(1,0,-1,0)~, \quad
s_r(w_2)=(-1,0,-1,0)~, \cr 
&\mathbf{8}_C\ni  w_3=(1,0,0,-1)~, \quad
s_r(w_3)=(-1,0,0,-1)~.  
\eea 
It follows from the proof using Schur's lemma that finite-dimensional $\SH$-modules associated to these weights are all isomorphic. Therefore, it suffices to focus on the second choice $w:=w_2$.

Consider the polynomial representation $\scP^w$ where the action of lowering operators $L_m^w$ contains a factor $(q^{m}-t_1^{-2})(q^{m}-t_1^{-1}t_2t_3t_4)$, which can be verified from \eqref{lowering}. 
Suppose that we impose the two shortening conditions simultaneously 
\be \label{2-shortening}
q^k=t_1^{-2}~, \quad \textrm{and}  \quad q^{k+\ell}=t_1^{-1}t_2t_3t_4~,
\ee 
where the first condition is for $\scV_k$ (see Table \ref{Tab:Matching_A-brane_Rep_I0}). Then, 
we obtain the following short exact sequence from $\scP^w$:
\be\label{Exact_Seq_W}
\begin{tikzcd}[cramped, sep=small]
    0 \arrow[r] & \Ker \pi_* \arrow[r] & \scN^{r,w}_{k+\ell} \arrow[r, "\pi_*"] & \scV_k \arrow[r] & 0
\end{tikzcd}~.
\ee  

Since \eqref{2-shortening} implies $q^\ell= t_1t_2t_3t_4$ that is for $\scD^{(1)}_\ell$, it is easy to see that $\scD^{(1)}_{\ell} \cong \Ker \pi$ by \eqref{Relations_Raising_Lowering}, and as a result, we obtain the short exact sequence
\be\label{Exact_Sequence_D1}  
\begin{tikzcd}[cramped, sep=small]  
    0 \arrow[r] & \scD^{(1)}_{\ell} \arrow[r] & \scN^{r,w}_{k+\ell} \arrow[r] & \scV_k \arrow[r] & 0~.
\end{tikzcd} 
\ee
This represents a non-trivial element in $\Ext^1(\scV_k, \scD^{(1)}_{\ell})$.

If we switch the role of $w$ and $s_r(w)$,  we can construct a short exact sequence in the opposite direction, representing an element in $\Ext^1(\scD^{(1)}_\ell, \scV_k)$:
\be\label{Exact_Sequence_D1_2}  
\begin{tikzcd}[cramped, sep=small]  
    0 \arrow[r] & \scV_k \arrow[r] & \scN^{r,s_r(w)}_{k+\ell} \arrow[r, "\pi"] & \scD^{(1)}_{\ell} \arrow[r] & 0~.
\end{tikzcd}  
\ee  
This is indeed the Poincar\'e dual of \eqref{Exact_Sequence_D1} in the representation category.

\begin{figure}[ht]
    \centering
    \includegraphics[width=.8\linewidth]{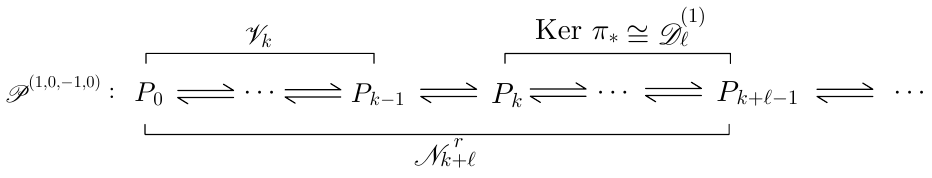}
    \caption{The finite-dimensional module $\scN_k^{r, w}$ is shown as an extension of $\scV_k$ by $\scD^{(1)}_\ell$ through the short exact sequence \eqref{Exact_Sequence_D1}.The weights $w = (1, 0, -1, 0)$ are omitted from the superscripts of the polynomials for simplicity. Note that $\scD^{(1)}_\ell$ is not obtained as a quotient of $\scP^w$ but rather from $\scP^{s_r(w)}$, where $s_r(w) = (-1, 0, -1, 0)$.}
    \label{fig:ext_vk_dl1}
\end{figure}

\bigskip

In this way, one can construct bound states of compact $A$-branes and their corresponding $\SH$-modules. Let us consider two more illustrative examples. 
For the first example, consider an $A$-brane $\brane_{\bfN_r}$ supported on the union of cycles $\bfN_r = \bfD_1 \cup \bfD_2 \cup \bfD_3 \cup \bfV$, which corresponds to the root
\be
r = (-1, -1, -1, 1).
\ee
Imposing the shortening conditions for $\scD^{(j)}_{\ell_j}$ and $\scV_k$ simultaneously (see Table \ref{Tab:Matching_A-brane_Rep_I0}), we obtain:
\be 
q^m = t_1 t_2 t_3 t_4^{-1},
\ee
where $m = k + \ell_1 + \ell_2 + \ell_3$. This result is consistent with the $A$-brane condition
\be\label{I0_Example}
m = \dim\Hom(\brane_{\bfN_r}, \Bcc) = \frac{1}{\hbar} - \frac{\talpha_1 + \talpha_2 + \talpha_3 - \talpha_4}{\hbar}.
\ee
To identify the corresponding $\SH$-module explicitly, we choose $w = (-1, 0, -1, 0)$ and consider the polynomial representation $\scP^w$ under the simultaneous shortening conditions. Following a similar argument as before, we find that there exists a short exact sequence
\be
\begin{tikzcd}[cramped, sep=small]
    0 \arrow[r] & \bigoplus_{j=1,2,3}\scD^{(j)}_{\ell_j} \arrow[r] & \scN^{r,w}_m \arrow[r] & \scV_k \arrow[r] & 0~.
\end{tikzcd}
\ee
The module $\scN^{r,w}_m$ can be understood as a diagonal element in the representation category, $q_1 + q_2 + q_3 \in \Hom^*(\brane_{\bfV}, \bigoplus_{j=1,2,3}\brane_{\bfD_j})$. Thus, we can express the corresponding $A$-brane as:
\be
\brane_{\bfN_r} \in \Hom^*\left(\brane_{\bfV}, \bigoplus_{j=1,2,3} \brane_{\bfD_j}\right).
\ee
In the limit of the DAHA of type $A_1$, this module corresponds to the finite-dimensional representation newly discovered in \cite[\S2.7,2]{Gukov:2022gei}. While this representation cannot be constructed as a quotient of the polynomial representation of the DAHA of type $A_1$, the correspondence with the $A$-brane predicts its existence. For the DAHA of type $C^\vee C_1$, however, this module can be obtained as a quotient of the polynomial representation by leveraging the $W(D_4)$ Weyl group and the $\bZ_3$ cyclic symmetry of the algebra.

As the second example, consider the $A$-brane supported on the entire global nilpotent cone $\bfN_{r} := h^{-1}(0) = \bfV \cup \bigcup_{j=1}^4 \bfD_j$ (see Figure \ref{fig:nilpotent_cone}), which corresponds to the root
\be
r = (-2, 0, 0, 0).
\ee
The dimension formula for the corresponding $A$-brane $\brane_{\bfN_{r}}$ can be evaluated as
\be
m = \dim\Hom(\brane_{\bfN_r}, \Bcc) = \int_{\bfN_r} \frac{F + B}{2\pi} = \frac{1}{\hbar} - \frac{2\talpha_1}{\hbar}.
\ee
This yields the shortening condition $q^m = t_1^2$. If we further assume the existence of $\brane_\bfV$, which implies $q^k = t_1^{-2}$, then $q$ must be a $(m + k)$-root of unity. As discussed in \S\ref{sec:genericfiber}, this condition aligns with the $A$-brane condition for $\brane_{\bfF}$ supported on a generic Hitchin fiber.
The representation associated with the nilpotent cone $\bfN_r$ can then be constructed as the quotient module
\be
\scN_{m}^r := \scF_{m+k} / \scV_k.
\ee
By construction, this quotient module fits into the following short exact sequence
\be\label{Exact_Seq_Nilpotent_cone}
\begin{tikzcd}[cramped, sep=small]
    0 \arrow[r] & \scV_k \arrow[r] & \scF_{m+k} \arrow[r] & \scN_{m}^r \arrow[r] & 0~.
\end{tikzcd}
\ee

If we impose additional shortening conditions appropriately, a similar analysis yields the short exact sequence:
\be\label{NilpotentConeMorphism}
\begin{tikzcd}[cramped, sep=small]
    0 \arrow[r] & \bigoplus_{j=1}^4 \scD^{(j)}_{\ell_j} \arrow[r] & \scN_a \arrow[r] & \scV_k \arrow[r] & 0~,
\end{tikzcd}
\ee
which represents an element in $\Ext^1(\scV_k, \bigoplus_{j=1}^4 \scD^{(j)}_{\ell_j})$.

As a consistency check, consider the limit of the DAHA of type $A_1$ specified by \eqref{A1Parameters}. Under this specification, we recover all the $A$-brane conditions discussed in \cite{Gukov:2022gei}.

\paragraph{Affine braid group action.}

Under our current assumptions $\tgamma=0$ and vanishing $B$-field flux along the global nilpotent cone, we may set $v=1$.  
As explained in \S\ref{sec:category}, under these conditions the affine braid group $\dt\Br(D_4)$ of type $D_4$ acts naturally on the category of $A$-branes.  
We now describe explicitly how this affine braid group acts on compactly supported Lagrangian $A$-branes.

Fix a base point in parameter space such that the parameter $\talpha$ satisfies the existence conditions for all irreducible $A$-branes listed in Table~\ref{Tab:Matching_A-brane_Rep_I0}.  
Starting from this base point, consider a loop in the $(\talpha,\tbeta)$-parameter space encircling a locus where a two-cycle $\bfW_a$ shrinks to zero volume.  
The associated monodromy defines a generator $T_a$ of the affine braid group $\dt\Br(D_4)$.

At the level of homology, the supports of Lagrangian branes are transformed according to the Picard--Lefschetz formula \eqref{mPL}.  
At the categorical level, however, the same monodromy is lifted to an autoequivalence of the derived category of $A$-branes, encoded by the exact triangle
\cite{Seidel1999,seidel2000graded,Seidel2008}
\be \label{exact-triangle}
\Hom(\brane_{\bfW_a}, \brane_{\bfW_b}) \otimes \brane_{\bfW_a}
\xrightarrow{\mathrm{ev}}
\brane_{\bfW_b}
\longrightarrow
T_a(\brane_{\bfW_b})
\longrightarrow
\bigl(\Hom(\brane_{\bfW_a}, \brane_{\bfW_b}) \otimes \brane_{\bfW_a}\bigr)[1]~,
\ee
where $\mathrm{ev}$ denotes the evaluation map.

Consider two distinct Lagrangian $A$-branes $\brane_{\bfW_a}$ and $\brane_{\bfW_b}$ whose supports intersect transversely at a single point $p$.  
In this case,
\be
\Hom^1(\brane_{\bfW_a}, \brane_{\bfW_b})
\cong \bC\langle p\rangle~,
\ee
and the exact triangle \eqref{exact-triangle} reduces to
\be
\brane_{\bfW_a}[-1]
\xrightarrow{\mathrm{ev}}
\brane_{\bfW_b}
\longrightarrow
T_a(\brane_{\bfW_b})
\longrightarrow
\brane_{\bfW_a}~.
\ee
Here the grading is determined by the Maslov index, as reviewed in \S\ref{sec:brane-quantization}.  
It follows that $T_a(\brane_{\bfW_b})$ is the bound state of $\brane_{\bfW_a}$ and $\brane_{\bfW_b}$, supported on the cycle
$[\bfW_{ab}] = [\bfW_a] + [\bfW_b]$,
namely
\begin{equation}\label{mPL-2}
T_a(\brane_{\bfW_b}) = \brane_{\bfW_{ab}}~.
\end{equation}

Next, consider the action of $T_a$ on $\brane_{\bfW_a}$ itself.  
Since each irreducible component of the global nilpotent cone is topologically an $S^2$, its self-intersection number is $-2$.  
At the level of homology, the Picard--Lefschetz transformation therefore acts as
\be \label{orientation-reverse}
T_a : [\bfW_a] \longmapsto -[\bfW_a]~,
\ee 
corresponding to a reversal of orientation.

In the derived category of $A$-branes, $\brane_{\bfW_a}$ is a spherical object, satisfying
\be
\Hom(\brane_{\bfW_a}, \brane_{\bfW_a})
\cong H^*(S^2)
\cong \bC \oplus \bC[-2]~.
\ee
The exact triangle \eqref{exact-triangle} then becomes
\be
\brane_{\bfW_a} \oplus \brane_{\bfW_a}[-2]
\xrightarrow{\mathrm{ev}}
\brane_{\bfW_a}
\longrightarrow
T_a(\brane_{\bfW_a})
\longrightarrow
\brane_{\bfW_a}[1] \oplus \brane_{\bfW_a}[-1]~,
\ee
which implies that the symplectic Dehn twist shifts the grading by $-1$:
\be
T_a(\brane_{\bfW_a}) = \brane_{\bfW_a}[-1]~.
\ee
Thus, although the homology class changes sign \eqref{orientation-reverse}, the underlying Lagrangian support remains unchanged; the nontrivial effect in the $A$-brane category is instead a shift of the Maslov grading.
In particular, $T_a^2$ is \emph{not} the identity but shifts the grading by $-2$.  
As a result, the generators $T_a$ satisfy affine braid relations \eqref{braid-relation} rather than affine Weyl group relations.

Correspondingly, objects in the derived category of $\SH$-modules are naturally $\bZ$-graded.  
The affine braid group action on $D^b(\Rep(\SH))$ is defined in an entirely parallel manner, via exact triangles of the same form.

\subsubsection{At a singular fiber of type \texorpdfstring{$I_4$}{I4}}\label{sec:brane-I4}

In the previous subsection, we provided a detailed analysis of the correspondence between compact $A$-branes and finite-dimensional $\SH$-modules in the case where $\tbeta_j = \tgamma_j = 0$ and $\hbar$ is real. We now turn to the analysis of other fibration configurations listed in Table \ref{tab:fiber_sing_classification2}. To proceed, we keep $\hbar$ to be real and adjust the ramification parameter $\tgamma$ to realize the $I_4$ singular fiber, as illustrated in \S\ref{I_4 singular fiber}. In this case, a generic fiber remains Lagrangian. Therefore, all the irreducible components of the $I_4$ singular fiber are also Lagrangian.
To be concrete, we focus on the specific case where $(\tgamma_1, \tgamma_2, \tgamma_3, \tgamma_4) = (\tgamma_1, 0, \tgamma_1, 0)$. 

\paragraph{Object Matching} 
We apply the same technique as in the previous cases to match the $A$-branes and finite-dimensional $\SH$-modules, by comparing the shortening conditions and $A$-brane conditions. 
The compact $A$-branes $\brane_{\bfF}$ and $\brane_{\bfU_i}$ $(i=1,2,3,4)$ can exist when $q$ is a root of unity and specific $\boldsymbol{t}$ are chosen. In this case, the homology classes and the volumes of the irreducible components have been analyzed in \eqref{I_4-volume}. The $A$-brane condition is thus specified by
\begin{equation}
d_i:=\dim\Hom(\brane_{\bfU_i},\Bcc)=\int_{\bfU_i}\frac{F+B}{2\pi}=\frac{
    \text{vol}_I(\bfU_i)}{\hbar} ,\qquad (i=1,2,3,4)
\end{equation}
from which one could read off the corresponding shortening conditions. The finite-dimensional representation $\scU^{(i)}_{d_i}$ can be constructed by a quotient of an appropriate polynomial representation $\scP^w$. The correspondence between irreducible components and finite-dimensional representations is summarized in Table \ref{Matching_A-brane_Rep_I4}.

\begin{table}[ht]
    \centering
    \renewcommand{\arraystretch}{1.3}
    \begin{tabular}{c|c|c|c}
\hline \text { finite-dim rep } & \text { shortening condition }  &\text { $A$-brane}& \text {$A$-brane condition } \\
\hline $\scU^{(1)}_{d_1}$ &$q^{d_1}=t_1^{-1}t_2t_3t_4$  &$\brane_{\bfU_1}$& $d_1=\frac{1}{\hbar}-\frac{-\talpha_1+\talpha_2+\talpha_3+\talpha_4}{\hbar}$\\
\hline $\scU^{(2)}_{d_2}$ & $q^{d_2}=t_1t_2t_3^{-1}t_4^{-1}$  &$\brane_{\bfU_2}$&$d_2=\frac{-\talpha_1-\talpha_2+\talpha_3+\talpha_4}{\hbar}$ \\
\hline $\scU^{(3)}_{d_3}$ & $q^{d_3}=t_1^{-1}t_2^{-1}t_3t_4^{-1}$  &$\brane_{\bfU_3}$&$d_3= \frac{\talpha_1+\talpha_2-\talpha_3+\talpha_4}{\hbar}$\\
\hline $\scU^{(4)}_{d_4}$  & $q^{d_4}=t_1t_2^{-1}t_3^{-1}t_4$ &$\brane_{\bfU_4}$&$d_4=\frac{-\talpha_1+\talpha_2+\talpha_3-\talpha_4}{\hbar}$ \\
\hline
    \end{tabular}
    \caption{A summary of finite-dimensional $\protect\SH$-modules with their shortening conditions and the corresponding $A$-brane configurations at the $I_4$ singular fiber, under the assumption $|q|=1$.}
    \label{Matching_A-brane_Rep_I4}
\end{table}

\paragraph{Morphism Matching}
 From Figure \ref{I_4 singular fiber}, the fiber of type $I_4$ consists of four $\mathbb{CP}^1$ joining like a necklace, or the affine $\dot{A}_3$ Dynkin diagram. Therefore, the morphism of the $A$-branes $\brane_{\bfU_i}$ and $\brane_{\bfU_{i+1}}$ is
\be
\Hom^*\left(\brane_{\bfU_i}, \brane_{\bfU_{i+1}}\right):=C F^*\left(\brane_{\bfU_{i}}, \brane_{\bfU_{i+1}}\right) \cong \bC\left\langle p_i\right\rangle  ,\qquad (i=1,2,3,4).
\ee
where we denote $\bfU_{5}:=\bfU_{1}$ and $p_i$ as the intersection point between $\bfU_{i}$ and $\bfU_{i+1}$. 

From the perspective of representation theory, the morphism space can be directly constructed from the analysis above, as the cycles $\bfU_{i}$ and $\bfU_{i+1}$ intersect at a single point. Specifically, let $\scU^{(i,i+1)}_{d_i + d_{i+1}}$ denote the finite-dimensional representation corresponding to $A$-brane supported on the cycle $\bfU_i \cup \bfU_{i+1}$. Using the method described in the $I_0^*$ singular fiber, the following short exact sequence can be constructed in the representation theory side:
\be\label{Exact_Seq_I4}
\begin{tikzcd}[cramped, sep=small]
    0 \arrow[r]&\scU^{(i)}_{d_i}\arrow[r]&\scU^{(i,i+1)}_{d_i+d_{i+1}}\arrow[r]&\scU^{(i+1)}_{d_{i+1}}\arrow[r]&0 \quad,\qquad (i=1,2,3,4)
\end{tikzcd}
\ee
which yields an element in $\Ext^1(\scU^{(i+1)}_{d_{i+1}},\scU^{(i)}_{d_i})$.
Thus, we establish the matching of the morphism space in the case of the $I_4$ singular fiber.

\subsubsection{At a singular fiber of type \texorpdfstring{$I_3$}{I3}}\label{sec:brane-I3}

We adjust the ramification parameters to realize the $I_3$ singular fiber. Since a generic fiber is Lagrangian, all irreducible components of the $I_3$ singular fiber are also Lagrangian. For clarity, we focus on the specific case where $(\tgamma_1, \tgamma_2, \tgamma_3, \tgamma_4) = (\tgamma_1, 0, \tgamma_3, \tgamma_1+\tgamma_3)$. 

\paragraph{Object Matching}
We apply the same technique as in the previous cases to match the $A$-branes and finite-dimensional representations, by comparing the shortening conditions and $A$-brane conditions. With the volumes of the irreducible components $\bfU_i$ evaluated in \eqref{volume_I3}, the $A$-brane condition writes: 
\be
d_i := \dim\Hom(\brane_{\bfU_i}, \Bcc) = \int_{\bfU_i} \frac{F + B}{2\pi} = \frac{\text{vol}_I(\bfU_i)}{\hbar}~i=1,2,3
\ee
We can extract the shortening condition. The finite-dimensional representation $\scU^{(i)}_{d_i}$ is constructed as a quotient of $\scP^w$ by $(P_{d_i}^w)$, where $w$ is an appropriate element in $W(D_4)$. The correspondence between the irreducible components and the finite-dimensional representations is summarized in Table \ref{Matching_A-brane_Rep_I3}.

\begin{table}[ht]
    \centering
    \renewcommand{\arraystretch}{1.3}
    \begin{tabular}{c|c|c|c}
\hline \text { finite-dim rep } & \text { shortening condition }  &\text {$A$-brane}& \text {$A$-brane condition } \\
\hline $\scU^{(1)}_{d_1}$ & $q^{d_1}=t_2^2$  &$\brane_{\bfU_1}$& $d_1=\frac{1}{\hbar}-\frac{2\talpha_2}{\hbar}$ \\
\hline $\scU^{(2)}_{d_2}$ & $q^{d_2}=t_1^{-1}t_2^{-1}t_3^{-1}t_4$  &$\brane_{\bfU_2}$& $d_2=\frac{\talpha_1+\talpha_2+\talpha_3-\talpha_4}{\hbar}$ \\
\hline $\scU^{(3)}_{d_3}$ & $q^{d_3}=t_1t_2^{-1}t_3t_4^{-1}$  &$\brane_{\bfU_3}$& $d_3= \frac{-\talpha_1+\talpha_2-\talpha_3+\talpha_4}{\hbar}$ \\
\hline
    \end{tabular}
    \caption{A summary of finite-dimensional $\protect\SH$-modules with their shortening conditions and the corresponding $A$-brane configurations at the $I_3$ singular fiber, under the assumption $|q|=1$.}
    \label{Matching_A-brane_Rep_I3}
\end{table}

\paragraph{Morphism Matching}
The analysis is the same as in an $I_4$ fiber:
\be
\Hom^*\left(\brane_{\bfU_i}, \brane_{\bfU_{i+1}}\right):=C F^*\left(\brane_{\bfU_{i}}, \brane_{\bfU_{i+1}}\right) \cong \bC\left\langle p_i\right\rangle  ,i=1,2,3
\ee
where we denote $\bfU_{4}:=\bfU_{1}$ and $p_i$ as the intersection point between $\bfU_{i}$ and $\bfU_{i+1}$.

From the viewpoint of representation theory, the morphism space can be constructed directly based on the analysis above, given that the cycles $\bfU_i$ and $\bfU_{i+1}$ intersect at a single point. More precisely, the corresponding $\SH$-module can be expressed in terms of the short exact sequence
\be\label{Exact_Seq_I3}
\begin{tikzcd}[cramped, sep=small]
    0 \arrow[r]&\scU^{(i)}_{d_i}\arrow[r]&\scU^{(i,i+1)}_{d_i+d_{i+1}}\arrow[r]&\scU^{(i+1)}_{d_{i+1}}\arrow[r]&0~,
\end{tikzcd}
\ee
which yields an element in $\Ext^1(\scU^{(i+1)}_{d_{i+1}},\scU^{(i)}_{d_i})$. Hence, we conclude the matching of the morphism space in the case of $I_3$ singular fiber.

\subsubsection{At a singular fiber of type \texorpdfstring{$I_2$}{I2}}\label{sec:brane-I2}

We adjust the ramification parameters to realize the $I_2$ singular fiber. Since a generic fiber is Lagrangian, all irreducible components of the $I_2$ singular fiber are also Lagrangian. For clarity, we focus on the specific case where $(\tgamma_1, \tgamma_2, \tgamma_3, \tgamma_4) = (\tgamma_1, 0, \tgamma_3, \tgamma_4)$. 

\paragraph{Object Matching}
The volumes of the irreducible components have been previously analyzed in \eqref{volume_I2_2}, with associated $A$-brane condition:
\be
d_i := \dim\Hom(\brane_{\bfU_i}, \Bcc) = \int_{\bfU_i} \frac{F + B}{2\pi} = \frac{\text{vol}_I(\bfU_i)}{\hbar}, i=1,2
\ee
From this, we can extract the shortening conditions. The finite-dimensional representation $\scU^{(i)}_{d_i}$ is constructed by quotienting $\scP^w$ over $P_{d_i}^w$. The correspondence between the irreducible components and the finite-dimensional representations is summarized in Table \ref{Matching_A-brane_Rep_I3}.

\begin{table}[ht]
    \centering
    \renewcommand{\arraystretch}{1.3}
    \begin{tabular}{c|c|c|c}
\hline \text { finite-dim rep } & \text { shortening condition }  &\text{$A$-brane}& \text {$A$-brane condition } \\
\hline $\scU^{(1)}_{d_1}$ & $q^{d_1}=t_2^2$  &$\brane_{\bfU_1}$& $d_1=\frac{1}{\hbar}-\frac{2\talpha_2}{\hbar}$\\
\hline $\scU^{(2)}_{d_2}$ & $q^{d_2}=t_2^{-2}$  &$\brane_{\bfU_2}$& $d_2=\frac{2\talpha_2}{\hbar}$ \\
\hline
    \end{tabular}
    \caption{A summary of finite-dimensional $\protect\SH$-modules with their shortening conditions and the corresponding $A$-brane configurations at the $I_2$ singular fiber, under the assumption $|q|=1$.}
    \label{Matching_A-brane_Rep_I2}
\end{table}

\paragraph{Morphism Matching}

The $I_2$ case differs from the analysis above, as two irreducible components $\bfU_{1}$ and $\bfU_{2}$ intersect at two points, denoted as $p_1,p_2$.
In this case, the morphism of the $A$-branes $\brane_{\bfU_1}$ and $\brane_{\bfU_{2}}$ is
\be
\Hom^*\left(\brane_{\bfU_1}, \brane_{\bfU_{2}}\right):=C F^*\left(\brane_{\bfU_{1}}, \brane_{\bfU_{2}}\right) \cong \bC\left\langle p_1\right\rangle\oplus\bC\left\langle p_2\right\rangle .
\ee

From the viewpoint of representation theory, the morphism space is expected to be two-dimensional. In this case, applying a similar technique as before, a short exact sequence is obtained from the polynomial representation $\scP^w$ with $w=(0,1,0,1)$
\be\label{Exact_Seq_I2}
\begin{tikzcd}[cramped, sep=small]
    0 \arrow[r]&\scU^{(1)}_{d_1}\arrow[r]&\scF^{(0,1,0,1)}_{d_1+d_{2}}\arrow[r]&\scU^{(2)}_{d_{2}}\arrow[r]&0~,
\end{tikzcd}
\ee
which yields an element in $\Ext^1(\scU^{(2)}_{d_{2}},\scU^{(1)}_{d_1})$. 

However, as the morphism space is two-dimensional, from the representation theory side, there must be another morphism that one can find. To identify the corresponding module, we consider the quotients of the polynomial representation with a different weight $w=(0,1,0,-1)$
\be\label{Exact_Seq_I2(2)}
\begin{tikzcd}[cramped, sep=small]
    0 \arrow[r]&\scU^{(1)}_{d_1}\arrow[r]&\scF^{(0,1,0,-1)}_{d_1+d_{2}}\arrow[r]&\scU^{(2)}_{d_{2}}\arrow[r]&0~,
\end{tikzcd}
\ee
which yields another element in $\Ext^1(\scU^{(2)}_{d_{2}},\scU^{(1)}_{d_1})$. As the raising/lowering operators in $\scF^{(0,1,0,1)}_{d_1+d_{2}}$ and $\scF^{(0,1,0,-1)}_{d_1+d_{2}}$ are different, they provide two distinct generators for $\Ext^1(\scU^{(2)}_{d_{2}},\scU^{(1)}_{d_1})$.
Therefore, we conclude the matching of the morphism space in the case of $I_2$ singular fiber.

\subsubsection{Generic \texorpdfstring{$\hbar$}{hbar} parameter}\label{sec:generic_hbar}

Until now, we have assumed that $\hbar$ is real. However, the analysis of finite-dimensional representations in \eqref{generic_shortening_cond} does not rely on this assumption. Let us now explore the scenario where $\hbar = |\hbar| e^{i\theta}$ is complex rather than real. In this case, the symplectic form $\omega_\X = \Im \Omega$ of the $A$-model is no longer proportional to $\omega_K$ but instead becomes a linear combination of $\omega_I$ and $\omega_K$. As a result, for generic values of $\hbar$, a Hitchin fiber is no longer Lagrangian with respect to $\omega_{\MS}$. 

As shown in \eqref{Homology_Lattice}, the second integral homology group of the target space $\X$ is isomorphic to the affine $D_4$ root lattice. The standard generators $[\bfD_j]$ ($j=1,2,3,4$) and $[\bfV]$ correspond to the simple roots $e^\ell$ ($\ell=0,\ldots,4$) of the affine $D_4$ root system, as shown in \eqref{homology-root}. For simplicity, we denote the homology classes by their corresponding roots in the following paragraphs. Using \eqref{vol-middle-chamber4}, their volumes with respect to $\Omega$ can be written as
\be 
\int_{e^\ell}\frac{\Omega}{2\pi} = \frac{e^\ell \cdot (\talpha - i\tgamma)}{\hbar}~,\qquad (\ell = 0,\ldots,4).
\ee 
Therefore, a necessary condition for the homology class $e^\ell$ to be represented by a Lagrangian submanifold with respect to $\omega_\X$ is
\be 
\Im \frac{e^\ell \cdot (\talpha - i\tgamma)}{\hbar} = 0~.
\ee 
The comparison with the representations of $\SH$ indeed indicates that it is also sufficient although a rigorous derivation is unknown to us. In other words, for any homology class $e^\ell$ satisfying this condition, there exists a corresponding Lagrangian submanifold with respect to $\omega_\X$. Based on this assumption, the correspondence between compact $A$-branes and $\SH$-modules for generic $\hbar$ is summarized in Table \ref{Tab:Matching_A-brane_Rep_I0_generic}. 

Furthermore, in direct analogy with the discussion following \eqref{generic_cycle}, for any root $r\in \sfR(D_4)$ there exists a compact Lagrangian brane $\brane_{\bfN_r}$ when the parameters satisfy $q^m=\boldsymbol{t}^{-r}$. The corresponding $\SH$-module is given by a quotient of the polynomial representation,
\be \label{scN-2}
\scN^{r,w}_m = \scP^w /(P^w_m),
\ee
where the weight $w$ obeys $\langle r, w\rangle>0$.

When $q$ is not a root of unity, all finite-dimensional representations of $\HH$ are rigid. In this case, the classification of finite-dimensional representations has been completed in \cite{oblomkov2009finite}, where it was shown that every finite-dimensional representation arises as a quotient of a polynomial representation. The above correspondence, therefore, provides a nontrivial consistency check: the classification of compact Lagrangian $A$-branes obtained from our geometric considerations aligns precisely with the known classification \cite{oblomkov2009finite} of finite-dimensional $\HH$-modules.

\begin{table}[ht]
    \centering
    \renewcommand{\arraystretch}{1.3}
    \begin{tabular}{c|c|c|c}
\hline \text { finite-dim rep } & \text { shortening condition } & \text{$A$-brane} & \text {$A$-brane condition } \\
\hline $\scD^{(1)}_{\ell_1}$ & $q^{\ell_1}=t_1t_2t_3t_4$ & $\brane_{\bfD_1}$ & $\ell_1=\frac{1}{\hbar}-\frac{\theta\cdot(\talpha-i\tgamma)}{\hbar}$ \\
\hline $\scD^{(2)}_{\ell_2}$ & $q^{\ell_2}=t_1t_2t_3^{-1}t_4^{-1}$ & 
$\brane_{\bfD_2}$ & $\ell_2=\frac{e^1\cdot(\talpha-i\tgamma)}{\hbar}$ \\
\hline $\scD^{(3)}_{\ell_3}$ & $q^{\ell_3}=t_1t_2^{-1}t_3t_4^{-1}$ & 
$\brane_{\bfD_3}$ & $\ell_3=\frac{e^2\cdot(\talpha-i\tgamma)}{\hbar}$ \\
\hline $\scD^{(4)}_{\ell_4}$ & $q^{\ell_4}=t_1t_2^{-1}t_3^{-1}t_4$ & 
$\brane_{\bfD_4}$ & $\ell_4=\frac{e^3\cdot(\talpha-i\tgamma)}{\hbar}$ \\
\hline $\scV_{k}$ & $q^{k}=t_1^{-2}$& 
$\brane_{\bfV}$ & $k=\frac{e^4\cdot(\talpha-i\tgamma)}{ \hbar}$ \\
\hline
    \end{tabular}
    \caption{A summary of finite-dimensional representations of $\protect\SH$ with corresponding shortening and $A$-brane conditions at $I_0^*$ singular fiber.}
    \label{Tab:Matching_A-brane_Rep_I0_generic}
\end{table}

As a special case, consider the scenario where $\hbar$ is purely imaginary ($i\hbar\in\bR_{>0}$) so that the symplectic form becomes $\omega_\X = \omega_I / |\hbar|$. As described in \eqref{suspended-homology}, the suspended cycles $\bfW_{ij}$ serve as generators of the second homology group, corresponding to the simple roots $e^\ell$ ($\ell=1,2,3,4$), when the ramification parameters $\tgamma_j$ are generic, and the Hitchin fibration has  $6I_1$.

A particularly interesting limit occurs when $\talpha_j = 0$ and $\tgamma_j$ lie within the chamber specified by \eqref{gamma-chamber}. In this case, the volumes of the suspended cycles with respect to $\Omega$ are given by \eqref{Volume_suspended_cycles}:
\be 
\int_{e^\ell}\frac{\Omega}{2\pi} =-\frac{ie^\ell \cdot \tgamma}{\hbar} \in \bR, \qquad (\ell = 1, 2, 3, 4).
\ee 
Then, applying the above assumption, there exist four compact branes of type $(A, A, B)$ suspended between the $I_1$ singular fibers. The correspondence between these $A$-branes and $\SH$-modules is summarized in Table \ref{Tab:Matching_A-brane_Rep_General}.

\begin{table}[ht]
    \centering
    \renewcommand{\arraystretch}{1.3}
    \begin{tabular}{c|c|c|c}
\hline \text { finite-dim rep } & \text { shortening condition } & $A$-brane & \text {$A$-brane condition } \\
\hline $\scW^{(13)}_{d_1}$ & $q^{d_1}=t_1^{-2}$& $\brane_{\bfW_{13}}$ & $d_1=\frac{2 \tgamma_1}{ |\hbar|}$ \\
\hline $\scW^{(12)}_{d_2}$ & $q^{d_2}=t_1t_2t_3^{-1}t_4^{-1}$ & 
$\brane_{\bfW_{12}}$ & $d_2=\frac{-\tgamma_1-\tgamma_2+\tgamma_3+\tgamma_4}{|\hbar|}$ \\
\hline $\scW^{(34)}_{d_3}$ & $q^{d_3}=t_1t_2^{-1}t_3t_4^{-1}$ & 
$\brane_{\bfW_{34}}$ & $d_3=\frac{-\tgamma_1+\tgamma_2-\tgamma_3+\tgamma_4}{|\hbar|}$ \\
\hline $\scW^{(56)}_{d_4}$ & $q^{d_4}=t_1t_2^{-1}t_3^{-1}t_4$ & 
$\brane_{\bfW_{56}}$ & $d_4=\frac{-\tgamma_1+\tgamma_2+\tgamma_3-\tgamma_4}{|\hbar|}$ \\
\hline
    \end{tabular}
    \caption{A summary of finite-dimensional representation of $\protect\SH$ with corresponding shortening and $A$-brane conditions under $\talpha_j=0$ limit.}
    \label{Tab:Matching_A-brane_Rep_General}
\end{table}

\acknowledgments
SN  would like to thank Sergei Gukov, Peter Koroteev, Du Pei and Ingmar Saberi for the collaboration in \cite{Gukov:2022gei}, based on which this paper is written. The authors would like to thank Yutaka Yoshida for sharing his draft \cite{Yoshida:2025iae} on the related topic with us.  In addition, SN is grateful to Chris Brav, Yixuan Li, Umut Varolgunes and Meri Zaimi for discussions.
This work is supported by National Natural Science Foundation of China No.12050410234, and Shanghai Municipal Science and Technology Major Project No. 22WZ2502100 and No. 24ZR1403900. The work of Z.Y. is supported by National Natural Science Foundation of China No.123B1010.

\appendix

\section{Notations}\label{app:notation}

The conventions and notations largely follow those in \cite{Gukov:2022gei}.  
\begin{itemize}\setlength\itemsep{.05pt}
    \item {Sans-serif symbols:} Single sans-serif symbols are used to denote lattices or free $\bZ$-modules (e.g., $\sfQ$ and $\sfP$ for the root and weight lattices, respectively). Words in sans-serif type (e.g., \ABrane) refer to categories.  

    \item {Calligraphic letters:} Symbols such as $\mathcal{M}$ or $\mathcal{B}$ are reserved for objects that are moduli spaces or closely related to them.  

    \item {Boldface symbols:} Boldface symbols are used for two-cycles of the target space, often representing the support of $A$-branes (e.g., $\bfF$ for a generic fiber of the Hitchin fibration).  

    \item {Gothic letters:} Capital gothic symbols (e.g., $\X$ for the target space) denote objects equipped with the structure required by the topological $A$-model. For example, $\brane$ represents an $A$-brane associated with specific data, while $\brane_{\bfF}$ denotes a brane supported on a generic fiber $\bfF$ of the Hitchin fibration.

    \item {Script letters:} Script letters are used for modules over the algebra $\OO^q(\X)$, with the specific algebra being clear from the context. For consistency, the same symbol is used for a brane and its corresponding representation. For instance, a representation $\scF$ of the spherical DAHA $\SH$ is identified with an $A$-brane $\brane_{\bfF}$ under the equivalence \eqref{eq:functor} between the two categories.
\end{itemize}

Let $\bC[q^{ \pm \frac{1}{2}}, \boldsymbol{t}^\pm]:= \bC[q^{ \pm \frac{1}{2}}, t_1^{ \pm},t_2^{ \pm},t_3^{ \pm},t_4^{ \pm}]$ be the ring of Laurent polynomials in the formal parameters $q^{1 / 2}$ and $\boldsymbol{t}=(t_1,t_2,t_3,t_4)$, and consider a multiplicative system $M$ in $\bC[q^{ \pm \frac{1}{2}}, \boldsymbol{t}^\pm]$ generated by elements of the form $(q^{\ell/2} t_1t_3-q^{-\ell/2} t_1^{-1}t_3^{-1})$ for any non-negative  integer $\ell \in \bZ_{\geq 0}$. We define the coefficient ring $\bC_{q, t}$ to be the localization (or formal ``fraction'')  of the ring $\bC[q^{ \pm \frac{1}{2}}, \boldsymbol{t}^{ \pm}]$ at $M$:
\be \label{localization}
\bC_{q,\boldsymbol{t}}=M^{-1} \bC[q^{ \pm \frac{1}{2}}, \boldsymbol{t}^{ \pm}] .
\ee

The standard notation of DAHA of $C^\vee C_1$ used in \cite{oblomkov2004cubic} is given by
\begin{equation}
\begin{aligned}
\left(T_0-t_1\right)\left(T_0+t_1^{-1}\right) &=0, \cr
\left(T_0^{\vee}-t_2\right)\left(T_0^{\vee}+t_2^{-1}\right) &=0, \cr
\left(T_1-t_3\right)\left(T_1+t_3^{-1}\right) &=0, \cr
\left(T_1^{\vee}-t_4\right)\left(T_1^{\vee}+t_4^{-1}\right) &=0,
\end{aligned}
\end{equation}
and
\begin{equation}
T_1^\vee T_1 T_0 T_0^\vee=q^{-\frac{1}{2}}.
\end{equation}

To match with the geometry side, we make the following change of variables to obtain our definition on DAHA \eqref{sec:DAHA-CC1-subsec}:
\begin{equation}
\begin{aligned}
    (T_0,T_0^\vee,T_1,T_1^\vee) &\to(T_1,T_1^\vee,T_0,T_0^\vee) \cr
    (t_1,t_2,t_3,t_4) &\to(-it_3,-it_4,-q^{\frac{1}{2}}it_1,-it_2) ~.
\end{aligned}
\end{equation}
\subsubsection*{Root system and weight lattice}

In this paper, we establish notation for the root system and weight lattice of type $D_4$. For the discussion of the root system, we assume an orthogonal basis in $\bR^4$ equipped with the standard Euclidean inner product, allowing us to express the roots accordingly.
The standard convention for the $D_4$ root system and weight systems is as follows, with a tilde ($\sim$) placed above the notation to distinguish it from our chosen notation. The $D_4$ root system consists of 24 non-zero roots and 4 zero roots, among whom the non-zero roots are expressed as:

\be
\tilde{\sfR}(D_4) = \{\pm \epsilon^i \pm \epsilon^j \mid i,j \in \{1,2,3,4\}, i > j \},
\ee
where the basis vectors $\epsilon^i$ are defined as:
\be
\epsilon^1 = (1,0,0,0), \quad \epsilon^2 = (0,1,0,0), \quad \epsilon^3 = (0,0,1,0), \quad \epsilon^4 = (0,0,0,1).
\ee

To match the geometry of the cubic surface, we adopt a slightly unconventional presentation of the $D_4$ root system:
\begin{equation}\label{D4rootsystem}
    \sfR(D_4) = \left\{(\pm 1, \pm 1, \pm 1, \pm 1), (\pm 2, 0, 0, 0), (0, \pm 2, 0, 0), (0, 0, \pm 2, 0), (0, 0, 0, \pm 2) \right\}.
\end{equation}
Two conventions differ by change of bases, specified by
\begin{equation}
    r= A \tilde{r}~,\qquad   A=\begin{pmatrix}
1 & 1 & 0 & 0 \\
1 & -1 & 0 & 0 \\
0 & 0 & 1 & 1 \\
0 & 0 & 1 & -1 
\end{pmatrix}  ,\quad \tilde{r}\in\tilde{\sfR}(D_4)
\end{equation}
As a result, the norm of vectors in $\sfR(D_4)$ in our notation is normalized to 2 instead of $\sqrt{2}$.
We take the set of simple roots for $D_4$ root system as
\begin{equation}
    \{e^1, e^2, e^3, e^4\} = \{(-1, -1, 1, 1), (-1, 1, -1, 1), (-1, 1, 1, -1), (2, 0, 0, 0)\}~,
\end{equation}
Then, the highest root is expressed by $\theta=e^1+e^2+e^3+2e^4 = (1, 1, 1, 1)$, as drawn in Figure \ref{fig:affineD4}.
For the affine $D_4$ root system, we use $\delta$ to denote the imaginary root. The extra simple root in affine $D_4$ root system is given by $e^0=\delta-\theta$.

Under this convention, the weights of inside $\mathbf{8}_V,\mathbf{8}_S,\mathbf{8}_C$ are as follows, which we denote as $\sfP(\mathbf{8}_V),\sfP(\mathbf{8}_S)$ and $\sfP(\mathbf{8}_C)$.
\begin{equation}\label{D4weightsystem}
\begin{aligned}
    \sfP(\mathbf{8}_V)=\{(\pm1,\pm1,0,0),(0,0,\pm1,\pm1)\}\\
    \sfP(\mathbf{8}_S)=\{(\pm1,0,\pm1,0),(0,\pm1,0,\pm1)\}\\
    \sfP(\mathbf{8}_C)=\{(\pm1,0,0,\pm1),(0,\pm1,\pm1,0)\}\\
\end{aligned}
\end{equation}
Here all the weights have the multiplicity one. 
Throughout the paper, we adopt the following notation:
\be 
\overline{t_j} \equiv t_j+t_j^{-1}~.
\ee 
Then, the characters can be expressed as
\begin{equation}
\begin{aligned}
        \chi_{\mathbf{8}_{V}}&=\sum_{w\in \sfP(\mathbf{8}_V)}\boldsymbol{t}^{w}=\bar{t}_1\bar{t}_2+\bar{t}_3\bar{t}_4=\theta_1\\
        \chi_{\mathbf{8}_{S}}&=\sum_{w\in \sfP(\mathbf{8}_S)}\boldsymbol{t}^{w}=\bar{t}_1\bar{t}_3+\bar{t}_2\bar{t}_4=\theta_2\\
        \chi_{\mathbf{8}_{C}}&=\sum_{w\in \sfP(\mathbf{8}_C)}\boldsymbol{t}^{w}=\bar{t}_1\bar{t}_4+\bar{t}_2\bar{t}_3=\theta_3\\
        \chi_{\text{adj}}&=\sum_{r\in \sfR(D_4)\cup\{(0,0,0,0)\}^{\times4}}\boldsymbol{t}^{r}= -4+\overline{t_1}^2+\overline{t_2}^2+\overline{t_3}^2+\overline{t_4}^2+\overline{t_1} \ \overline{t_2} \ \overline{t_3} \ \overline{t_4}=\theta_4\\
\end{aligned}
\end{equation}
where $\boldsymbol{t}^{w}=t_1^{w_1}t_2^{w_2}t_3^{w_3}t_4^{w_4}$, $\{(0,0,0,0)\}^{\times4}$ means four copies of the zero root coming from the Cartan subalgebra. In the last formula, we apply the identity $t_1^2+t_1^{-2}=\overline{t_1}^2-2$.

\subsubsection*{Infinite/finite-dimensional representations}

In \S\ref{sec:brane-poly}, we demonstrated that there are 24 line-like $(A, B, A)$-branes, whose support corresponds to the 24 lines in the cubic surface. These lines, characterized by their slopes, can be labeled by the 24 shortest weights in the $D_4$ weight lattice $\sfP(D_4)$. Accordingly, we label the infinite-dimensional representation associated with the shortest weight $w$ as $\scP^w$.

In this context, we denote the raising and lowering operators in the representation $\scP^w$ as $R^w_n$ and $L^w_n$, respectively, where the corresponding Askey-Wilson polynomials are denoted as $P^w_n$. In this notation,  the polynomial representation discussed in \S\ref{sec:polyrep} is denoted as $\scP^{w=(1,0,1,0)}$, with the raising and lowering operators in equation \eqref{raising_Lowering_operators} labeled as $R^{(1,0,1,0)}_n$ and $L^{(1,0,1,0)}_n$, respectively.  However, for brevity, we often omit the explicit weight $w = (1,0,1,0)$ when referring to the polynomial representation in \S\ref{sec:polyrep}, unless explicit clarification is needed.

\section{24 lines in affine cubic surface}\label{app:24lines}

We have mentioned in \S\ref{sec:brane-poly} that there are, in total, 24 lines in the affine cubic surface. We can use the symmetry action of $\SH$ in \S\ref{sec:SH} to identify the positions of these lines.
In \S\ref{sec:brane-poly}, we denote the slopes of the lines located on the plane where $x$,$y$ or $z$ is constant as $\bS_x,\bS_y,\bS_z$ respectively.  Moreover, the slopes of the lines $\bS_x,\bS_y,\bS_z$ are in one-to-one correspondence with the weights in the $SO(8)$ vector, spinor, and cospinor representations $\sfP(\mathbf{8}_V),\sfP(\mathbf{8}_S),\sfP(\mathbf{8}_C)$ respectively. 
\be\label{SlopeAppendix}
\begin{aligned}
        \bS_x &= \{-\boldsymbol{t}^w \mid w \in \sfP(\mathbf{8}_V)\}, \\
        \bS_y &= \{-\boldsymbol{t}^w \mid w \in \sfP(\mathbf{8}_S)\}, \\
        \bS_z &= \{-\boldsymbol{t}^w \mid w \in \sfP(\mathbf{8}_C)\},
\end{aligned}
\ee
where $\boldsymbol{t}^w=t_1^{w_1}t_2^{w_2}t_3^{w_3}t_4^{w_4}$. Therefore, we could label the 24 lines with the weight in the $\bfP_w$, whose slope is $-\boldsymbol{t}^w$.
As mentioned in \S\ref{sec:brane-poly}, the 24 lines serve as the support of $(A,B,A)$-brane inside the Hitchin moduli space. 

Eight lines on the plane where $x$ is constant are denoted by $\mathbf{P}_w$, where $w \in \mathsf{P}(\mathbf{8}_V)$.
\begin{equation}
    \begin{aligned}
        \bfP_{(-1,-1,0,0)}&=\left\{x=-t_1t_2-t_1^{-1}t_2^{-1},y= -t_1^{-1}t_2^{-1}z-t_1^{-1}t_3-t_1^{-1} t_3^{-1}-t_2^{-1}t_4-t_2^{-1} t_4^{-1}\right\},\\
        \bfP_{(1,1,0,0)}&=\left\{x=-t_1t_2-t_1^{-1}t_2^{-1},y= -t_1t_2 z-t_1t_3 -t_1t_3^{-1}-t_2 t_4-t_2t_4^{-1}\right\},\\
        \bfP_{(1,-1,0,0)}&=\left\{x=-t_1^{-1}t_2-t_1t_2^{-1},y= -t_1t_2^{-1} z- t_1t_3-t_1t_3^{-1}-t_2^{-1}t_4-t_2^{-1} t_4^{-1}\right\},\\
        \bfP_{(-1,1,0,0)}&=\left\{x=-t_1^{-1}t_2-t_1t_2^{-1},y= -t_1^{-1}t_2 z-t_2t_4 -t_2t_4^{-1}-t_1^{-1}t_3-t_1^{-1} t_3^{-1}\right\},\\
        \bfP_{(0,0,-1,-1)}&=\left\{x=-t_3t_4-t_3^{-1}t_4^{-1}, y= -t_3^{-1} t_4^{-1}z-t_1t_3^{-1}-t_1^{-1} t_3^{-1}-t_2t_4^{-1}-t_2^{-1} t_4^{-1}\right\},\\
        \bfP_{(0,0,1,1)}&=\left\{x=-t_3t_4-t_3^{-1}t_4^{-1},y= -t_3t_4 z-t_1 t_3-t_1^{-1}t_3-t_2 t_4-t_2^{-1}t_4\right\},\\
        \bfP_{(0,0,1,-1)}&=\left\{x=-t_3^{-1}t_4-t_3t_4^{-1},y= -t_3t_4^{-1} z-t_1 t_3-t_1^{-1}t_3- t_2^{-1}t_4^{-1}-t_2t_4^{-1}\right\},\\
        \bfP_{(0,0,-1,1)}&=\left\{x=-t_3^{-1}t_4-t_3t_4^{-1},y= -t_3^{-1}t_4 z-t_1t_3^{-1}- t_1^{-1}t_3^{-1}-t_2 t_4-t_2^{-1}t_4\right\}.
    \end{aligned}
\end{equation}

Eight lines on the plane where $y$ is constant are denoted by $\mathbf{P}_w$, where $w \in \mathsf{P}(\mathbf{8}_S)$.
\begin{equation}
    \begin{aligned}
        \bfP_{(1,0,1,0)}&=\left\{y=-t_1t_3-t_1^{-1}t_3^{-1},z= -t_1t_3 x- t_1t_4-t_1t_4^{-1}-t_2 t_3-t_2^{-1}t_3\right\},\\
        \bfP_{(-1,0,-1,0)}&=\left\{y=-t_1t_3-t_1^{-1}t_3^{-1},z= -t_1^{-1} t_3^{-1}x-t_1^{-1}t_4-t_1^{-1} t_4^{-1}-t_2t_3^{-1}-t_2^{-1}t_3^{-1}\right\},\\
        \bfP_{(-1,0,1,0)}&=\left\{y=-t_1^{-1}t_3-t_1t_3^{-1},z= -t_1^{-1}t_3 x-t_2 t_3-t_2^{-1}t_3-t_1^{-1}t_4-t_1^{-1} t_4^{-1}\right\},\\
        \bfP_{(1,0,-1,0)}&= \left\{y=-t_1^{-1}t_3-t_1t_3^{-1},z= -t_1 t_3^{-1}x- t_1t_4-t_1t_4^{-1}-t_2t_3^{-1}-t_2^{-1} t_3^{-1}\right\},\\
        \bfP_{(0,1,0,1)}&=\left\{y=-t_2t_4-t_2^{-1}t_4^{-1},z= -t_2 t_4  x-t_2 t_3 -t_2t_3^{-1}-t_1 t_4-t_1^{-1}t_4\right\},\\
        \bfP_{(0,-1,0,-1)}&=\left\{y=-t_2t_4-t_2^{-1}t_4^{-1},z= -t_2^{-1} t_4^{-1}x-t_2^{-1}t_3-t_2^{-1} t_3^{-1}-t_1t_4^{-1}-t_1^{-1} t_4^{-1}\right\},\\
        \bfP_{(0,-1,0,1)}&=\left\{y=-t_2^{-1}t_4-t_2t_4^{-1},z= -t_2^{-1}t_4 x-t_2^{-1}t_3-t_1 t_4-t_1^{-1}t_4-t_2^{-1} t_3^{-1}\right\},\\
        \bfP_{(0,1,0,-1)}&=\left\{y=-t_2^{-1}t_4-t_2t_4^{-1},z= -t_2 t_4^{-1} x-t_1t_4^{-1}-t_2 t_3-t_2t_3^{-1}- t_1^{-1}t_4^{-1}\right\}. \\
    \end{aligned}
\end{equation}

Eight lines on the plane where $z$ is constant are denoted by $\mathbf{P}_w$, where $w \in \mathsf{P}(\mathbf{8}_C)$.
\begin{equation}
    \begin{aligned}
        \bfP_{(-1,0,0,-1)}&=\left\{z=-t_1t_4-t_1^{-1}t_4^{-1}, x= -t_1^{-1} t_4^{-1}y-t_1^{-1}t_2-t_1^{-1} t_2^{-1}-t_3t_4^{-1}-t_3^{-1} t_4^{-1}\right\},\\
        \bfP_{(1,0,0,1)}&=\left\{z=-t_1t_4-t_1^{-1}t_4^{-1},x= -t_1 t_4 y-t_1t_2 -t_1t_2^{-1}-t_3 t_4-t_3^{-1}t_4\right\},\\
        \bfP_{(1,0,0,-1)}&=\left\{z=-t_1^{-1}t_4-t_1t_4^{-1},x= -t_1 t_4^{-1}y-t_1t_2 -t_1t_2^{-1}-t_3t_4^{-1}-t_3^{-1} t_4^{-1}\right\},\\
        \bfP_{(-1,0,0,1)}&=\left\{z=-t_1^{-1}t_4-t_1t_4^{-1},x= -t_1^{-1}t_4 y-t_1^{-1}t_2-t_1^{-1} t_2^{-1}-t_3 t_4-t_3^{-1}t_4\right\},\\
        \bfP_{(0,-1,-1,0)}&=\left\{z=-t_2t_3-t_2^{-1}t_3^{-1},x= -t_2^{-1} t_3^{-1}y-t_1t_2^{-1}-t_1^{-1} t_2^{-1}-t_3^{-1}t_4-t_3^{-1} t_4^{-1}\right\},\\
        \bfP_{(0,1,1,0)}&=\left\{z=-t_2t_3-t_2^{-1}t_3^{-1},x=- t_2t_3 y-t_1 t_2-t_1^{-1}t_2-t_3 t_4-t_3t_4^{-1}\right\}\\
        \bfP_{(0,1,-1,0)}&=\left\{z=-t_2^{-1}t_3-t_2t_3^{-1},x= -t_2t_3^{-1} y-t_1 t_2-t_1^{-1}t_2-t_3^{-1}t_4-t_3^{-1} t_4^{-1}\right\},\\
        \bfP_{(0,-1,1,0)}&=\left\{z=-t_2^{-1}t_3-t_2t_3^{-1},x= -t_2^{-1}t_3 y-t_1t_2^{-1}-t_1^{-1}t_2 ^{-1}-t_3 t_4-t_3t_4^{-1}\right\}
    \end{aligned}
\end{equation}

As a remark, the eight polynomial representations $\mathcal P^{\epsilon_0,\epsilon_1}, \overline{\mathcal P}^{\delta_0,\delta_1}$ $(\epsilon_i=\pm 1, \delta_i=\pm 1)$  of $\HH$ discussed in \cite{oblomkov2009finite} correspond to the eight lines $\mathbf{P}_w$ with $w \in \mathsf{P}(\mathbf{8}_V)$. Specifically,
\begin{equation}
    \mathcal P^{\epsilon_0,\epsilon_1} \leftrightarrow\mathbf{P}_{(0,0,\epsilon_0,\epsilon_1)}~,\quad \overline{\mathcal P}^{\delta_0,\delta_1} \leftrightarrow \mathbf{P}_{(\delta_0,\delta_1,0,0)}~.
\end{equation}

\begin{figure}[ht]
    \centering
    \includegraphics{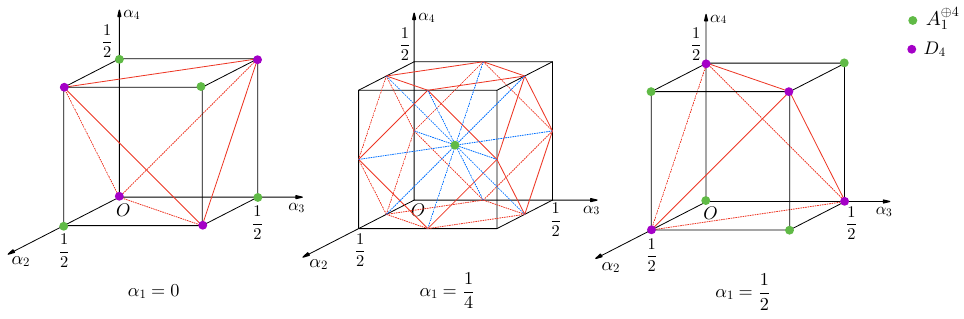}
    \caption{3-cube cross sections of $\mathsf{Cube}$ at $\talpha_1=0,\frac{1}{4},\frac{1}{2}$. The $A_1^{\oplus4}$ and $D_4$ singularities are marked out, which show up only for these special values of $\talpha_1$. }
    \label{fig:chambers_special_alpha1}
\end{figure}

\section{Chamber structures}\label{app:chamber-structure}

In this Appendix, we provide a detailed analysis of the chamber structure for the volumes of the irreducible components in the global nilpotent cone, as discussed in \S\ref{sec:wallcrossing}. Thanks to the periodicity of the parameters $\talpha_j \to \talpha_j + 1$ and the Weyl group symmetry $\talpha_j \to -\talpha_j$ of the cubic equation, we can restrict the parameter space to the 4-dimensional unit cube:
\be
\mathsf{Cube} = \left\{ (\talpha_1, \talpha_2, \talpha_3, \talpha_4) \mid \talpha_j \in \left[0, \frac{1}{2} \right] \right\} \subset \bR^4.
\ee

The walls described in \eqref{wall} divide $\mathsf{Cube}$ into 24 distinct chambers. Singularities of all types, as classified in Table \ref{tab:surface_sing_classification}, appear on these walls and at their intersection points. The configurations of $\bfV$ and $\bfD_j$ corresponding to different singularities are illustrated in Figure \ref{fig:du_Val}.

Within each chamber, the volumes of $\bfV$ and $\bfD_j$ are linear functions of the parameters $\talpha_j$. As we move from one chamber to another, crossing a wall, the volume functions exhibit a wall-crossing phenomenon.

We can slice the cube $\mathsf{Cube}$ along a fixed value of $\talpha_1$, resulting in a 3-dimensional sub-cube, as shown in Figure \ref{fig:chambers}:
\begin{enumerate}[nosep]
    \item Each of the eight vertices of this 3-cube, which corresponds to an edge of $\mathsf{Cube}$, lies entirely in a single chamber (including its boundary). These vertices correspond either to $A_1^{\oplus3}$ singularities (shown in blue) or $A_3$ singularities (shown in red).
    \item Each line segment in Figure \ref{fig:chambers} represents either an $A_1^{\oplus2}$ singularity (in blue or black) or an $A_2$ singularity (in red). 
    \item Each wall in Figure \ref{fig:chambers} corresponds to an $A_1$ singularity, where either $\bfV$ or some $\bfD_j$ shrinks to a point. The positions of these walls are summarized in Table \ref{tab:A1_wall_position} and are explicitly illustrated in Figure \ref{fig:A1_wall}.
\end{enumerate}
At special values of $\talpha_1$, $A_1^{\oplus4}$ or $D_4$ singularities may appear, as shown in Figure \ref{fig:chambers_special_alpha1}.

\begin{table}
    \centering
    \begin{tabular}{c|c}
         $A_1$ singularities&  Positions of $A_1$ walls\\
         \hline
         $\textrm{vol}_I(\bfD_1)=0$& 
    $BCHLIF,~BCQ,~CHR,~HLV,~LIS,~IFT,~FBP$\\
    \hline
 $\textrm{vol}_I(\bfD_2)=0$&$CDEIJG,~CDR,~DEO,~EIS,~IJT,~JGU,~GCQ$\\
 \hline
 $\textrm{vol}_I(\bfD_3)=0$&$ABGKLE,~ABP,~BGQ,~GKU,~KLV,~LES,~EAO$\\
 \hline
 $\textrm{vol}_I(\bfD_4)=0$&$ADHKJF,~ADO,~DHR,~HKV,~KJU,~JFT,~FAP$\\
 \hline
 $\textrm{vol}_I(\bfV)=0$&$ABCD,AEIF,BFJG,CGKH,DELH,IJKL,ADE,BCG,IJF,HKL$\\
 \end{tabular}
    \caption{Positions of $A_1$ walls in Figure \ref{fig:chambers} when either $\bfV$ or one of $\bfD_j$ shrinks to a point.}
    \label{tab:A1_wall_position}
\end{table}

\begin{figure}
    \centering
    \includegraphics[width=0.8\linewidth]{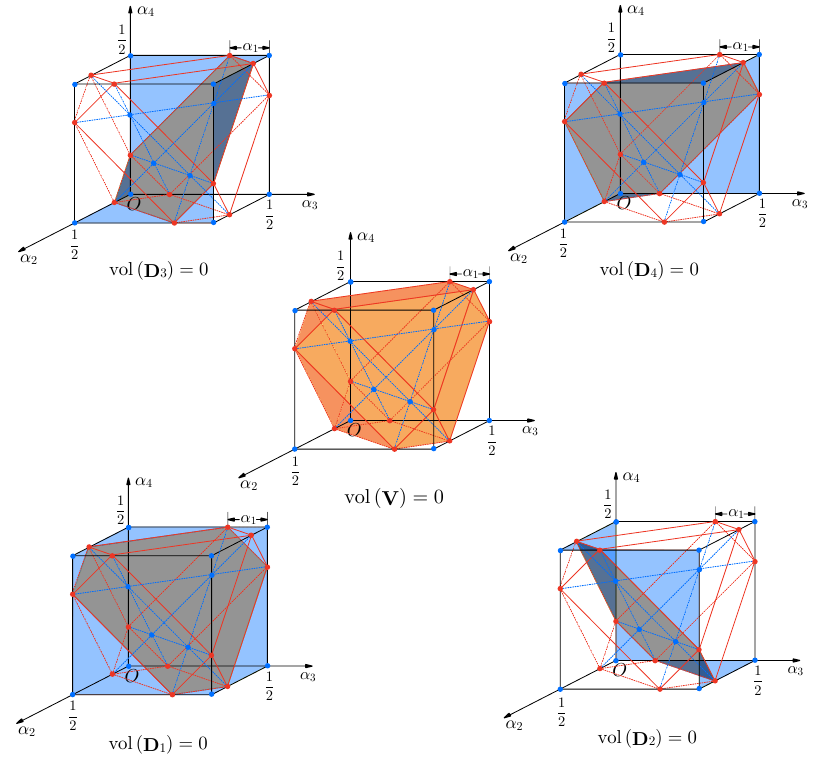}
    \caption{$A_1$ singularities can show up in 5 cases where either $\bfV$ or one $\bfD_j$ shrinks to a point. They correspond to the planes inside the 3-cube in Figure \ref{fig:chambers} as summarized in Table \ref{tab:A1_wall_position}. }
    \label{fig:A1_wall}
\end{figure}

At the origin, where $\talpha_j = 0$ for all $j = 1, 2, 3, 4$, the cubic surface exhibits a $D_4$ singularity. At this point, $\bfV$ and three of the $\bfD_j$'s vanish, as shown in the bottom-right case of Figure \ref{fig:du_Val}. In $\mathsf{Cube}$, the origin lies at the boundary of 12 chambers simultaneously, four of which contain the $\talpha_j$-axes. The chamber that contains the $\talpha_1$-axis is the tetrahedral chamber $OADE$ in the 3-cube shown in Figure \ref{fig:chambers}. 

We label the four $\bfD_j$'s as $j = 1, 2, 3, 4$, and in this chamber, the volumes are given by:
\be
\textrm{vol}_I(\bfD_j) = \left( 1 - 2\talpha_1, 2\talpha_2, 2\talpha_3, 2\talpha_4 \right), \quad \textrm{vol}_I(\bfV) = \talpha_1 - \talpha_2 - \talpha_3 - \talpha_4.
\ee
In this convention, the volume functions for the 8 chambers $OADE$, $PABF$, $RCDH$, $SEIL$, $QBCG$, $TFIJ$, $VHKL$, and $UGJK$ in the 3-cube are listed in Table \ref{tab:chamber_8_vertices}.

\begin{table}
    \centering
    \begin{tabular}{c|c|c|c}
     Chambers&vertices contained& $\textrm{vol}_I(\bfD_j)$& $\textrm{vol}_I(\bfV)$ \\
    \hline
     $OADE$&$(0,0,0)$ & $(1-2\talpha_1,2\talpha_2,2\talpha_3,2\talpha_4)$& $\talpha_1-\talpha_2-\talpha_3-\talpha_4$\\
    \hline
     $PABF$&$(1,0,0)$ & $(1-2\talpha_2,2\talpha_1,2\talpha_4,2\talpha_3)$& $-\talpha_1+\talpha_2-\talpha_3-\talpha_4$\\
    \hline
     $RCDH$&$(0,1,0)$ & $(1-2\talpha_3,2\talpha_4,2\talpha_1,2\talpha_2)$& $-\talpha_1-\talpha_2+\talpha_3-\talpha_4$\\
    \hline
     $SEIL$&$(0,0,1)$ & $(1-2\talpha_4,2\talpha_3,2\talpha_2,2\talpha_1)$& $-\talpha_1-\talpha_2-\talpha_3+\talpha_4$\\
    \hline
     $QBCG$&$(1,1,0)$ & $(2\talpha_4,1-2\talpha_3,1-2\talpha_2,2\talpha_1)$& $-\frac{1}{2} -\talpha_1+\talpha_2+\talpha_3-\talpha_4$\\
    \hline
     $TFIJ$&$(1,0,1)$ & $(2\talpha_3,1-2\talpha_4,2\talpha_1,1-2\talpha_2)$& $-\frac{1}{2}- \talpha_1+\talpha_2-\talpha_3+\talpha_4$\\
    \hline
     $VHKL$&$(0,1,1)$ & $(2\talpha_2,2\talpha_1,1-2\talpha_4,1-2\talpha_3)$& $-\frac{1}{2}- \talpha_1-\talpha_2+\talpha_3+\talpha_4$\\
    \hline
     $UGJK$&$(1,1,1)$ & $(2\talpha_1,1-2\talpha_2,1-2\talpha_3,1-2\talpha_4)$& $-1 -\talpha_1+\talpha_2+\talpha_3+\talpha_4$\end{tabular}
    \caption{Volume functions at 8 chambers containing the 3-cube vertices. The first column lists the chambers in Figure \ref{fig:chambers}. The second column lists the 3-cube vertices contained in each chamber.}
    \label{tab:chamber_8_vertices}
\end{table}

There are 16 additional chambers, which we call the internal chambers. The volume functions for these chambers are uniformly expressed as

\begin{align}\label{internal_chambers}
    \textrm{vol}_I(\bfD_j) &= \Big( \left| -1 + \talpha_1 + \talpha_2 + \talpha_3 + \talpha_4 \right|, \left| \talpha_1 + \talpha_2 - \talpha_3 - \talpha_4 \right|, \cr
    & \quad \left| \talpha_1 - \talpha_2 + \talpha_3 - \talpha_4 \right|, \left| \talpha_1 - \talpha_2 - \talpha_3 + \talpha_4 \right| \Big)~, \cr
    \textrm{vol}_I(\bfV) &= \frac{1}{2} \left( 1 - \sum_{j=1}^4 \textrm{vol}_I(\bfD_j) \right).
\end{align}

In these 16 chambers, the expressions in the absolute values in \eqref{internal_chambers} take different combinations of signs. Since there are four parameters $\talpha_j$, there are $2^4 = 16$ possible sign combinations, each corresponding to one of the 16 chambers. The specific sign combinations and the corresponding volumes of $\bfV$ are explicitly listed in Table \ref{tab:internal_chambers}.

\begin{table}\footnotesize
    \centering
    \begin{tabular}{c|c|c|c|c|c}
     Chambers&signs & $\textrm {vol}_I(\bfV)$ &  Chambers&signs & $\textrm {vol}_I(\bfV)$ \\[.3em]
    \hline
     $WXYZ-$&$(----)$& $2\talpha_1$&  $WXYZ+$&$(++++)$& $1-2\talpha_1$\\[.3em]
    \hline
    $WDEXHL$ &$(--++)$& $2\talpha_2$& $ZBGYFJ$ &$(++--)$& $1-2\talpha_2$\\[.3em]
    \hline
    $WAEYFI$ &$(-+-+)$& $2\talpha_3$& $ZCGXHK$ &$(+-+-)$& $1-2\talpha_3$\\[.3em]
    \hline
    $WADZBC$ &$(-++-)$& $2\talpha_4$& $XKLYJI$ &$(+--+)$& $1-2\talpha_4$\\[.3em]
    \hline
     $XYZKJG$&$(+---)$ & $1+ \talpha_1-\talpha_2-\talpha_3-\talpha_4$&  $WADE$&$(-+++)$ & $-\talpha_1+\talpha_2+\talpha_3+\talpha_4$\\[.3em]
    \hline
     $WYZAFB$&$(-+--)$& $ \talpha_1-\talpha_2+\talpha_3+\talpha_4$&  $XHKL$&$(+-++)$& $1-\talpha_1+\talpha_2-\talpha_3-\talpha_4$\\[.3em]
    \hline
     $WXZDHC$&$(--+-)$& $ \talpha_1+\talpha_2-\talpha_3+\talpha_4$&  $YFIJ$&$(++-+)$& $1-\talpha_1-\talpha_2+\talpha_3-\talpha_4$\\[.3em]
    \hline
     $WXYELI$&$(---+)$& $\talpha_1+\talpha_2+\talpha_3-\talpha_4$& $ZBCG$ &$(+++-)$& $1- \talpha_1-\talpha_2-\talpha_3+\talpha_4$\end{tabular}
    \caption{Volume functions at 16 internal chambers. The first and fourth columns list the chambers in Figure \ref{fig:chambers}. In particular $WXYZ-$ and $WXYZ+$ are the center chambers when $\talpha_1<\frac{1}{4}$ and $\talpha_1>\frac{1}{4}$ respectively. The second and fifth columns list the signs of the expressions in the absolute values in \eqref{internal_chambers}.}
    \label{tab:internal_chambers}
\end{table}

There are two conditions, given by \eqref{A1_t_cond1} and \eqref{A1_t_cond2}, under which we reach the \texorpdfstring{$A_1$}{A1} limit of the DAHA. The loci of these points in $\mathsf{Cube}$ are shown in Figure \ref{fig:A1_limit_cube}.

\begin{figure}[ht]
    \centering
    \includegraphics{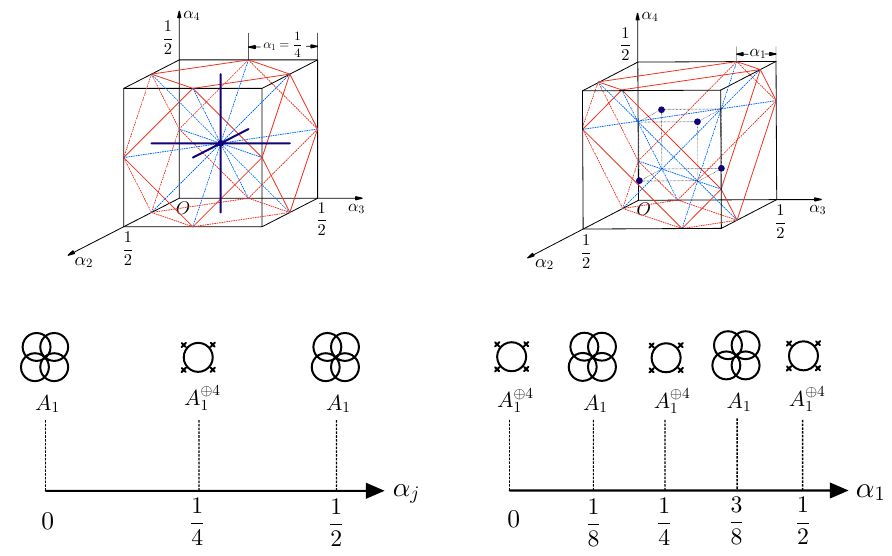}
    \caption{There are four lines satisfying the first $A_1$ limit condition \eqref{A1_t_cond1} in $\mathsf{Cube}$, along whom three of $\talpha_j$ equal $\frac{1}{2}$ and the remaining $\talpha$ is the direction this line extends along. Three of these lines are visible in the 3-cube when $\talpha_1=\frac{1}{4}$, as drawn in dark blue in the top left figure. The remaining line is the trajectory of the center point $(\talpha_1,\frac{1}{4},\frac{1}{4},\frac{1}{4})$ propagating along $\talpha_1$ direction in $\mathsf{Cube}$. There are again four lines satisfying the second $A_1$ limit condition \eqref{A1_t_cond2}. Along three of these lines, three of $\talpha_j$ equals $\talpha_1$ and the remaining one $\frac{1}{2}-\talpha_1$. Along the remaining line, all $\talpha_j$ equal $\frac{1}{2}-\talpha_1$ except $\talpha_1$ itself. They are drawn on the top right, and there are four distinct points whose trajectories along $\talpha_1$ are these lines.  The du Val singularities along each of the 8 lines are drawn out in the bottom figures.}
    \label{fig:A1_limit_cube}
\end{figure}

\section{Winding cycles and Argyres-Douglas surface}\label{app:monodromy}

In this Appendix, we provide a detailed analysis of the monodromies associated with singular fibers and demonstrate the existence of 2-cycles that wind around these fibers. Furthermore, we pinpoint the locations of Argyres-Douglas (AD) points, which arise from the collisions of singular fibers containing mutually non-local degrees of freedom. As a result, we reveal that the AD surfaces introduce a non-trivial monodromy to the $\gamma$ space. This insight explains why the winding cycles are homologous to the straight ones.

\subsubsection*{Winding cycles}
As discussed in \S\ref{sec:classification}, in addition to the 2-cycles depicted in Figure \ref{fig:6I1 cycles} (referred to as \emph{straight cycles} in the following), there exist additional cycles that wind around singular fibers. 

To demonstrate their existence, we begin by examining how the monodromies evolve when two $I_1$ fibers rotate around each other. Consider two singular fibers labeled 1 and 2, with monodromies $M_1$ and $M_2$, as shown in Figure \ref{fig:rotate 2 fibers}. After successive rotations of $\pi$,
\begin{equation}\label{monodromy_rotate}
\begin{aligned}
    M_1&\xrightarrow{\pi} M_2 M_1 M_2^{-1}\xrightarrow{\pi} (M_1M_2)M_1(M_1M_2)^{-1}\xrightarrow{\pi} (M_2M_1M_2)M_1(M_2M_1M_2)^{-1}\rightarrow\cdots,\\
    M_2&\xrightarrow{\pi} M_2 M_2 M_2^{-1} \xrightarrow{\pi} (M_1M_2)M_2(M_1M_2)^{-1}\xrightarrow{\pi} (M_2M_1M_2)M_2(M_2M_1M_2)^{-1}\rightarrow\cdots.
\end{aligned} 
\end{equation}

\begin{figure}[ht]
    \centering
    \includegraphics[width=1\linewidth]{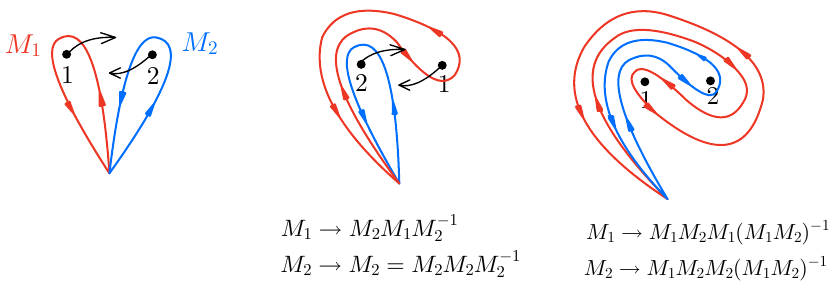}
    \caption{Rotate two singular fibers with monodromies $M_1$ and $M_2$ respectively. Every $\pi$ rotation conjugates both monodromies and $M_i$. When $M_1=M_2$, monodromies are unchanged and when $M_1=M_q,M_2=M_d$, monodromies behave as \eqref{monodromy_periodicity}. }
    \label{fig:rotate 2 fibers}
\end{figure}

When $M_1=M_2$ equals $M_q$ or $M_d$ defined above \eqref{modular_matrices}, the monodromies are unaltered under any rotation. However, this invariance does not hold when $M_1 = M_q$ and $M_2 = M_d$. A straightforward calculation shows that the resulting monodromies evolve as follows:  
\begin{equation}\label{monodromy_periodicity}
   \begin{aligned}
       M_q &\xrightarrow{3\pi} M_d \xrightarrow{3\pi} M_q, \\
       M_d &\xrightarrow{3\pi} M_q \xrightarrow{3\pi} M_d.
   \end{aligned}
\end{equation}
This demonstrates that after a $3\pi$ rotation, the monodromy matrices of the two fibers are exchanged. Following another $3\pi$ rotation, their monodromy matrices return to their original values.

Let us now examine the implications of this property for the second homology. Recall that a suspended cycle is constructed as the result of a deformation process between two singular fibers of the same type, i.e., fibers with identical monodromy matrices, in a fixed frame. 

Consider two singular fibers located at $p_i$ and $p_j$, and assume their monodromy matrices are identical. In this case, there exists a straight suspended cycle, denoted $\bfW$, connecting these two fibers. Now, suppose there is an additional singular fiber at $p_k$ nearby. If the monodromy matrix around $p_k$ equals those around $p_i$ and $p_j$, the deformation process between $p_i$ and $p_j$ can wind around $p_k$ an arbitrary number $n\in\bZ$ of times, since rotations involving $p_k$ do not alter the monodromies. Consequently, we can construct a cycle suspended between $p_i$ and $p_j$ that winds around $p_k$ for $n$ times. Extending this idea, we can construct cycles suspended between any pair of the three singular fibers while winding around the third. However, these winding cycles can always be decomposed into straight cycles suspended between the fibers.

On the other hand, if the monodromy matrix around $p_k$ is distinct from those around $p_i$ and $p_j$, the $6\pi$ periodicity described in \eqref{monodromy_periodicity} implies that the deformation process between $p_i$ and $p_j$ can wind around $p_k$ an additional $3n$ times, where $n \in \bZ$. This results in the construction of a new cycle, $\bfW^\prime$, suspended between $p_i$ and $p_j$ and winding $p_k$ $3n$ times more than $\bfW$, as illustrated by the red cycle in Figure \ref{fig:winding cycles}. Similarly, by the $3\pi$ exchange in \eqref{monodromy_periodicity}, we can construct another cycle, $\bfW^{\prime\prime}$, suspended between $p_j$ and $p_k$, shown as the blue cycle in Figure \ref{fig:winding cycles}. 

Unlike the previous case, these winding cycles cannot be directly decomposed into straight cycles. However, as we will demonstrate later, these cycles are homologous to the straight cycles by going around the AD surfaces in the parameter space.

\begin{figure}[ht]
    \centering
    \includegraphics[width=0.3\linewidth]{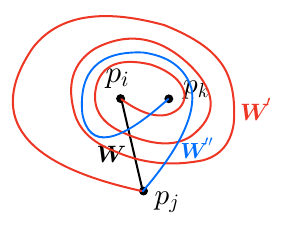}
    \caption{For every cycle suspended between $p_i$ and $p_j$, if there is a nearby singular fiber $p_k$ of a distinct type, it is possible to construct cycles suspended between $p_i$ and $p_j$ that wind around $p_k$ $3n$ times for any $n \in \bZ$. For instance, one such cycle is $\bfW^\prime$, where $n = -1$. Additionally, there exists another cycle, $\bfW^{\prime\prime}$, suspended between $p_j$ and $p_k$, which winds around $p_i$ . }
    \label{fig:winding cycles}
\end{figure}

\subsubsection*{Argyres-Douglas types of singular fibers}

When two $I_1$ fibers of distinct types, i.e., with distinct monodromies, collide into a single fiber, no new fiber cycle is generated, as there is no straight suspended cycle between them. We refer to this type of collision as the AD collision. This terminology reflects the fact that distinct types of $I_1$ fibers correspond to mutually non-local massless degrees of freedom, and their collision leads to the SCFT of AD type discovered in \cite{Argyres:1995jj,Argyres:1995xn}. 

This phenomenon can be verified by analyzing the Weierstrass form of the Seiberg-Witten curve at the collision points, which takes the form
\begin{equation}\label{AD_SW}
    y^2 = (x - x_0(\tau))^3,
\end{equation}
for some $\tau$-dependent $x_0$.

The Kodaira type of the resulting AD singular fiber can be determined by computing the product of the monodromy matrices associated with the colliding fibers. Choosing the frame as described in \S\ref{sec:classification}, let the monodromy matrices around the two $I_1$ fibers be $M_q$ and $M_d$, respectively. The monodromy of the resulting singular fiber is given by\footnote{The singular type of the resulting fiber is independent of the order of the matrix product.}
\begin{equation}
      M_d\cdot M_q=\begin{pmatrix}
-1 & 3 \\
-1 & 2 
\end{pmatrix}  \sim 
\begin{pmatrix}
1 & 1 \\
-1 & 0 
\end{pmatrix} ,
\end{equation}
which is the monodromy matrix associated with fibers of Kodaira type $II$. Geometrically, the $M_q$ fibers are tori with $(0,1)$-cycles shrinking to a point, while $M_d$ fibers are tori with $(1,-2)$-cycles shrinking to a point. The collision of an $M_q$ fiber and an $M_d$ fiber is shown at the bottom of Figure \ref{fig:collide I1}.

Physically, the resulting AD singular fiber corresponds to the $(A_1, A_2)$ AD theory, as both a quark and a dyon become massless simultaneously. A similar analysis can be carried out for other collisions, such as when an $I_2$ or $I_3$ fiber collides with a dyon $I_1$ fiber. These collisions result in singular fibers of Kodaira types $III$ and $IV$, respectively, corresponding to the $(A_1, A_3)$ and $(A_1, D_4)$ AD theories. The possibilities of AD-type collisions are summarized as the red arrows in Figure \ref{fig:Kodaira}. As noted in \cite{Argyres:1995jj}, there are no higher-level AD theories in the Coulomb branch moduli space of SU(2) $N_f = 4$ SQCD.

\begin{figure}[ht]
    \centering
    \includegraphics[width=0.65\linewidth]{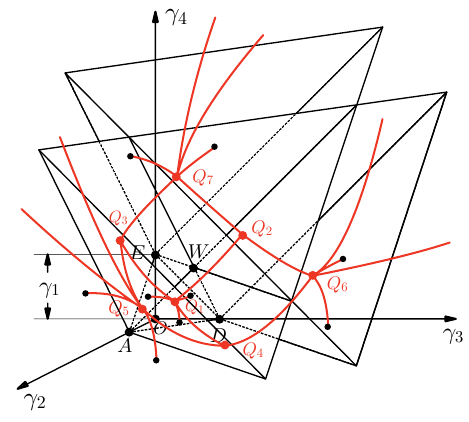}
    \caption{The structure of the $\gamma$-space is visualized by taking a section along a fixed value of $\tgamma_1$. By the black lines, we draw the walls in $\gamma$-space where normal collisions occur, such as $\Delta ADE$ and $\Delta WDE$, as determined by the conditions in Table \ref{tab:fiber_sing_classification2}. The red curves correspond to AD surfaces, which are codimension-2 surfaces defined by solving equation \eqref{ADeqn}. Points along these red curves correspond to $(A_1, A_2)$ AD theories, while the vertex points labeled $Q_1$ to $Q_7$ represent $(A_1, A_3)$ AD theories. The numerical visualization of these AD surfaces is shown in Figure \ref{fig:mmaAD}.}
    \label{fig:gamma_space}
\end{figure}

\begin{figure}[ht]
    \centering
    \includegraphics[width=0.5\linewidth]{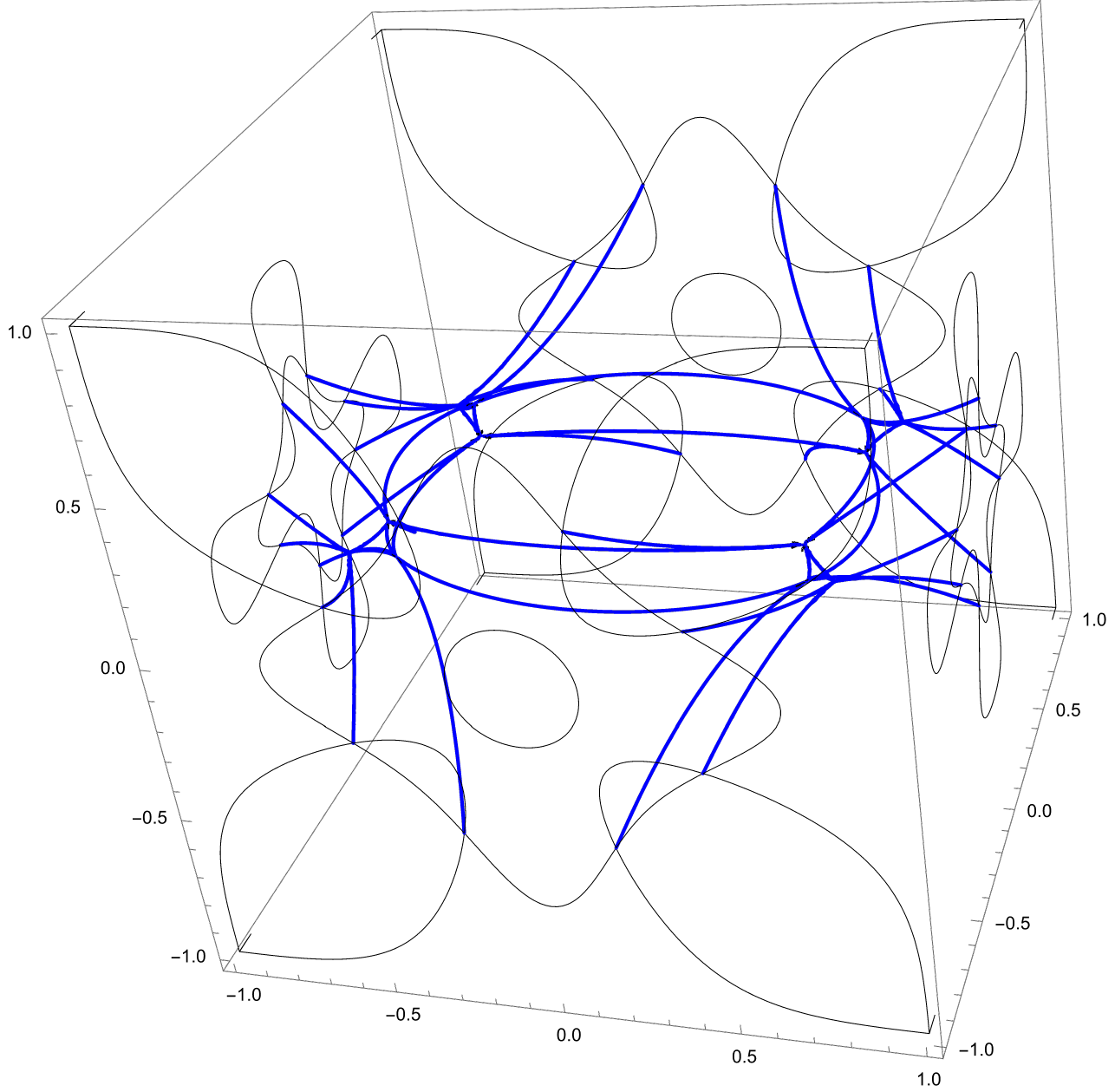}
    \caption{The Argyres-Douglas surfaces, as plotted by Mathematica, correspond to the case where $q = 0.001i$ and $\tgamma_1 = \frac{1}{2}$.}
    \label{fig:mmaAD}
\end{figure}

\begin{figure}[ht]
    \centering
    \includegraphics[width=0.75\linewidth]{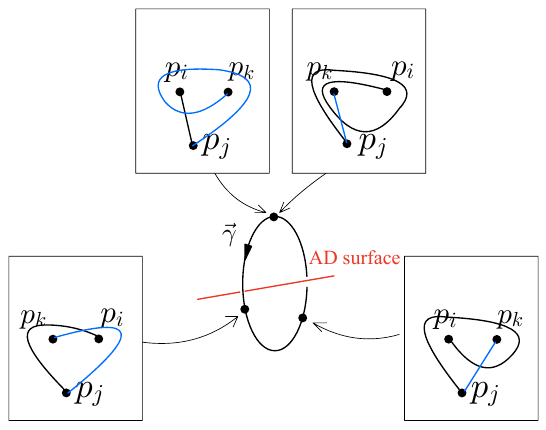}
    \caption{In the $\tgamma$ space, monodromy around AD-surface colliding $p_i$ and $p_k$ will produce a $3\pi$ rotation between $p_i$ and $p_k$. Therefore, the straight cycle between $p_i$ and $p_j$ will wind around $p_k$, while the winding cycle between $p_j$ and $p_k$ will be straightened.}
    \label{fig:monodromy_around_AD_surface}
\end{figure}

The conditions for the appearance of AD-type singular fibers are more complicated than those for the usual $I_k$-type or $I_0^*$-type fibers. In addition to tuning the mass parameters, the complex coupling $\tau$ must also be carefully adjusted. These conditions can be determined by solving \eqref{AD_SW} directly. The generic Seiberg-Witten curve is given by \eqref{SW-Nf4-generic}. By equating \eqref{AD_SW} with this generic form and eliminating $x_0$, one can derive two complex equations involving the $u$ parameter.
\begin{multline}\nonumber
    9B \tilde u^2+12A(\tgamma_1^2+\tgamma_2^2+\tgamma_3^2+\tgamma_4^2)\tilde u+16B^2(\tgamma_1^4+\tgamma_2^4+\tgamma_3^4+\tgamma_4^4)\\
 +16C(-\vartheta_3^4,\vartheta_4^4)(\tgamma_1^2\tgamma_2^2+\tgamma_3^2\tgamma_4^2)+16C(\vartheta_2^4,-\vartheta_3^4)(\tgamma_1^2\tgamma_3^2+\tgamma_2^2\tgamma_4^2)+16C(\vartheta_4^4,\vartheta_2^4)(\tgamma_1^2\tgamma_4^2+\tgamma_2^2\tgamma_3^2)=0,
    \end{multline}
    \begin{multline}\label{ADeqn}
    27A\tilde u^3-144\Big(AB(\tgamma_1^4+\tgamma_2^4+\tgamma_3^4+\tgamma_4^4)\\
   \qquad +D(-\vartheta_3^4,\vartheta_4^4)(\tgamma_1^2\tgamma_2^2+\tgamma_3^2\tgamma_4^2)+D(\vartheta_2^4,-\vartheta_3^4)(\tgamma_1^2\tgamma_3^2+\tgamma_2^2\tgamma_4^2)+D(\vartheta_4^4,\vartheta_2^4)(\tgamma_1^2\tgamma_4^2+\tgamma_2^2\tgamma_3^2)\Big)\tilde u\\
    +64A^2 \left(\tgamma _1^6+\tgamma _2^6+\tgamma _3^6+\tgamma _4^6\right)-512 B^3 \left(\tgamma _1^2+\tgamma _2^2+\tgamma _3^2+\tgamma _4^2\right)^3\\-192  F\left(\tgamma _1^2\tgamma _2^2 \left(\tgamma _3^2+\tgamma _4^2\right) +\tgamma _3^2 \tgamma _4^2\left(\tgamma _1^2+\tgamma _2^2\right) \right)\\
    +192  E(-\vartheta_3^4,\vartheta_4^4)\left(\tgamma _1^2  \tgamma _2^2\left(\tgamma _1^2+\tgamma _2^2\right)+\tgamma _3^2 \tgamma _4^2 \left(\tgamma _3^2+\tgamma _4^2\right)\right)\\+192 E(\vartheta_2^4,-\vartheta_3^4)\left(\tgamma _1^2  \tgamma _3^2\left(\tgamma _1^2+\tgamma _3^2\right)+\tgamma _2^2 \tgamma _4^2 \left(\tgamma _2^2+\tgamma _4^2\right)\right)  \\
    +192  E(\vartheta_4^4,\vartheta_2^4)\left(\tgamma _1^2 \tgamma _4^2 \left(\tgamma _1^2+\tgamma _4^2\right)+\tgamma _2^2 \tgamma _3^2 \left(\tgamma _2^2+\tgamma _3^2\right)\right) =0,
\end{multline}
with
\begin{equation}
\begin{aligned}
\tilde u&=u-\frac{2}{3}\vartheta_2^4 (\tgamma_1^2+\tgamma_2^2+\tgamma_3^2+\tgamma_4^2),\\
    A&=\big[(-\vartheta_3^4)-\vartheta_4^4\big]\big[\vartheta_2^4-(-\vartheta_3^4)\big]\big[\vartheta_2^4-\vartheta_4^4\big],\\
    B&=\vartheta_2^8-\vartheta_2^4\vartheta_3^4+\vartheta_3^8,\\
    C(a,b)&=2 a^4+31a^3b+60 a^2b^2+31ab^3+2 b^4,\\
    D(a,b)&=(a-b)(4 a^4-13 a^3b-36 a^2b^2-13 ab^3+4 b^4),\\
 E(a,b)&=\left(a-b\right)^2(4 a^4-7 a^3b-21 a^2b^2-7 ab^3+4 b^4),\\
 F&=-8 \vartheta_2^{24}+24 \vartheta _3^4 \vartheta_2^{20}+249 \vartheta _3^8 \vartheta_2^{16}-538 \vartheta _3^{12} \vartheta_2^{12}+249 \vartheta _3^{16} \vartheta_2^8+24 \vartheta _3^{20} \vartheta_2^4-8 \vartheta _3^{24},
\end{aligned}
\end{equation}
where $a$ and $b$ are arbitrary expressions, and $\vartheta_i = \vartheta_i(\tau)$ are Jacobi theta functions given by \eqref{jacob_theta}. By eliminating the $u$ parameter, one can obtain a single complex equation, which corresponds to two real equations that constrain $\tgamma_j$. These equations involve theta functions of $\tau$ as their coefficients. Therefore, the loci where AD-type fibers appear form codimension-2 surfaces in the $\tgamma$ space, which vary with $\tau$. We refer to these surfaces as the AD surfaces, as illustrated in Figure \ref{fig:gamma_space} and numerically plotted in Figure \ref{fig:mmaAD}.

In addition to the codimension-2 AD surfaces, there are codimension-1 walls determined by normal collisions, similar to those in Figure \ref{fig:chambers}. The conditions for these normal collisions are precisely the conditions for the appearance of $I_k$-type or $I_0^*$-type singular fibers, as listed in Table \ref{tab:fiber_sing_classification2}. These conditions define the codimension-1 walls in the $\tgamma$ space, as illustrated in Figure \ref{fig:gamma_space}. 

However, unlike the walls in Figure \ref{fig:chambers}, there is no periodicity for the $\tgamma$ parameters. As a result, the wall-crossings form the Weyl group $W(D_4)$ of the usual $D_4$ algebra, rather than the affine one.

Similar to the $I_0^*$ case, the Picard-Lefschetz transformation occurs when passing through a normal collision wall. For example, starting with a length assignment on the left side of Figure \ref{fig:6I1 cycles}, if we tune $\tgamma_1-\tgamma_2-\tgamma_3-\tgamma_4$ from positive to negative, all the cycles will change according to the PL transformation rule described in \eqref{mPL}. As a result, the length assignment will transform to the configuration shown on the right side of Figure \ref{fig:6I1 cycles}.

In each chamber, there are three components of AD surfaces where $(A_1, A_2)$ AD theories appear. These three surfaces intersect simultaneously at one of the walls of the chamber, such as the $Q_1$ point on the $ADE$ wall, where the $(A_1, A_3)$ AD theory appears. By tuning the $\tau$ parameter appropriately, the intersection points $Q_i$ can situate at the intersection of two walls, leading to the appearance of the $(A_1, D_4)$ AD theory.

\subsubsection*{Monodromies around Argyres-Douglas surfaces}

Finally, we observe an interesting phenomenon: the AD surfaces endow the $\tgamma$ space with non-trivial monodromies. When the parameter point lies on an AD surface, two $I_1$ fibers of different types collide. If the parameter point is slightly displaced from the AD surface and goes around a loop linking this AD surface, it is observed that the two $I_1$ fibers will rotate $3\pi$ times around each other, as shown in Figure \ref{fig:monodromy_around_AD_surface}. This results in a continuous exchange of the types and positions of the two fibers.

An immediate consequence is that we can rotate the fibers $p_i$ and $p_k$ in Figure \ref{fig:winding cycles} by going two rounds around an AD surface in the $\tgamma$ space. This allows the cycle $\bfW^\prime$ to be continuously deformed into $\bfW$. In this way, every winding cycle can be deformed into, and thus be homologous to, a straight cycle as shown in Figure \ref{fig:6I1 cycles}.

\newpage

\bibliography{references}
\bibliographystyle{hyperamsalpha}

\end{document}